\documentclass[letterpaper,11pt]{article}

\usepackage[runin]{abstract}
\usepackage{geometry}
\usepackage{titlesec}
\usepackage{graphicx}
\usepackage{amsmath}
\usepackage{amsthm}
\usepackage{amsfonts}
\usepackage{amssymb}
\usepackage{empheq}
\usepackage{algorithmic}
\usepackage{algorithm}
\usepackage[authoryear]{natbib}
\usepackage{url}

\usepackage[font=sf]{caption}

\usepackage{mathtools}

\titleformat*{\section}{\Large\bfseries}
\titleformat*{\subsection}{\large\sc}
\titleformat*{\subsubsection}{\itshape}

\usepackage{epigraph}

\usepackage{etoolbox}

\makeatletter
\patchcmd{\epigraph}{\@epitext{#1}}{\itshape\@epitext{#1}}{}{}
\makeatother

\begin{document}

\title{{\bf On incremental deployability}}

\author{\large{Ioannis Avramopoulos}}

\date{}

\maketitle

\thispagestyle{empty} 

\newtheorem{definition}{Definition}
\newtheorem{proposition}{Proposition}
\newtheorem{theorem}{Theorem}
\newtheorem*{theorem*}{Theorem}
\newtheorem{corollary}{Corollary}
\newtheorem{lemma}{Lemma}
\newtheorem{axiom}{Axiom}
\newtheorem{thesis}{Thesis}

\vspace*{-0.2truecm}

\begin{abstract}
Motivated by the difficulty of effecting fundamental change in the architecture of the Internet, in this paper, we study from a theoretical perspective the question of how individuals can join forces toward collective ventures. To that end, we draw on an elementary concept in Internet systems engineering, namely, that of {\em incremental deployability}, which we study mathematically and computationally. For example, we show that incremental deployability is at least as general a concept as the Nash equilibrium (in that the latter can be derived from the former). We then draw on this foundation to design and analyze institutional mechanisms that are not only promising to bootstrap emerging Internet architectures but they also have broader applications in social organization beyond its predominant market (and finance)-based character.
\end{abstract}

\section{Introduction}

\setlength{\epigraphwidth}{0.28\textwidth}
\epigraph{``God is an operating system.''}{--- \textup{Cat Pierce}}

\subsection{Our question: Incentives as a first principle of social organization}

Our main goal in this paper is to contribute to the development of a mathematical theory of {\em social organization}. To a large extent, organization is about individuals joining forces in a collective venture. The output of a venture can be an institution (such as a state organization) offering a service (such as the services a public sector offers) or a technology, which can in turn be either a material product (such as computer hardware) or an information product (such as a computer algorithm or industrial process). We believe that the individual should be a principal entity in the design and operation of organizational ventures, for example, in the fashion the individual player serves as a foundation of the theory of games laid out by von Neumann and Morgenstern or in the fashion Kenneth Arrow studies in his seminal work on {\em social choice} the problem of aggregating individual social preferences into a preference function for society (albeit with negative results).

\subsubsection{Incentives as a first principle of the organization of the Internet}

We believe that organization theory can benefit from the study of the Internet as a global communications service, a service provided outside the scope of any particular industrial or other organizational entity, and that this study (as we attempt in this paper) can shed light on parochial corporate and state organization (such as, for example, a bureaucratic architecture). But the Internet will also benefit from advances in organization and economic theory: The architecture of the Internet suffers from a pressing stalemate (stagnation) whose origins puzzle systems engineers.

The Internet has enabled significant innovation in society at large, however, it has resisted repeated efforts to effect innovation at its core architecture (known as the TCP/IP architecture) consisting of the protocols running at the network and transport layers of the Internet such as IP (Internet Protocol), TCP (Transmission Control Protocol), and BGP (Border Gateway Protocol). This architecture has facilitated significant innovation at the link layer (e.g., the rapid transition to 3G, 4G, and 5G systems in wireless and the rapid transition to optical technologies for fixed access) and the application layer (e.g., Google, Facebook, Twitter), however, it is increasingly being held responsible for stifling the emergence of new applications. In this paper, we aim to both {\em explain} why innovation is failing and to {\em propose a course of action} on how innovation can be effected.

Technology is something engineers design and, therefore, we would reasonably expect to know precisely how our own technological artifacts work. That may be true for some technologies (such as computer processors) but it is certainly not true for all of them. To a large degree the Internet belongs to the class of those technologies of which we are trying to understand the principles of their operation in the same manner that we are trying to understand the principles of operation of natural, biological, or social phenomena. This paper was motivated by an effort to rigorously understand the nature of the Internet (in particular, its technical architecture), how it evolves, whether we should intervene in its evolution, and how we might attempt such an intervention. Our inquiry on the principles at work in the organization of the Internet's architecture is mathematical. Our analysis benefits from mathematical work in economics and can further inform economic science.

In this paper, we stipulate that the success of a technological venture in a social system (for example, whether an emerging technology or architecture will be adopted in the Internet) is a matter that transcends technical challenges and critically relies on the {\em incentives} that shape the adoption environment (as measured by the utilities actions and decisions confer to {\em individual} members of the social system, rather than, for example, on aggregate social utility metrics). But what, if any, incentives theory can facilitate our understanding of the Internet's architectural stalemate and how could that theory illuminate social and technological organization more broadly? To the extent of our knowledge, there are two principal incentives theories in economic science, namely, that concerning the principal-agent model and game theory. Let us examine both theories in turn.

\subsubsection{A theory of incentives based on the principal-agent model}

The dominant at present institution facilitating social organization is the {\em market}. Drawing on the dominance of markets and, in particular, their monetary orientation, every venture that requires coordination among different principals rests on a {\em financial structure} and incentives are typically, if not always, being understood as being tantamount to {\em financial incentives}.\footnote{We should note that there are economists who question the parity between incentives to perform in organizational tasks and financial incentives. For example, economist Dan Ariely in his books {\em Dollars and Sense} and {\em Predictably Irrational} (co-authored with Jeff Kreisler) challenges that financial decisions are rational in fundamental ways.} The elementary primitives atop which financial structures are built are {\em capital} and {\em contracts}. Together these structures are used to implement {\em delegation,} which we may define as the act of assigning executive responsibility for the implementation of a task to one individual or organization on the principles, instruction, or specification of another individual or organization. Delegation is ubiquitous in bureaucratic organization forming the architectural basis of the majority of public sectors but also of corporate entities in the private sector (which are organized around a notion of delegation hierarchy). 

The model of organization of the Internet departs significantly from a bureaucratic hierarchy, and we found that the branch of economics studying delegation, namely, the {\em principal-agent theory} (for example, see \citep{Laffont}), does not shed light on our pursuit to understand the architectural stalemate or to effect change: In the Internet, it is hard (if it is all meaningful) to draw a distinction between agents and principals among the entities whose choices determine the Internet's organizational structure (architecture). Our study of architectural evolution and innovation thus focuses elsewhere but, nevertheless, we should point out that the prominent in principal-agent theory concept of a {\em commitment} will be central in the mechanisms we propose to counteract the Internet's stalemate (albeit our use of commitments is of a different nature from that in principal-agent theory in that favorable outcomes can be induced by the mere possibility of entering into commitment without any individual necessarily committing to any specific strategy).

\subsubsection{A theory of incentives based on mathematical games}

The Internet is a global communications system whose architecture is derived from technological choice in a social system of equipment manufacturers, service providers, content providers, and endusers. The members of this social system are free to adopt whichever technologies they choose: To a first approximation, the Internet is an ``anarchy'' in that there is no authority that shapes the Internet's architecture. Furthermore, decisions made in this system are coupled: The benefit of a decision to adopt any given technology depends not only on an agent's own choice but also on the choices of the other agents. Choice is, therefore, constrained by the {\em incentive structure} of the social and technical environment and it is natural to model this structure after a {\em mathematical game}. 

This is a thought experiment that is not unique to our effort in this paper: Christos Papadimitriou often quotes Scott Shenker as having said that ``The Internet is an equilibrium, we just have to identify the game'' \citep{Algorithmic-Game-Theory}. Papadimitriou is a theoretical computer scientist who could in many ways be deemed as the father of {\em algorithmic game theory,} the branch of theoretical computer science concerned with the computational foundations of game theory. Understanding the Internet as an equilibrium in a mathematical game is a thought experiment bound to stretch one's imagination granted theoretical computer science is the starting point. The centerpiece of game theory, and in many ways of economic thought at large, is the {\em Nash equilibrium}. Unfortunately, the complexity of computing a Nash equilibrium is an open question in mathematics. That a theory of incentives explaining the manifestation of the Internet (and informing organization theory at large) should be based on the Nash equilibrium as its elementary foundation is naturally questionable.

\subsubsection{A theory of incentives drawing on Internet systems engineering}

Scott Shenker is a computer scientist whose main contributions are in Internet systems research. One of the fundamental questions in this discipline concerns how to improve the technical performance of networked systems. There are two primary approaches to that end that could be appositely called the {\em clean-slate approach} and the {\em evolutionary research approach} \citep{Clean-Slate}. In the first approach, systems research pursues designs that are free from the constraints of the incumbent environment and, in the second approach, systems research uses the incumbent environment as the starting point trying to identify incremental changes that improve the status quo. The improvements in the second approach are typically not as substantial.

But the technologies identified in the second approach are typically {\em incrementally deployable} (in that early adopters benefit) unlike the disruptive changes that are typically required to implement and deploy clean-slate architectural designs. Much of the effort of the networking community toward effecting architectural change in the Internet has been focused on rendering clean slate designs incrementally deployable (for example, through deployment in {\em testbed environments} where technologies can also gain not only experimental but also real deployment traction). But the success of these efforts has been limited: Rendering fundamental architectural changes incrementally deployable using software and hardware mechanisms alone faces inherent barriers to success.

Nevertheless, incremental deployability is thus of paramount important in Internet systems engineering. This raises the question if a theory of incentives explaining the manifestation of the Internet (and informing organization theory at large) can be based on the concept of incremental deployability. In this vein, to say, for example, that the Internet is an equilibrium seems to capture the basic fact that emerging technologies are not incrementally deployable against the incumbent architecture of the Internet. Can we make this intuition precise? It turns out the answer is positive but to answer these questions we must be precise on the definition of incremental deployability.

\subsection{Our contributions and techniques}

\subsubsection{On the mathematical foundations of incremental deployability}

Consider a model of evolution in a social system where change is driven by incremental deployability. The most natural question to ask in this system is if there are social states such that no other state is incrementally deployable against them. Such states would be stable in this model of evolution characterizing equilibrium behavior in the system. This thought experiment can be brought to bear to obtain a characterization of {\em Nash equilibrium} (and more generally of the concept of equilibrium as a solution to a variational inequality) and of {\em sink equilibrium} (an equilibrium concept in discrete evolution spaces). One of the main contributions of this paper is to derive these characterizations.

Our characterizations are based on a mathematical concept we call a {\em polyorder,} which a binary relation over the {\em evolution space} (for example, the space of pure or mixed strategy profiles of a strategic game) relating pairs of states according to whether either is incrementally deployable against the other. In our characterization, Nash and sink equilibria emerge as {\em maximal} or {\em undominated} elements of such a polyorder. This is in sharp contrast to the standard characterization of equilibria as a {\em fixed points} of a corresponding map in economic theory (for example, Nash equilibria are typically understood as fixed points of a corresponding best response correspondence).

We note that polyorders are in general neither transitive nor complete in sharp contrast to the vast majority of preference relations studied in {\em order theory} as transitivity and completeness are basic postulates of the definition of rationality (see \cite{MWG}). Since maximal elements of polyorders characterize equilibria (which economic theory interprets as manifestations of collective behavior) we take the liberty of interpreting this property as the common-sense ascertainment that groups of (in principle) rational individuals are not necessarily (group-wise) rational.

We further use our formalism to derive new (equilibrium) solution concepts in mathematical games. One of these concepts is the {\em drifting equilibrium} that has analogues in both continuous and discrete evolution spaces. In symmetric $2$-player games (that is, $2$-player games where the payoff matrix of each player is the transpose of that of the other), we use the drifting equilibrium to characterize {\em neutrally stable strategies} (a solution concept from evolutionary game theory \citep{Evolution}). But we show that drifting equilibria are more general than neutral stability in a strong sense. The discrete sink equilibrium analogue (which may be understood as an analogue of neutral stability in discrete evolution spaces) finds direct application in the analysis of the mechanisms we design to effect fundamental change in the Internet architecture (and that, more generally, induce incrementally deployable equilibrium transitions in coordination games). 

Sink equilibria are typically, if not exclusively, defined in the literature based on dynamics characterized by better or profitable best responses: They correspond to collections of strategy profiles that evolution based on unilateral better or profitable best responses cannot escape. Such dynamics give rise to the pure Nash equilibrium as the equilibrium solution concept of sink equilibria. Our drifting sink equilibrium solution concept is characterized by a notion of dynamics wherein evolution is allowed to {\em drift} according to unilateral deviations that are not harmful. We call the equilibrium solution concept of a drifting sink equilibrium a {\em strongly maximal equilibrium}. Such equilibria are necessarily pure Nash equilibria, but the converse is generally false. A distinguishing feature is that strongly maximal equilibria can manifest in equilibrium classes or {\em clusters}. We argue such an equilibrium concept, is as, if not more, plausible, than the pure Nash equilibrium. 

On the way we answer an open question in game theory framed in \cite{Peleg}.

\subsubsection{On the computational foundations of incremental deployability}

One of the main contributions of this paper is showing that a {\em globally evolutionarily stable strategy} (abbreviated as GESS) (which is also a {\em globally incrementally deployable strategy} according to a definition of incremental deployability giving rise to the Nash equilibrium as an undominated element of a corresponding polyorder) is attractive under a {\em multiplicative weights dynamic} provided a parameter of that dynamic, called the {\em learning rate,} diminishes according to standard set of scheduling conditions used to prove exact convergence to a global optimum in convex optimization. 

Our multiplicative weights dynamic is known as Hedge \citep{FreundSchapire1, FreundSchapire2}. A well-known result in theoretical machine learning and equilibrium theory is that by iteratively applying Hedge we can efficiently approximate the value of a zero-sum game to any arbitrary accuracy (and that the average of iterates converges to an approximate minimax strategy). In this paper, we generalize that analysis to settings of evolutionary dominance far more general than convex optimization. That proof techniques from convex optimization apply to our problem is not straightforward to observe: Our starting point is a lemma showing the composition of the {\em relative entropy function} with Hedge to be a convex function of the learning rate. This lemma enables us to use relative entropy as a potential function in fashion similar to \cite{FreundSchapire2}, (however, we lift the assumption of a monotone payoff function Freund and Schapire assume).

Our equilibrium computation problem is framed according to an equilibrium concept (the GESS and a relaxation thereof as defined in the sequel) formulated in the setting of {\em evolutionary game theory} \citep{Weibull, PopulationGames}. The {\em LaSalle invariance principle} using relative entropy as a Lyapunov function implies that a GESS is interior globally asymptotically stable under the {\em replicator dynamic} of \cite{TaylorJonker}. Although we could not trace this result in the literature we refer to \citep{Weibull, SelfishRoutingEvolution, StableGames, PopulationGames, Piliouras} for similar results in this vein. Obtaining analogous results for Hedge (that generalizes the replicator dynamic in discrete time) motivated some of our work. Lyapunov theory proved resistant as a tool to obtain asymptotic convergence of Hedge to a GESS. 

Besides solving zero-sum games (which are equivalent to linear programming), Hedge and other {\em regret-minimizing algorithms} have been applied to equilibrium computation problems in {\em nonatomic selfish routing games}. These are population games \citep{PopulationGames} equipped with a convex potential function \citep{Beckmann-McGuire-Winsten}. Therefore, this setting is actually a convex optimization setting (and, thus, more restrictive than ours). The earliest work analyzing the effect of regret minimizing behavior in this routing setting is by \cite{RoutingRegret, RoutingRegret2}. \cite{Piliouras} analyze, in the same setting, the behavior of the replicator dynamic as obtained from a stochastic version of Hedge in the limit of the learning rate tending to zero and show asymptotic performance that is more favorable than generic regret-minimizing algorithms. \cite{Mertikopoulos} consider stochastic (bandit) versions of Hedge in potential $N$-player games \citep{Potential} and show convergence to an (approximate) Nash equilibrium (for a general potential function).

From a complexity theory perspective, computing an equilibrium in bimatrix games is {\bf PPAD}-complete \citep{Daskalakis, CDT} and a fully polynomial time approximation of a Nash equilibrium in this class of games implies {\bf P = PPAD}. The symmetrization of \cite{Jurg} implies the {\bf PPAD}-hardness of (exact) equilibrium computation is inherited in symmetric bimatrix games (that is, bimatrix games where the payoff matrix of each player is the transpose of that of the other). Drawing on this result, we show in this paper that a fully polynomial time approximation scheme for a symmetric equilibrium in symmetric games implies {\bf P = PPAD}.

The best known bound for computing an equilibrium in bimatrix games is {\em quasipolynomial} \citep{LMM}. The algorithm performs brute force search on a grid in the space of strategies. The best known polynomial time algorithm for approximating a Nash equilibrium in bimatrix games achieves a 0.3393 approximation \citep{Tsaknakis-Spirakis-journal}. Our empirical results on symmetric equilibrium computation in symmetric bimatrix games (using iterative applications of Hedge as the algorthm) suggest {\bf PPAD} may not be a hard complexity class: In randomly generated symmetric bimatrix games, we find Hedge always approximates a symmetric equilibrium strategy.

Deciding if an evolutionarily stable strategy (ESS), the local version of a GESS, exists in a symmetric bimatrix game is a $\Sigma^2_p$-complete problem \citep{Conitzer-ESS} and recognizing a strategy as such is a {\bf coNP}-complete problem \citep{Etessami, Nisan-ESS} even for a symmetric payoff matrix. Recognizing a GESS is also {\bf coNP}-complete (even in doubly symmetric games, where the payoff matrices are symmetric). Note also that the ESS and its weak analogue, the neutrally stable strategy (NSS) generalize the well-known in optimization theory notions of {\em strict local optimizer} and {\em (weak) local optimizer} respectively. Furthermore, if a symmetric bimatrix game is equipped with a potential function, these notions coincide. Similarly, verifying a point to be local optimizer is {\bf coNP}-hard even in quadratic programs \citep{Pardalos2}.

\subsubsection{On equilibrium computation in the architecture of the Internet}

There are protocols at the core architecture of the Internet whose computational properties are not sound. The notable example is BGP, the routing protocol that builds network paths across administrative domains in the Internet. As we discuss in more detail in the sequel, BGP routing computes a {\em sink equilibrium} of a related mathematical game using an algorithm that is not guaranteed to converge in a small number of steps, and harmful protocol oscillations have been known to occasionally emerge in practice. Although the detailed study of the technical properties of routing protocols (such as BGP) falls outside the scope of this paper, we note that routing systems whose equilibrium outcomes are (essentially) unique (and at the same time globally incrementally deployable) are not unthinkable. For example, it is easy to show that a {\em Wardrop equilibrium} (see \citep{HowBadisSelfishRouting, Roughgarden}) is, under mild assumptions, a GESS.

But there are also protocols at the core architecture of the Internet whose mathematical foundation is puzzling. The notable example is TCP/IP congestion control, the part of the architecture that ensures communication links are shared fairly among competing source-destination pairs who try to exchange traffic. Whether TCP/IP congestion control is stable from a game-theoretic perspective has been an open question at the intersection of theoretical computer science and Internet systems engineering \citep{Internet-Game-Theory}. In this paper, we derive a plausible model that explains this congestion control algorithm as computing a Nash equilibrium (in a {\em supergame}). To a large extent, TCP/IP congestion control lies as the core of a global scale {\em institution} that has been formed around the sharing of link bandwidth as a {\em common pool resource} \citep{Hardin, Ostrom}.

The technical and institutional nature of the Internet infrastructure evolved in a grassroots manner on the fertile soil of an ingenious technical design, but other than the designers making elementary provisions in the routing system for the commercialization of the infrastructure, this nature was by no means designed. Since evolution does not generally lead in itself to the most favorable of outcomes, the case of designing technology and institutions in tandem presents itself naturally, and gives rise to the following fundamental dilemma: {\em What should the division of responsibility between code and institutions be in the future Internet?} In trying to answer this dilemma we propose and analyze an approach for effecting architectural innovation in the Internet that promotes institutional design as a first order responsibility of network architects and engineers.

\subsubsection{On incrementally deployable equilibrium transitions}

We have insofar argued that a theory of incentives with incremental deployability at its foundation is a reasonable thought experiment. Our concern then is to use this theory in the organization of collective ventures and, in particular, on how to escape the aforementioned stagnation in the Internet's architecture. Let us state the problem motivating our inquiry in more precise terms. 

The Internet as a global communications system is characterized by {\em positive network externalities effects:} The utility of adopting a technology to any given adopter increases with the number of adopters. However, it is typically the case that if the number of adopters is small, utility is negative in that the adoption effort exceeds the benefit. Therefore, Internet technologies require a {\em critical mass} of adoption to gain deployment traction (and thrive) through evolutionary forces alone. 

Such network effects are ubiquitous in the Internet whether in the applications built on top of the Internet or in the protocols giving rise to the TCP/IP architecture itself. For example, as social networking platforms (such as Twitter or Facebook) form, early adopters are likely to disappear without a sustainable adoption rate, however, as adopters increase the platform becomes stable. The growth of the TCP/IP architecture into a global platform faced a similar evolutionary pattern in its competition against other technologies. However, the ``architectural stability'' of TCP/IP has side effects preventing the community from retrofitting new features in this architecture.

Consider, for example, Internet routing: That BGP (the routing protocol providing global connectivity) is vulnerable to accidental misconfiguration and malicious attack is not a particularly novel observation, but that Internet routing remains vulnerable despite significant efforts by the research community, governmental organizations, and service providers to effect change in BGP is arguably worrisome. The pattern that emerges in the efforts to deploy a secure version of BGP such as Secure-BGP \citep{Kent} involves the aforementioned network effects: Unilateral Secure-BGP deployment is costly to the corresponding organization (typically an ``autonomous system'') without providing security benefits, however, as the autonomous systems that deploy Secure-BGP increase, security benefits quickly ramp up to outrun those of BGP \citep{ASR}.

Motivated by this general pattern, the question we are thus concerned with is how to ``jump-start'' the adoption of technologies whose success depends on positive externalities. We frame this question using the game-theoretic model of the {\em stag hunt}. Various social coordination problems have been modeled after this game \citep{PlayingFair,SocialContract,Skyrms,Medina} and, more generally, coordination games have also been used in the study of technological competition in networks \citep{Economides1}. The application span of our analysis is, therefore, rather broad.

The stag hunt is an archetypical in coordination theory game in which there are two pure strategies available to each player: One is to {\em cooperate} in choosing a superior strategy whose benefit, however, critically depends on whether other players choose alike or instead to {\em defect} in choosing a status quo strategy that reaps moderate benefits independently of what other players do. Furthermore, selecting the superior strategy incurs a loss if other players do not cooperate. 

The equilibrium structure is simple: Every stag hunt has three Nash equilibria, namely, one corresponding to universally choosing the superior strategy, another pure equilibrium corresponding to universally defecting (and choosing the inferior strategy), and a mixed equilibrium of low predictive value. Predicting the outcome of a stag hunt is a notoriously hard {\em equilibrium selection problem} (cf. \cite{HS}). In this paper, we are interested in enticing selection of the good equilibrium, without being vulnerable to the indefinite predictions of equilibrium selection theory, through an intervention in the incentive structure shaping the adoption environment. 

Our approach to jump-staring coordination in a stag hunt is to assist the adoption of the superior emerging technology by means of {\em institutions} that change the incentive structure of the adoption environment so that the superior outcome is globally incrementally deployable (even against the incumbent equilibrium outcome). The term {\em institution} is to be understood precisely as a {\em social mechanism} aiming to coordinate activity among the members of a population \citep{Schotter}. The social coordination mechanisms we propose, which we refer to simply as {\em coordination mechanisms,} have an ``opt-in character'' in that although there are rules players should abide by should they voluntarily choose to use these mechanisms, they are not obliged to do so otherwise. 

Our use of the term {\em mechanism} suggests keen relevance of our approach to {\em mechanism design} \citep{Nisan}, and although there are many similarities, there are also important differences the most important of which is that we are not attempting to entice a truthful revelation of preferences among outcomes (as is the typical situation in the design of mechanisms and their theoretical analysis), but rather to incur an {\em equilibrium transition} from an inferior to a superior outcome, assuming that the incentive structure (and, therefore, the respective preferences) is well-known.

Our contributions in vein are two-fold: Our first contribution is to propose two coordination mechanisms both of which rely on the capability to enforce voluntary player commitments. 

The first coordination mechanism relies on the presence of an {\em insurance carrier} having the capability to insure costly efforts to adopt the emerging technology. Players have the option to purchase insurance provided they commit to adopt. Our analysis indicates the surprising finding that players may incur the investment cost of deploying the emerging technology while eschewing purchasing insurance as the mere presence of the carrier curtails the adoption risk. 

The second mechanism is based on an {\em election} among the population of adopters. Players have the option to vote in favor of adopting the emerging technology on the condition that a unanimous vote amounts to an adoption commitment. The analysis of this mechanism is more involved than that of the insurance-based counterpart, but we show adoption is the only plausible outcome once this mechanism is present (allowing for the possibility players may eschew participation in voting).

Our second contribution is the analysis of the aforementioned social coordination mechanisms using Nash equilibrium theory, weak and strict iterated dominance, as well as the aforementioned solution concept of the strongly maximal equilibrium (derived from a mathematical model of incremental deployability in discrete evolution spaces). To that end, we note that theorists have been concerned with games wherein better and best response dynamics necessarily give rise to pure Nash equilibria as the outcome of long-term asymptotic behavior. Such games are known as {\em weakly acyclic games} \citep{Peyton2} in the sense that, although cyclic behavior can manifest in their dynamics, such behavior is necessarily transient. In the interest of analyzing the election mechanism, we introduce a class of games that we call {\em weakly ordinally acyclic} capturing the aforementioned strongly maximal equilibria as the outcome of long-term asymptotic dynamics: We show that the election mechanism applied to the stag hunt induces a game that falls into this class. This implies that the cooperative social outcomes is globally incrementally deployable starting from any social state. In the course of our analysis, we prove that the stag hunt is a {\em weakly acyclic game}.

Our methods to solving social coordination problems using an evolutionary (incrementally deployable) approach do not involve devising {\em code} (software or hardware) but rather {\em social institutions} whose efficacy we are able to mathematically prove. From a practical perspective, the value of devising such proofs is in the formulation of testable predictions on the behavior of the Internet population. That is, we cannot claim that once the institutions we propose are instantiated adoption will take place, however, if it doesn't, then it is one or more of our assumptions that are false (such as, for example, the rationality of the players) rather than the logic of our argument.

\subsection{Overview of the rest of this paper}

Section \ref{game_theory_preliminaries} introduces background concepts in game theory, corresponding to a main analytical framework for our work in this paper. The mathematical formalization of incremental deployability and the corresponding mathematical foundation relating it to equilibrium in mathematical games and elsewhere is given in Section \ref{mathematical_foundations} whereas in Section \ref{computational_foundations} we present our algorithmic results on incremental deployability. In Section \ref{BGP_section}, we discuss equilibrium computation in Internet routing protocols whereas, in Section \ref{TCP_section}, we model TCP/IP congestion control as an equilibrium computation algorithm. Section \ref{equilibrium_transitions} studies incrementally deployable transitions from an inferior to a superior equilibrium in coordination games, in particular, the stag hunt. Section \ref{other_related_work} discusses related approaches to overcoming deployability barriers. Finally, in Section \ref{conclusion} we conclude this paper.

\section{Game theory preliminaries}
\label{game_theory_preliminaries}

\subsection{Games in strategic form and pure Nash equilibria}

Let us begin with the definition of games in (finite) {\em strategic form} (also called a {\em normal form}). To define a game in this form, we need to specify the set of players, the set of strategies available to each player, and a utility function for each player defined over all possible combinations of strategies that determines a player's payoffs. Formally, a  strategic-form game $\Gamma$ is a triple $$\Gamma = (I, (S_i)_{i \in I}, (u_i)_{i \in I}),$$ where $I$ is the set of players, $S_i$ is the set of pure strategies available to player $i$, and $u_i: S \rightarrow \mathbb{R}$ is the utility function of player $i$ where $S = \times_i S_i$ is the set of all strategy profiles (combinations of strategies). Let $n$ be the number of players. We often wish to vary the strategy of a single player while holding other players' strategies fixed. To that end, we let $s_{-i} \in S_{-i}$ denote a strategy selection for all players but $i$, and write $(s'_i, s_{-i})$ for the profile
\begin{align*}
(s_1,\ldots,s_{i-1}, s'_i,s_{i+1},\ldots, s_n).
\end{align*}
A pure strategy profile $\sigma^*$ is a {\em Nash equilibrium} if, for all players $i$,
\begin{align*}
u_i(\sigma_i^*, \sigma_{-i}^*) \geq u_i(s_i, \sigma_{-i}^*) \text{ for all } s_i \in S_i.
\end{align*}
That is, a Nash equilibrium is a strategy profile such that no player can obtain a larger payoff using a different strategy while the other players' strategies remain fixed. If unilateral deviations from a Nash equilibrium are necessarily harmful, the Nash equilibrium is called {\em strict,} and it is called {\em weak} otherwise. The previous definitions concern Nash equilibria in pure strategies, however, we note that Nash equilibria admit an analogous definition assuming players may use {\em mixed strategies,} that is, probability distributions over their respective pure strategy sets. Nash equilibria in mixed strategies are known to always exist in strategic form games, in contrast to pure Nash equilibria.

\subsection{Bimatrix games and mixed Nash equilibria}

A $2$-player (bimatrix) game in normal form is specified by a pair of $n \times m$ matrices $A$ and $B$, the former corresponding to the {\em row player} and the latter to the {\em column player}. A {\em mixed strategy} for the row player is a probability vector $P \in \mathbb{R}^n$ and a mixed strategy for the column player is a probability vector $Q \in \mathbb{R}^m$. The {\em payoff} to the row player of $P$ against $Q$ is $P \cdot A Q$ and that to the column player is $P \cdot B Q$. Let us denote the space of probability vectors for the row player by $\mathbb{P}$ and the corresponding space for the column player by $\mathbb{Q}$. A Nash equilibrium of a $2$-player game $(A, B)$ is a pair of mixed strategies $P^*$ and $Q^*$ such that all unilateral deviations from these strategies are not profitable, that is, for all $P \in \mathbb{P}$ and $Q \in \mathbb{Q}$, we simultaneously have that
\begin{align}
P^* \cdot AQ^* &\geq P \cdot AQ^*\label{eqone}\\
P^* \cdot BQ^* &\geq P^* \cdot BQ.\label{eqtwo}
\end{align}
We denote the set of Nash equilibria of $(A, B)$ by $NE(A, B)$.

\subsection{Symmetric bimatrix games and nonlinear payoff function}

If $B = A^T$, where $A^T$ is the transpose, the game is called {\em symmetric}. Let $(C, C^T)$ be a symmetric bimatrix game. We call $(P^*, Q^*) \in NE(C, C^T)$  symmetric if $P^* = Q^*$. If $(X^*, X^*)$ is a symmetric equilibrium, we call $X^*$ a symmetric equilibrium strategy. Nash showed that every symmetric $N$-person game (and, therefore, every symmetric bimatrix game) has a symmetric equilibrium. 

Denoting by $\mathbb{X}(C)$ the probability simplex corresponding to the $n \times n$ matrix $C$, the Nash equilibrium conditions \eqref{eqone} and \eqref{eqtwo} simplify as follows for a symmetric equilibrium strategy $X^*$:
\begin{align*}
\forall X \in \mathbb{X}(C) : (X^* - X) \cdot CX^* \geq 0.
\end{align*}
If the previous inequality is an equality for all $X \in \mathbb{X}(C)$, we call $X^*$ an {\em equalizer}.

In symmetric bimatrix games, the payoff vector $CX$ of a strategy $X$ (a $n \times 1$ vector) is given by multiplying $X$ with the $n \times n$ payoff matrix $C$. Our analysis in the sequel applies, in part, to a more general setting, where the payoff vector $C(X) \equiv CX$ at $X$ is given by a (nonlinear) continuous operator $C$. We denote the space of continuous operators over the probability simplex by $\mathbb{C}$. We denote the space of continuous operators with payoffs in the range $[0, 1]$ by $\mathbb{\hat{C}}$. 

Let us fix some further notation at this point: Given $C \in \mathbb{C}$, we denote its set of pure strategies by $\mathcal{K}(C) = \{1, \ldots, n\}$. Pure strategies are denoted either as $i$ or as $E_i$, a probability vector whose mass is concentrated in position $i$. We denote the probability simplex in $\mathbb{R}^n$, where $n$ is the number of pure strategies, by $\mathbb{X}(C)$. We denote the set of Nash equilibria of operator $C$ by $NE^+(C)$. If $C$ is linear this set corresponds to the set of {\em symmetric} Nash equilibrium strategies of $(C, C^T)$. Define
\begin{align*}
\mathcal{C}(X) \doteq \left\{ i \in \{1, \ldots, n \} | X(i) > 0 \right\}.
\end{align*}
$\mathcal{C}(X)$ is the {\em carrier} or {\em support} of $X$. $X \in \mathbb{X}(C)$ is {\em interior} if $\mathcal{C}(X) = \{ 1, \ldots, n \}$. We denote the relative interior of $\mathbb{X}(C)$ (that is, the set of all interior strategies of $\mathbb{X}(C)$) by $\mathbb{\mathring{X}}(C)$. The {\em carrier game} of strategy $X$ is the subgame of $C$ whose pure strategies correspond to $\mathcal{C}(X)$. Note that the previous definitions extend in a straightforward fashion to {\em population games} \citep{PopulationGames} where the domain (space of social states) is, more generally, a Cartesian product of simplices.

\subsection{Evolutionary and neutral stability}

Let us now define the further equilibrium concepts used in this paper: A refinement of symmetric Nash equilibrium is the notion of an {\em evolutionarily stable strategy (ESS)} originally defined in the setting of evolutionary selection \citep{TheLogicOfAnimalConflict, Evolution}. We adopt a different equivalent definition, due to \cite{HSS} (see also \cite[p. 45]{Weibull}), more convenient for our purposes. Our notation is indifferent to whether payoffs are linear or nonlinear (see \cite[Ch. 8.3]{PopulationGames} for a comprehensive treatment of nonlinear payoff functions and proofs of equivalence of alternative definitions of evolutionary stability in this general case).

\begin{definition}
\label{evolution_definitions}
Let $C \in \mathbb{C}$. We say $X^* \in \mathbb{X}(C)$ is an ESS of $C$ if 
\begin{align*}
\exists O \subseteq \mathbb{X} \mbox{ } \forall X \in O/\{X^*\} : (X^* - X) \cdot CX > 0.
\end{align*}
$O$ is a neighborhood of $X^*$ called a {\em superiority neighborhood}. $X^*$ is called a global ESS (GESS) if $O$ coincides with $\mathbb{X}(C)$. If the inequality is weak we obtain {\em neutral stability} (an NSS or GNSS).
\end{definition}

An intuitive interpretation of this definition can be obtained from the case where the operator $C$ corresponds to the gradient of a scalar potential function $f : \mathbb{X}(C) \rightarrow \mathbb{R}$. Then $(X^* - X) \cdot CX \equiv (X^* - X) \cdot \nabla f(X)$ is the directional derivative of $f$ at $X$ in the direction from $X$ to $X^*$. If $f$ is a strictly concave function, then $X^*$, the unique global optimizer of $f$, is a GESS and this can be shown using standard inequalities in convex analysis. Note that, in general, every ESS is an NSS and every NSS is an equilibrium strategy (the standard proof is contained in \citep{PopulationGames}).

A word of note is that in the setting of single-population games, the acronym ESS refers, as mentioned earlier, to an {\em evolutionarily stable strategy}. The term strategy is used in the sense of a mixed strategy (that is, a probability distribution over the set of pure strategies). In the multi-population setting, however, the acronym ESS means {\em evolutionarily stable state,} which corresponds to a combination of (mixed) strategies, one for each population. Keeping this in mind, no ambiguity results by referring to both situations with the term ESS (as is common in the literature).

The following proposition characterizes evolutionary and neutral stability (corresponding to how these notions were originally defined \citep{TheLogicOfAnimalConflict, Evolution}).

\begin{proposition}
\label{nss_1}
$X^*$ is an NSS of $C \in \mathbb{C}$ if and only if the following conditions hold simultaneously
\begin{align*}
X^* \cdot C X^* &\geq X \cdot C X^*, \mbox{ } \forall X \in \mathbb{X}(C), \mbox{ and }\\
X^* \cdot C X^* &= X \cdot C X^* \Rightarrow X^* \cdot C X \geq X \cdot C X, \mbox{ } \forall X \in \mathbb{X}(C)\mbox{ such that } X \neq X^*.
\end{align*}
If the inequalities are strict we obtain a characterization of an ESS.
\end{proposition}

A GESS satisfies the following intuitive property:

\begin{lemma}
\label{gess_unique_eq_s_lemma}
If $X^*$ is a GESS of $C \in \mathbb{C}$, $NE^+(C) = \{ X^* \}$. 
\end{lemma}

\begin{proof}
Suppose $Y \in NE^+(C)$ such that $Y \neq X^*$. Then
\begin{align*}
\forall X \in \mathbb{X}(C) : (Y - X) \cdot CY \geq 0.
\end{align*}
However, letting $X = X^*$, contradicts the definition of a GESS.
\end{proof}

That is, the previous lemma shows that if $X^*$ is a GESS, it is the unique symmetric equilibrium strategy. Relaxing this requirement, we obtain an {\em evolutionarily dominant set (EDset)}.

\begin{definition}
We say that the set of symmetric equilibrium strategies $NE^+(C)$ is an {\em EDset} if
\begin{align*}
\exists X^* \in NE^+(C) \mbox{ } \forall X \not\in NE^+(C) : (X^* - X) \cdot CX > 0.
\end{align*}
We call $X^*$ an {\em evolutionarily dominant strategy (EDS)}.
\end{definition}

In an EDset, an EDS is strictly superior to every strategy other than symmetric equilibrium strategies in which case it is weakly superior. Let us consider cases where this definition can manifest: Considering again the relationship to convex optimization, if $C$ corresponds to the gradient field of a concave potential function $f$, then the strategies in the convex set of global optimizers of $f$ form an EDset and every such strategy is an EDS. Thus EDset's emerge quite frequently in practical applications. Considering further examples, a generalization of an ESS is an {\em evolutionarily stable set (ESset)} \citep[p. 51-53]{Weibull}: We say that $\mathbb{X}^* \subset \mathbb{X}(C)$ is an ESset of $C \in \mathbb{C}$ if 
\begin{align*}
\forall X^* \in \mathbb{X}^* \mbox{ } \exists O \subseteq \mathbb{X} \mbox{ } \forall X \in O/\{X^*\} : (X^* - X) \cdot CX \geq 0,
\end{align*}
where the inequality is strict unless $X \in \mathbb{X}^*$. $O$ is a neighborhood of $X^*$. An ESset with a globally superior (GNSS) strategy is an EDset. If $C$ is not a gradient field of a concave function, then the EDset may not (in general) be convex.  (Note in passing a non-convex EDset can manifest even in the setting of ``doubly symmetric'' bimatrix games where payoff matrix is also symmetric.)

\section{Mathematical foundations of incremental deployability}
\label{mathematical_foundations}

In this section, we show that {\em incremental deployability} admits mathematical definitions that correspond to intuition on how the term is used in Internet systems engineering but acquire their own independent meaning in the abstract setting (of {\em vector fields} or {\em graphs}) where these definitions are formulated, and then use these definitions to characterize equilibrium concepts such as the {\em variational equilibrium} (for example, see \citep{VI}), generalizing the {\em Nash equilibrium} \citep{Nash, Nash2} and other equilibrium concepts in economics, {\em neutrally stable strategies} \citep{TheLogicOfAnimalConflict, Evolution}, and the {\em sink equilibrium} \citep{Goemans}.

\subsection{Order theory preliminaries}

In our formulation, equilibria emerge as {\em undominated elements} of an {\em order relation} relating pairs of elements in an {\em evolution space} (for example, the space of strategies in a normal-form game) according to whether they are ``incrementally deployable'' against each other, a characterization that departs from the standard in the literature understanding of equilibria as {\em fixed points} (for example, a Nash equilibrium is typically understood as a fixed point of the {\em best response correspondence} of a respective game). Let us, therefore, start with elementary background on order theory whose origins rest in {\em decision theory} (initially formalized by \cite{GameTheory}) and it is typically used to model the decision making behavior of individual players in strategic environments (see \citep{MWG}). In this section, we use order theory to formalize collective behavior.

\subsubsection{Binary relations}

Let $S$ be a nonempty set. A subset $R$ of $S \times S$ is called a {\em binary relation} on S. If $(s,s') \in R$, we write $sRs'$, and if $(s,s') \not\in R$, we write $\neg sRs'$. $R$ is called {\em reflexive} if $s R s$ for every $s \in S$. $R$ is called {\em complete} if for each $s,s' \in S$ either $sRs'$ or $s'Rs$. $R$ is {\em transitive} if for any $s, s', s'' \in S$ we have that $sRs'$ and $s'Rs''$ imply $sRs''$. $R$ is {\em symmetric} if, for any $s, s' \in S$, we have that $sRs' \Rightarrow s'Rs$ and {\em asymmetric} if, for any $s, s' \in S$, we have that $sRs' \Rightarrow \neg s'Rs$. Let $sPs' \Leftrightarrow sRs' \wedge \neg s'Rs$ and $sIs' \Leftrightarrow sRs' \wedge s'Rs$. Then $P$ and $I$ are also binary relations on S where $P \subseteq R$ and $I \subseteq R$. $P$ is called the {\em asymmetric} (or {\em strict}) part of $R$ and $I$ is called its {\em symmetric} part.

\subsubsection{Equivalence relations}

\begin{definition}
A binary relation $\sim$ on a nonempty set $S$ is called an {\em equivalence relation} if it is reflexive, symmetric, and transitive. For any $s \in S$, the {\em equivalence class} of $s$ relative to $\sim$ is defined as the set
\begin{align*}
[s]_{\sim} = \{ \sigma \in S | s \sim \sigma \}.
\end{align*}
The collection of all equivalence classes relative to $\sim$, denoted as $S/_\sim$, is called the {\em quotient set} of $S$ relative to $\sim$, that is, $S/_{\sim} = \{ [s]_{\sim} | s \in S \}$.
\end{definition}

It is a standard result that for any equivalence relation $\sim$ on a nonempty set $S$, the quotient set $S/_{\sim}$ is a partition of $S$. We draw on this fact as we study discrete evolution spaces.

\subsubsection{Maximal elements and preorders}

A central notion we use from order theory is that of a maximal element:

\begin{definition}
Let $S$ be an nonempty set, $\succeq$ a binary relation on $S$, and let $\succ$ be its asymmetric part. If $s \succ s'$ we say that $s$ {\em dominates} $s'$. An element $s^*$ of $S$ is called $\succeq$-maximal if it is {\em undominated}, that is, if for all $s \in S$, $\neg (s \succ s^*)$.
\end{definition}

Order theory is typically concerned with the study of {\em transitive binary relations}.

\begin{definition}
A binary relation $\succeq$ on a nonempty set $S$ is called a {\em preorder} on $S$ if it is transitive and reflexive.
\end{definition}

We note that in decision theory a binary (preference) relation is called rational if it is a {\em complete preorder} (for example, see \citep{MWG}). But rationality is not central in our development: The binary relation whose maximal elements characterize the mixed Nash equilibrium is not rational. We should note, however, that some authors relax the definition of rationality to correspond to behavior that optimizes some preference relation even if the relation being optimized is not rational \citep{Suzumura}. In this sense, the mixed Nash equilibrium is rational outcome. The following existence result is used in our treatment of discrete evolution spaces.\footnote{An elegant proof can be found in the online manuscript Elements of Order Theory by Efe Ok: \url{https://files.nyu.edu/eo1/public/books.html}}

\begin{proposition}
\label{weoprtiueorituoirtu}
Let $S$ be a nonempty finite set, and let $\succeq$ be a preorder on $S$. Then there exists an element $s^*$ of $S$ such that $s^*$ is $\succeq$-maximal.
\end{proposition}

\subsection{A general definition of equilibrium in continuous evolution spaces}

\subsubsection{The general formalism}

There is a general definition of equilibrium in continuous domains that captures equilibrium notions in various disciplines such as in physics, in optimization theory, in game theory, and in economic theory. This definition is captured by the variational inequality problem \citep{VI}. 

\begin{definition}
\label{var_in}
Given a nonempty, closed, and convex set $\mathbb{X} \subseteq \mathbb{R}^n$ and a continuous vector field $F: \mathbb{R}^n \rightarrow \mathbb{R}^n$, the finite-dimensional {\em variational inequality problem} is to find an element $X^* \in \mathbb{X}$ such that
\begin{align}
(X^* - X) \cdot F(X^*) \geq 0, \forall X \in \mathbb{X}.\notag
\end{align}
We call such an $X^*$ a {\em critical element or point} of $F$. 
\end{definition}

In optimization, the vector field $F$ corresponds to the gradient of a continuously differentiable objective function $f : \mathbb{X} \rightarrow \mathbb{R}$, and the critical points of $f$ are known to coincide with the solutions of the aforementioned variational inequality \citep{ConvexAnalysis}. In game theory, $F$ is derived out the players' payoff matrices. For example, in a symmetric bimatrix game whose payoff matrix is $C$, symmetric Nash equilibrium strategies coincide with the solutions of $(X^* - X) \cdot CX^* \geq 0$.

\subsubsection{Monotone equilibrium problems}

Much as the Nash equilibrium admits (game-theoretic) refinements, the general notion of equilibrium can also be refined within the variational inequalities framework. The finest of those refinements corresponds to the notion of equilibria in {\em monotone vector fields:} A vector field $F: \mathbb{X} \subseteq \mathbb{R}^n \rightarrow \mathbb{R}^n$ where $\mathbb{X}$ is a nonempty, closed, and convex set is called monotone if
\begin{align*}
(X - Y) \cdot (F(Y) - F(X)) \geq 0, \forall X, Y \in \mathbb{X}.
\end{align*}
Monotone equilibrium problems are not a variational inequality formalism per se, but they are ubiquitous as a modeling tool in science and engineering (as they are deemed to be boundary between between what is algorithmically feasible and what is not). Perhaps the most prominent example of a monotone equilibrium problem is {\em convex optimization} \citep{ConvexOptimization}.

\subsubsection{The Minty variational inequality}

One attempt to generalize monotonicity is to seek solutions of the {\em Minty variational inequality}.

\begin{definition}
\label{m_var_in}
Given a nonempty, closed, and convex set $\mathbb{X} \subseteq \mathbb{R}^n$ and a continuous vector field $F: \mathbb{R}^n \rightarrow \mathbb{R}^n$, the finite-dimensional {\em Minty variational inequality problem} is to find an element $X^* \in \mathbb{X}$ such that
\begin{align}
(X^* - X) \cdot F(X) \geq 0, \forall X \in \mathbb{X}.\notag
\end{align}
We call such an $X^*$ a {\em Minty equilibrium} of $F$. 
\end{definition}

It can be shown that a vector field is monotone if and only if its critical elements coincide with its Minty equilibria \citep{Crespi2}. But, in general, Minty equilibria do not coincide with critical elements, which gives rise to a rather broad spectrum of equilibrium problems. In the setting of population games, a solution to the Minty variational inequality is a GNSS (cf. Definiton \ref{evolution_definitions}). Unlike a GESS, the existence of a GNSS does not eliminate the possibility of other equilibria.

\subsection{Characterizations of variational equilibrium}

We call the formalism we use to characterize equilibrium notions in continuous and discrete domains a {\em polyorder}. A {\em continuous polyorder} consists of an {\em evolution space} that is a continuum (typically a convex set), a (typically continuous) {\em vector field} whose domain is this evolution space, and a {\em pairwise ordering} of the elements of the evolution space (that induces a binary (preference) relation between elements) such that the (incentive) structure of the vector field is what determines whether an element is preferred over another. There is a variety of ways to define meaningful polyorders. The polyorders defined in this paper bear a direct (mathematically precise) analogue to the colloquial notion of {\em incremental deployability}. We consider two polyorders in continuous spaces, namely, the {\em strict linear polyorder} and the {\em drifting polyorder}. We consider the strict linear polyorder first:

\subsubsection{Strict linear polyorders}

\begin{definition}
\label{alsdkjfhdskjfho}
Let $F: \mathbb{X} \subseteq \mathbb{R}^n \rightarrow \mathbb{R}^n$ be a continuous vector field where $\mathbb{X}$ is nonempty, closed, and convex. Let $X, Y \in \mathbb{X}$ and let $\mathcal{Y} : [0,1] \rightarrow \mathbb{X}$ be such that $\mathcal{Y}_\epsilon = \epsilon Y + (1-\epsilon) X$. Furthermore, let $\succeq_S$ be a binary relation on $\mathbb{X}$ such that 
\begin{align*}
X \succeq_S Y \Leftrightarrow \forall \epsilon \in [0,1]: X \cdot F(\mathcal{Y}_\epsilon) > Y \cdot F(\mathcal{Y}_\epsilon)
\end{align*} 
and let $\trianglerighteq_S$ be a binary relation on $\mathbb{X}$ such that
\begin{align*}
X \trianglerighteq_S Y \Leftrightarrow \exists \epsilon \in [0,1]: X \cdot F(\mathcal{Y}_\epsilon) \geq Y \cdot F(\mathcal{Y}_\epsilon).
\end{align*} 
We call the structures that $\succeq_S$ and $\trianglerighteq_S$ induce on $\mathbb{X}$ {\em strict linear polyorders} on $F$.
\end{definition}

Let us discuss the definition and to that end let us assume for simplicity that the evolution space is the standard (probability) simplex. We may then assume that the interaction is shaped by a single population of agents. Every element of the evolution space is both a strategy and a social state. Consider an agent using strategy $X$ (call him the $X$-agent) and another agent using strategy $Y$ (call her the $Y$-agent). Let $\mathcal{Y}_{\epsilon}$ be the social state. Then $X \cdot F(\mathcal{Y}_\epsilon)$ is the payoff the $X$-agent earns under that social state and $Y \cdot F(\mathcal{Y}_\epsilon)$ is that of the $Y$-agent. As parameter $\epsilon$ runs from $0$ to $1$, the social state traces the linear segment between social states $X$ and $Y$. That $X \succeq_S Y$ can then assume the interpretation that $X$ is {\em strictly incrementally deployable} against $Y$ in that as the social state traces the linear segment between $X$ and $Y$, the $X$-agent obtains a (strictly) higher payoff than the $Y$-agent. That $X \trianglerighteq_S Y$ similarly assumes the interpretation that $X$ is {\em not} strictly incrementally deployable against $Y$. It turns out that the maximal elements under both polyorders coincide with the equilibria of the corresponding variational inequality as we prove next.

\subsubsection{Maximal elements characterize variational equilibria}

Recall that a maximal element of a binary relation $\succeq$ is an element $X^*$ such that $\forall X : \neg(X \succ X^*)$. We have the following characterization of $\succeq_S$-maximal elements.

\begin{lemma}
\label{alskjdfjsdkfjdkjf}
Let $X^* \in \mathbb{X}$ and, for any $X \in \mathbb{X}$, let $X_\epsilon = \epsilon X + (1-\epsilon) X^*$. $X^*$ is $\succeq_S$-maximal if and only if 
\begin{align*}
\forall X \in \mathbb{X} (\forall \epsilon \in [0,1]: X^* \cdot F(X_\epsilon) > X \cdot F(X_\epsilon)) \vee (\exists \epsilon \in [0,1]: X^* \cdot F(X_\epsilon) \geq X \cdot F(X_\epsilon)).
\end{align*}
\end{lemma}

\begin{proof}
We take $\succeq$ to mean $\succeq_S$ for brevity. The following statements are all equivalent:
\begin{align*}
&X^* \mbox{ is maximal }\\
&\forall X: \neg(X \succ X^*)\\
&\forall X: \neg((X \succeq X^*) \wedge \neg(X^* \succeq X))\\
&\forall X: (X^* \succeq X) \vee \neg(X \succeq X^*)\\
&\forall X: (\forall \epsilon \in [0,1]: X^* \cdot F(X_\epsilon) > X \cdot F(X_\epsilon)) \vee \neg(\forall \epsilon \in [0,1]: X \cdot F(X_\epsilon) > X^* \cdot F(X_\epsilon))\\
&\forall X: (\forall \epsilon \in [0,1]: X^* \cdot F(X_\epsilon) > X \cdot F(X_\epsilon)) \vee (\exists \epsilon \in [0,1]: X \cdot F(X_\epsilon) \leq X^* \cdot F(X_\epsilon))\\
&\forall X: (\forall \epsilon \in [0,1]: X^* \cdot F(X_\epsilon) > X \cdot F(X_\epsilon)) \vee (\exists \epsilon \in [0,1]: X^* \cdot F(X_\epsilon) \geq X \cdot F(X_\epsilon)),
\end{align*}
as desired.
\end{proof}

\begin{lemma}
\label{alskjdfjsdkfjdkjfffff}
$X^* \in \mathbb{X}$ is $\succeq_S$-maximal if and only if it is $\trianglerighteq_S$-maximal.
\end{lemma}

\begin{proof}
We take $\succeq$ to mean $\trianglerighteq_S$ for brevity. The following statements are all equivalent:
\begin{align*}
&X^* \mbox{ is maximal }\\
&\forall X: \neg(X \succ X^*)\\
&\forall X: \neg((X \succeq X^*) \wedge \neg(X^* \succeq X))\\
&\forall X: (X^* \succeq X) \vee \neg(X \succeq X^*)\\
&\forall X: (\exists \epsilon \in [0,1]: X^* \cdot F(X_\epsilon) \geq X \cdot F(X_\epsilon)) \vee \neg(\exists \epsilon \in [0,1]: X \cdot F(X_\epsilon) \geq X^* \cdot F(X_\epsilon))\\
&\forall X: (\exists \epsilon \in [0,1]: X^* \cdot F(X_\epsilon) \geq X \cdot F(X_\epsilon)) \vee (\forall \epsilon \in [0,1]: X \cdot F(X_\epsilon) < X^* \cdot F(X_\epsilon))\\
&\forall X: (\exists \epsilon \in [0,1]: X^* \cdot F(X_\epsilon) \geq X \cdot F(X_\epsilon)) \vee (\forall \epsilon \in [0,1]: X^* \cdot F(X_\epsilon) > X \cdot F(X_\epsilon)).
\end{align*}
The lemma then follows by Lemma~\ref{alskjdfjsdkfjdkjf}.
\end{proof}

\begin{theorem}
\label{woeriutyeriutytaert}
$X^* \in \mathbb{X}$ is $\succeq_S$-maximal if and only if it is a critical element.
\end{theorem}

\begin{proof}
That every critical element is $\succeq_S$-maximal is a straightforward implication of Lemma~\ref{alskjdfjsdkfjdkjfffff} (take $\epsilon = 0$ in the definition of $\trianglerighteq_S$-maximality). Suppose now $X^* \in \mathbb{X}$ is $\succeq_S$-maximal and suppose further, for the sake of contradiction, that there exists $X \in \mathbb{X}$ such that $X^* \cdot F(X^*) < X \cdot F(X^*)$. Then, by continuity of $(X^* - X) \cdot F(X_\epsilon)$, where $X_\epsilon = \epsilon X + (1-\epsilon) X^*$, there exists a convex combination of $X^*$ and $X$, say $Y$, such that $Y \succeq_S X^*$, which implies that $\neg (X^* \succeq_S Y)$, and, therefore, that $T \succ_S X^*$, which contradicts the assumption that $X^*$ is $\succeq_S$-maximal. 
\end{proof}

\subsection{Drifting equilibria and neutral stability}

In the rest of this section, we use our formalism to define a new solution concept in continuous vector fields called a {\em drifting equilibrium,} which we show to be an ``intermediate'' concept between critical points and Minty solutions in that drifting equilibria are necessarily critical points and Minty solutions are necessarily drifting equilibria. But we further show that drifting equilibria are an intermediate concept in a strong sense. Let us give the precise definition first.

\begin{definition}
\label{alsdkjfhdskjfhooo}
Let $F: \mathbb{X} \subseteq \mathbb{R}^n \rightarrow \mathbb{R}^n$ be a continuous vector field where $\mathbb{X}$ is nonempty, closed, and convex, $X, Y \in \mathbb{X}$, and $\mathcal{Y} : [0,1] \rightarrow \mathbb{X}$ such that $\mathcal{Y}_\epsilon = \epsilon Y + (1-\epsilon) X$. Furthermore, let $\succeq$ be a binary relation on $X$ such that 
\begin{align*}
X \succeq Y \Leftrightarrow \forall \epsilon \in [0,1]: X \cdot F(\mathcal{Y}_\epsilon) \geq Y \cdot F(\mathcal{Y}_\epsilon)
\end{align*} 
and $\trianglerighteq$ be a binary relation on $\mathbb{X}$ such that
\begin{align*}
X \trianglerighteq Y \Leftrightarrow \exists \epsilon \in [0,1]: X \cdot F(\mathcal{Y}_\epsilon) > Y \cdot F(\mathcal{Y}_\epsilon).
\end{align*} 
We call the structures that $\succeq$ and $\trianglerighteq$ induce on $(F, X)$ {\em drifting linear polyorders}.
\end{definition}

Referring to our previous discussion after the definition of strict polyorders, we note that drifting polyorders capture the notion of {\em weak incremental deployability,} which gives rise (through maximality) to a stronger notion of equilibrium (in both continuous and discrete evolution spaces).

\begin{lemma}
\label{alskjdfjsdkfjdkjff}
Let $X^* \in \mathbb{X}$ and, for any $X \in \mathbb{X}$, let $X_\epsilon = \epsilon X + (1-\epsilon) X^*$. $X^*$ is $\succeq$-maximal if and only if 
\begin{align*}
\forall X \in \mathbb{X} : (\forall \epsilon \in [0,1]: X^* \cdot F(X_\epsilon) \geq X \cdot F(X_\epsilon)) \vee (\exists \epsilon \in [0,1]: X^* \cdot F(X_\epsilon) > X \cdot F(X_\epsilon)).
\end{align*}
\end{lemma}

\begin{proof}
The following statements are all equivalent:
\begin{align*}
&X^* \mbox{ is maximal }\\
&\forall X: \neg(X \succ X^*)\\
&\forall X: \neg((X \succeq X^*) \wedge \neg(X^* \succeq X))\\
&\forall X: (X^* \succeq X) \vee \neg(X \succeq X^*)\\
&\forall X: (\forall \epsilon \in [0,1]: X^* \cdot F(X_\epsilon) \geq X \cdot F(X_\epsilon)) \vee \neg(\forall \epsilon \in [0,1]: X \cdot F(X_\epsilon) \geq X^* \cdot F(X_\epsilon))\\
&\forall X: (\forall \epsilon \in [0,1]: X^* \cdot F(X_\epsilon) \geq X \cdot F(X_\epsilon)) \vee (\exists \epsilon \in [0,1]: X \cdot F(X_\epsilon) < X^* \cdot F(X_\epsilon))\\
&\forall X: (\forall \epsilon \in [0,1]: X^* \cdot F(X_\epsilon) \geq X \cdot F(X_\epsilon)) \vee (\exists \epsilon \in [0,1]: X^* \cdot F(X_\epsilon) > X \cdot F(X_\epsilon)),
\end{align*}
as desired.
\end{proof}

\begin{lemma}
\label{alskjdfjsdkfjdkjffff}
$X^* \in \mathbb{X}$ is $\succeq$-maximal if and only if it is $\trianglerighteq$-maximal.
\end{lemma}

\begin{proof}
We take $\succeq$ to mean $\trianglerighteq$ for brevity. The following statements are all equivalent:
\begin{align*}
&X^* \mbox{ is maximal }\\
&\forall X: \neg(X \succ X^*)\\
&\forall X: \neg((X \succeq X^*) \wedge \neg(X^* \succeq X))\\
&\forall X: (X^* \succeq X) \vee \neg(X \succeq X^*)\\
&\forall X: (\exists \epsilon \in [0,1]: X^* \cdot F(X_\epsilon) > X \cdot F(X_\epsilon)) \vee \neg(\exists \epsilon \in [0,1]: X \cdot F(X_\epsilon) > X^* \cdot F(X_\epsilon))\\
&\forall X: (\exists \epsilon \in [0,1]: X^* \cdot F(X_\epsilon) > X \cdot F(X_\epsilon)) \vee (\forall \epsilon \in [0,1]: X \cdot F(X_\epsilon) \leq X^* \cdot F(X_\epsilon))\\
&\forall X: (\exists \epsilon \in [0,1]: X^* \cdot F(X_\epsilon) > X \cdot F(X_\epsilon)) \vee (\forall \epsilon \in [0,1]: X^* \cdot F(X_\epsilon) \geq X \cdot F(X_\epsilon)).
\end{align*}
The lemma then follows by Lemma~\ref{alskjdfjsdkfjdkjff}.
\end{proof}

The previous lemma shows that the maximal elements of both drifting linear polyorders coincide. Our new solution concept is given in the following definition.

\begin{definition}
$X^* \in \mathbb{X}$ is a {\em drifting equilibrium} of $F: \mathbb{X} \subseteq \mathbb{R}^n \rightarrow \mathbb{R}^n$, where $F$ is continuous and $\mathbb{X}$ is nonempty, closed, and convex, if it is maximal element of a drifting linear polyorder.  
\end{definition}

The following theorems are a preliminary step in the interest of placing drifting equilibria in the broader literature on variational inequalities.

\begin{theorem}
If $X^*$ is a drifting equilibrium of $F$, then it is a critical element of $F$.
\end{theorem}

\begin{proof}
Let us assume for the sake of contradiction that $X^*$ is not a critical element of $F$. Then there exists $X \in \mathbb{X}$, where $\mathbb{X}$ is the evolution space of $F$, such that
\begin{align*}
(X^* - X) \cdot F(X^*) < 0
\end{align*}
or, rearranging,
\begin{align*}
(X - X^*) \cdot F(X^*) > 0.
\end{align*}
Let $X_{\epsilon} = \epsilon X^* + (1-\epsilon) X$, $\epsilon \in [0, 1]$. The continuity of $(X - X^*) \cdot F(X_{\epsilon})$ implies that there exists $\epsilon' \in [0,1]$ such that, $\forall \epsilon \in [\epsilon', 1]$, $X \cdot F(X_\epsilon) > X^* \cdot F(X_\epsilon)$. That is, there exists an element in the convex hull of $X^*$ and $X$ dominating $X^*$, contradicting that $X^*$ is a drifting equilibrium.
\end{proof}

\begin{theorem}
If $X^*$ is a Minty solution of $F$, then it is a drifting equilibrium of $F$.
\end{theorem}

\begin{proof}
Let $X_\epsilon = \epsilon X + (1-\epsilon) X^*$ where $\epsilon \in [0, 1]$ and $X$ is arbitrary. The definition of a Minty solution implies that, for all $X \in \mathbb{X}$, the evolution space of $F$,
\begin{align*}
(X^* - X_\epsilon) \cdot F(X_\epsilon) \geq 0,
\end{align*}
which, by straight algebra, is equivalent to
\begin{align*}
(X^* - X) \cdot F(X_\epsilon) \geq 0.
\end{align*}
The theorem is then implied by Lemma \ref{alskjdfjsdkfjdkjff}.
\end{proof}

The previous theorems imply that drifting equilibria are a refinement of critical points and that Minty equilibria are a refinement of drifting equilibria. A question of natural interest concerns the precise nature of the solution concept we have defined. To that end, we consider ESsets first.

\subsubsection{Evolutionarily stable sets and drifting equilibria}

\begin{theorem}
\label{qpwoeiruwoeiurur}
Every element of an ESset is a drifting equilibrium.
\end{theorem}

\begin{figure}[tb]
\centering
\includegraphics[width=6cm]{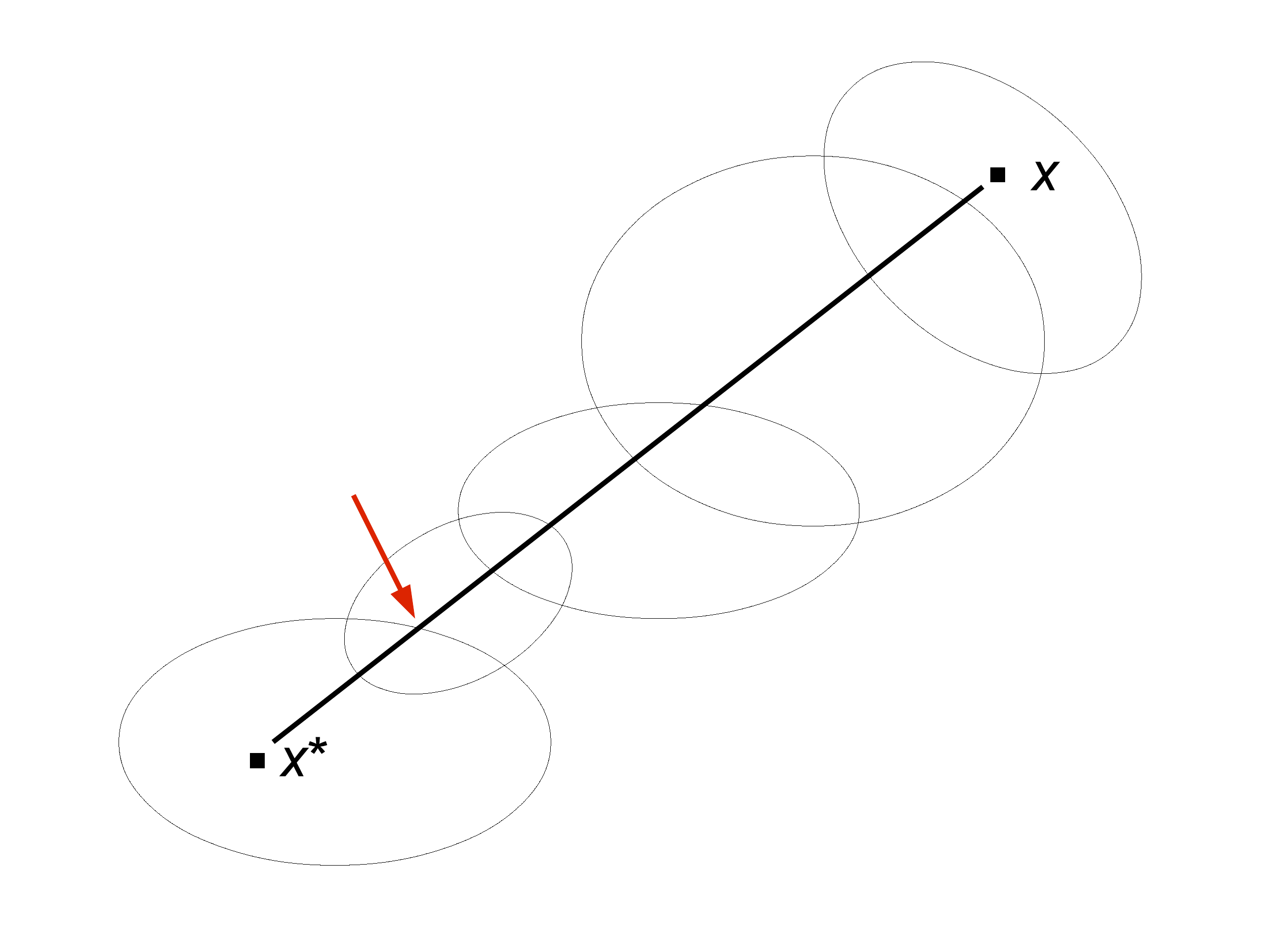}
\caption{\label{qpwoeirxwoexururx}
Used in the proof of Theorem~\ref{qpwoeiruwoeiurur}.}
\end{figure}

\begin{proof}
Let $F : \mathbb{X} \rightarrow \mathbb{R}^M$ be a population game and $\mathbf{X}^*$ an ES set of $F$. By Lemma \ref{alskjdfjsdkfjdkjff}, it suffices to show that
\begin{align*}
\forall X \in \mathbb{X}: (\forall \epsilon \in [0,1]: X^* \cdot F(X_\epsilon) \geq X \cdot F(X_\epsilon)) \vee (\exists \epsilon \in [0,1]: X^* \cdot F(X_\epsilon) > X \cdot F(X_\epsilon)).
\end{align*}
Since $X^* \in \mathbf{X}^*$, by the definition of an ESset, there exists a neighborhood $O(X^*)$ such that, $\forall X \in O(X^*)$, $(X^* - X) \cdot F(X) \geq 0$ (with strict inequality unless $X \in \mathbf{X}^*$). Since $X$ is an arbitrary element of $O(X^*)$ we have that
\begin{align*}
\forall X \in O(X^*) \forall \epsilon \in [0,1]: X^* \cdot F(X_\epsilon) \geq X \cdot F(X_\epsilon)
\end{align*}
and, therefore, trivially, that
\begin{align*}
\forall X \in O(X^*): (\forall \epsilon \in [0,1]: X^* \cdot F(X_\epsilon) \geq X \cdot F(X_\epsilon)) \vee (\exists \epsilon \in [0,1]: X^* \cdot F(X_\epsilon) > X \cdot F(X_\epsilon)).
\end{align*}
We need to show this holds for all $X \in \mathbb{X}$. Therefore, take any $X \in \mathbb{X}/O(X^*)$, and let $X_\epsilon = \epsilon X + (1-\epsilon) X^*$, $\epsilon \in [0,1]$. Let $X_{\epsilon'}$ be the point where $X_\epsilon$ intersects the boundary of $O(X^*)$ (pointed at by the arrow in Figure~\ref{qpwoeirxwoexururx}), and note that to prove our claim we only need to consider the case that, for all $\epsilon \in [0, \epsilon']$, $X^* \cdot F(X_\epsilon) = X_{\epsilon'} \cdot F(X_\epsilon)$. Note further that in this case $X_{\epsilon'} \in \mathbf{X^*}$ and consider $O(X_{\epsilon'})$. Let $X_{\epsilon''}$ be the point where $X_\epsilon$ intersects the boundary of $O(X_{\epsilon'})$ (in the direction toward $X$). Observe now that if there exists $\zeta \in [\epsilon', \epsilon'']$ such that $X_{\epsilon'} \cdot F(X_\zeta) > X_{\epsilon''} \cdot F(X_\zeta)$, then $X^* \cdot F(X_\zeta) > X \cdot F(X_\zeta)$, and that if no such $\zeta$ exists we may similarly consider the boundary point $X_{\epsilon''}$ of $O(X_{\epsilon'})$ noting that then, for all $\zeta \in [\epsilon', \epsilon'']$, $X_{\epsilon'} \cdot F(X_\zeta) = X_{\epsilon''} \cdot F(X_\zeta)$, and, therefore, that $X_{\epsilon''} \in \mathbf{X^*}$. Continuing in this way, we either obtain an $\epsilon \in [0,1]$ such that $X^* \cdot F(X_\epsilon) > X \cdot F(X_\epsilon)$ or, for all $\epsilon \in [0,1]$, $X^* \cdot F(X_\epsilon) = X \cdot F(X_\epsilon)$. Since $X$ is arbitrary, the theorem is proven.
\end{proof}

\subsubsection{Drifting equilibria are more general than neutrally stable strategies}

The saddle points in optimization theory are not drifting equilibria. Therefore, the class of critical points is more general than the class of drifting equilibria. It is natural to wonder if drifting equilibria characterize neutrally stable strategies. The answer is no in general and an example of a drifting equilibrium that is not a neutrally stable strategy is the center of the axes in Figure \ref{saldkjfnalxdkhjfg} plotting function $f(x) = x \sin(1/x)$. Note that the critical element $x=0$ is {\em both} maximal and {\em minimal} (in a manner analogous to our definition of maximal elements, we may define {\em minimal elements} as those elements of the evolution space that do not dominate any other element) {\em without} being a local maximum or minimum. We note this type of equilibrium behavior has been discovered before as a counterexample to the existence of a converse Lyapunov theorem \cite{Hahn}. \citet{Absil} study similar behavior in the setting of gradient dynamical systems. But within the setting of symmetric bimatrix games drifting equilibria characterize neutrally stable strategies.

\begin{figure}[tb]
\centering
\includegraphics[width=6cm]{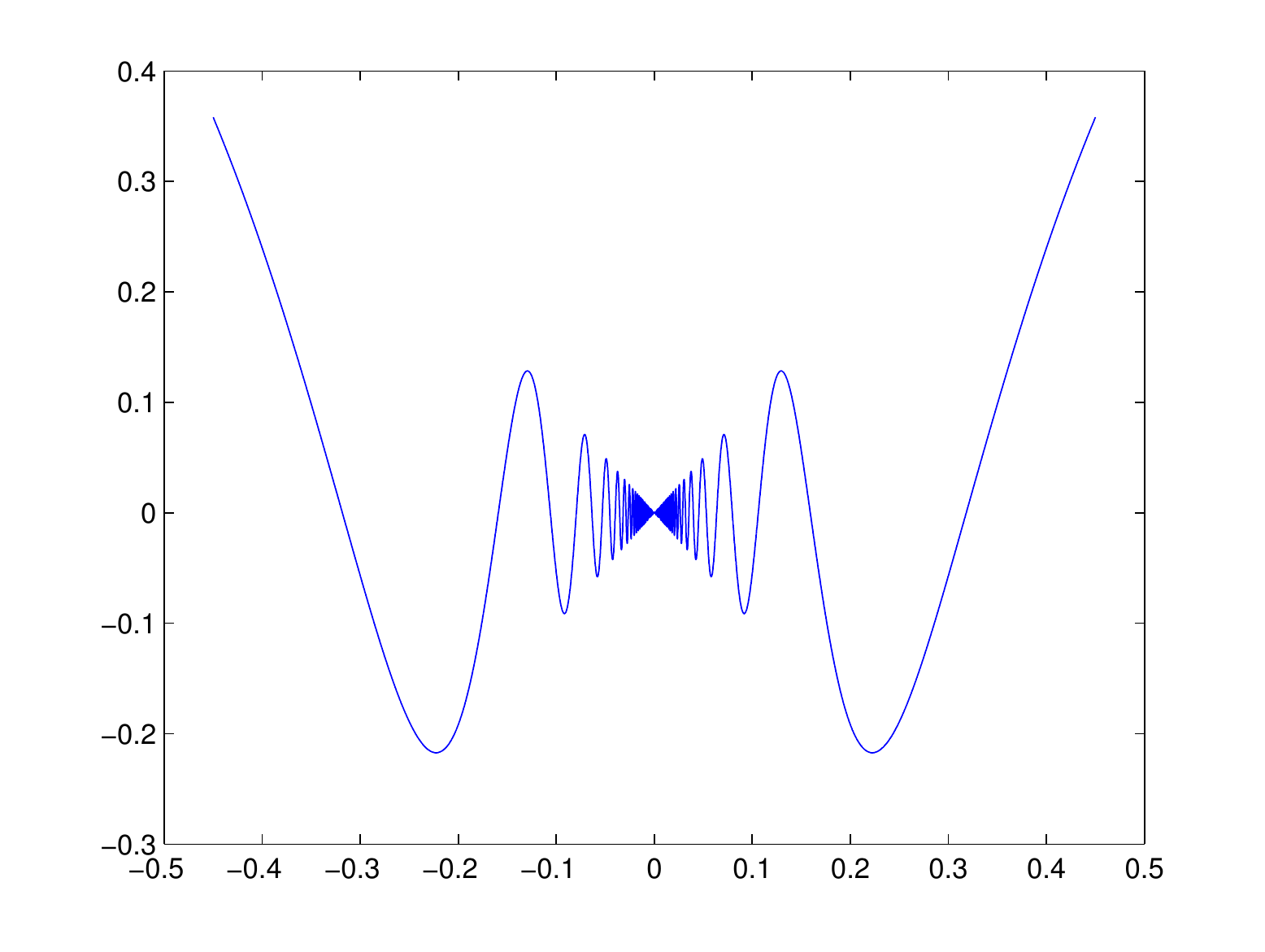}
\caption{\label{saldkjfnalxdkhjfg}
An example of a drifting equilibrium that is not neutrally stable.}
\end{figure}

\subsubsection{A characterization of neutrally stable strategies in symmetric bimatrix games}

Let us show now that in the setting of symmetric bimatrix games, drifting equilibria characterize neutrally stable strategies. Given a symmetric bimatrix game $(C, C^T)$, we use $(C, \succeq)$ to denote its drifting linear polyorder. Although we have defined two relevant polyorders, since we are only concerned with their maximal elements, in light of the previous lemma, there is no ambiguity. 

\begin{theorem}
\label{NSS_characterization}
Let $(C, C^T)$ be a symmetric bimatrix game. Then $X^* \in \mathbb{X}(C)$ is a maximal element of the drifting polyorder $(C, \succeq)$ if and only if $X^*$ is an NSS of $(C, C^T)$.
\end{theorem}

We prove this theorem in a series of lemmas.

\begin{lemma}
\label{cool_basic_lemma}
Let $(C, C^T)$ be a symmetric bimatrix game, let $(C, \succeq)$ be its drifting linear polyorder, and let $X^*$ be a maximal element of $(C, \succeq)$. Then, for all $X \in \mathbb{X}(C)$, we have that there exists $Y$ as a convex combination of $X^*$ and $X$ such that
\begin{align*}
\forall \epsilon \in [0,1] : X^* \cdot C\mathcal{Y}_\epsilon \geq Y \cdot C\mathcal{Y}_\epsilon
\end{align*}
where $\mathcal{Y}_\epsilon = \epsilon Y + (1-\epsilon) X^*$.
\end{lemma}

\begin{proof}
Let $X$ be arbitrary, and let $X_\epsilon = \epsilon X + (1-\epsilon) X^*$. Lemma \ref{alskjdfjsdkfjdkjff} implies then that either, for all $\epsilon \in [0,1]$, $X^* \cdot C X_\epsilon \geq X \cdot C X_\epsilon$ or that there exists $\epsilon \in [0,1]$ such that $X^* \cdot C X_\epsilon > X \cdot C X_\epsilon$. If the first possibility of the lemma manifests for all $X$, the theorem trivially holds. It, therefore, suffices to only consider elements $X$ such that the second possibility manifests. Let us, therefore, assume that $X$ is such. We may now consider two cases: Either for all $\epsilon \in [0,1]$ we have that $X^* \cdot C X_\epsilon > X \cdot C X_\epsilon$ or there exists $Y$ as a convex combination of $X^*$ and $X$ such that $X^* \cdot CY = X \cdot CY$. Let us, therefore, consider the second case. Since $Y$ is a convex combination of $X^*$ and $X$ and $X^* \cdot CY = X \cdot CY$, straight algebra implies that $X^* \cdot CY = Y \cdot CY$. Let $\mathcal{Y}_\epsilon  = \epsilon Y + (1-\epsilon) X^*$. We then have that
\begin{align*}
X^* \cdot C\mathcal{Y}_\epsilon  - \mathcal{Y}_\epsilon  \cdot C\mathcal{Y}_\epsilon  = X^* \cdot C\mathcal{Y}_\epsilon  - (\epsilon Y + (1-\epsilon) X^*) \cdot C\mathcal{Y}_\epsilon  = \epsilon (X^* - Y) \cdot C\mathcal{Y}_\epsilon 
\end{align*}
Furthermore,
\begin{align*}
(X^* - Y) \cdot C\mathcal{Y}_\epsilon  &= (X^*-Y) \cdot C(\epsilon Y + (1-\epsilon) X^*)\\
 &= \epsilon (X^*-Y) \cdot CY + (1-\epsilon) (X^*-Y) \cdot CX_*\\
 &= (1-\epsilon) (X^* - Y) \cdot CX^*,
\end{align*}
where we have used that $X^* \cdot CY = X \cdot CY$. But $X^*$ is a Nash equilibrium, and, therefore,
\begin{align*}
(X^* - Y) \cdot CX^* \geq 0.
\end{align*}
Therefore, 
\begin{align*}
(X^* - Y) \cdot C\mathcal{Y}_\epsilon  \geq 0.
\end{align*}
But since $\mathcal{Y}_\epsilon$ is an arbitrary convex combination of $X^*$ and $Y$, the lemma follows.
\end{proof}

\begin{lemma}
\label{near_final_1}
Let $(C, C^T)$ be a symmetric bimatrix game, and let $X^* \in \mathbb{X}$ be an NSS of $(C, C^T)$. Then $X^*$ is a maximal element of the drifting polyorder $(C, \succeq)$.
\end{lemma}

\begin{proof}
Let $X \in \mathbb{X}$ be arbitrary, and let $X_\epsilon = \epsilon X + (1-\epsilon) X^*$. If $(X^* - X) \cdot C X^* > 0$, then the condition in Lemma \ref{alskjdfjsdkfjdkjff} is trivially satisfied. Suppose now that $(X^* - X) \cdot C X^* = 0$. Then, for all $\epsilon \in [0,1]$,
\begin{align*}
(X^* - X_\epsilon) \cdot C X^* &= (X^* - \epsilon X - (1-\epsilon) X^*) \cdot CX^*\\
 &= \epsilon (X^* - X) \cdot CX^*\\
 &= 0.
\end{align*}
Therefore, Proposition \ref{nss_1} implies that, for all $\epsilon \in [0,1]$, $X^* \cdot CX \geq X \cdot CX$, and, thus, the condition in Lemma \ref{alskjdfjsdkfjdkjff} is also satisfied. This completes the proof.
\end{proof}

\begin{lemma}
\label{near_final_2}
Let $(C, C^T)$ be a symmetric bimatrix game, and let $X^* \in \mathbb{X}$ be a maximal element of the drifting polyorder $(C, \succeq)$. Then $X^*$ is an NSS of $(C, C^T)$.
\end{lemma}

\begin{proof}
Straightforward implication of Lemma \ref{cool_basic_lemma} and the definition of an NSS.
\end{proof}

\begin{proof}[Proof of Theorem \ref{NSS_characterization}]
Straightforward implication of Lemmas \ref{near_final_1} and \ref{near_final_2}.
\end{proof}

\subsection{Mathematical foundations in discrete evolution spaces}

Given a game $\Gamma$ in strategic form, we define its {\em deployment graph} as the directed graph $G$ whose vertices correspond to $\Gamma$'s strategy profiles and whose arcs correspond to a notion of ``feasible'' unilateral deviations. Therefore, the set of vertices of $G$ is the direct product $\times_i S_i$. The notion of feasibility may admit a variety of plausible definitions. In this paper, we explore two such definitions, namely, we either call a unilateral deviation feasible if it is (strictly) profitable for the corresponding player that deviates (switches strategies) or if it is not harmful. For example, the deployment graph of a stag hunt shown on the left in Figure \ref{dgsh} is shown on the right.

\begin{figure}[tb]
\centering
\includegraphics[width=8cm]{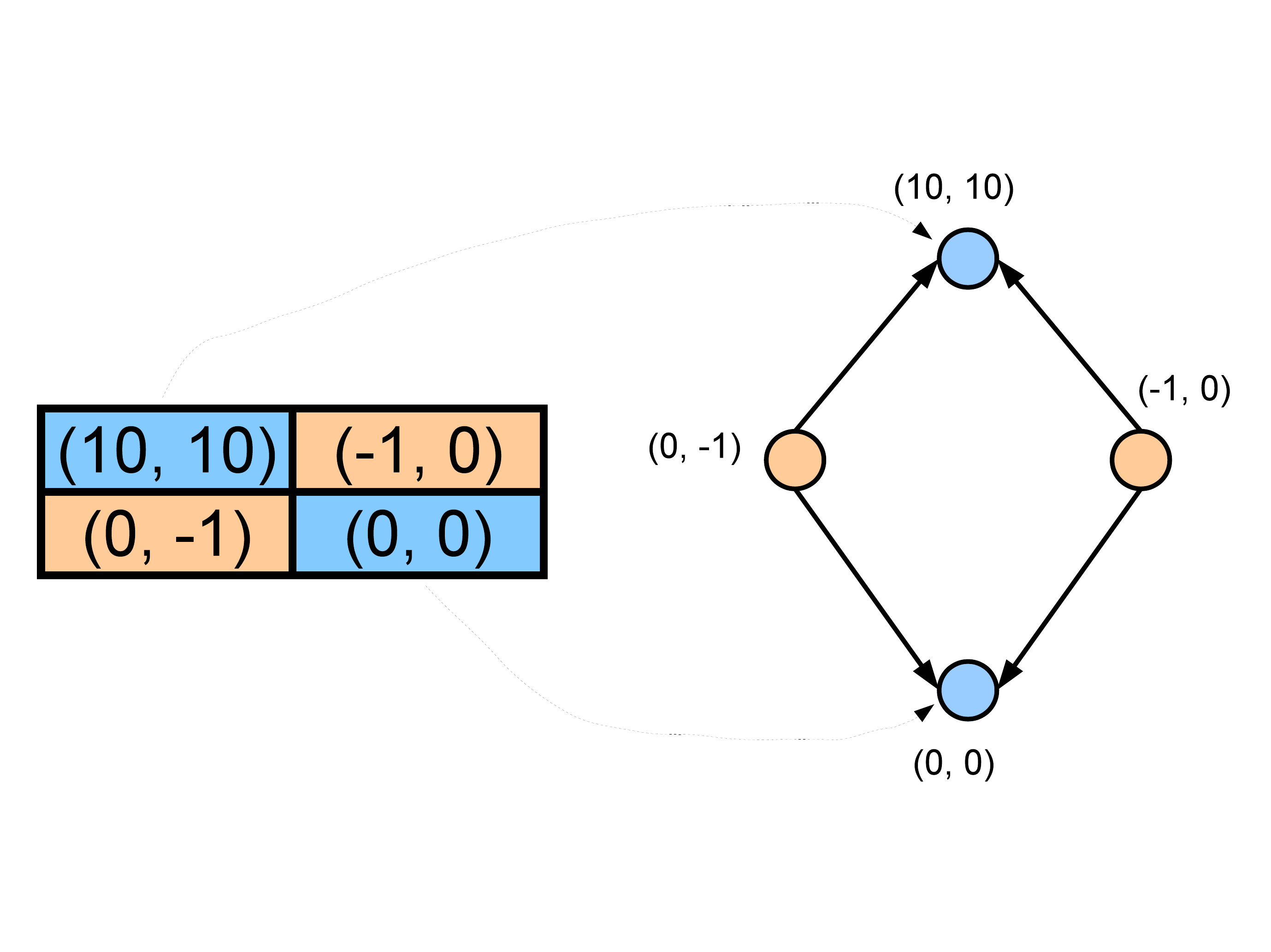}
\caption{\label{dgsh}
Example of deployment graph of a (two player) stag hunt.}
\end{figure}

We are now ready to provide a rigorous definition of incremental deployability: We say that a strategy profile (state) $\sigma$ of $\Gamma$ is {\em incrementally deployable} against strategy profile (state) $\sigma'$ if there exists a directed path from $\sigma'$ to $\sigma$ in $\Gamma$'s deployment graph $G$. To understand this definition in the particular setting of architectural innovation we consider, we may think of $\Gamma$ as capturing the incentive structure of the adoption environment where possibly multiple technologies (corresponding in this particular setting to strategies) compete for adoption. Then a state (a mix of technologies) is incrementally deployable against another state, if there exists a sequence of beneficial ``moves'' players of the adoption environment can take unilaterally, where a move corresponds to a player abandoning one technology in favor of another, to effect such change in the technology mix.

The following question seems to be in order given the aforementioned definitions: In an environment where evolution is driven by the process of incremental steps we have defined, which states, if any, would seem natural to emerge in the ``steady state'' as evolutionary forces have settled down? Although we do not specify a particular mechanism by which evolution should take place, it is natural to assume that if a state $s$ is incrementally deployable against another state $s'$ but $s'$ is not incrementally deployable against $s$, $s'$ cannot manifest other than as a transient state, an assumption in agreement with the colloquial notion of incremental deployability as that is used in the networking community. The following definitions rigorously capture this intuition.

\subsubsection{Maximal states as a game-theoretic solution concept}

Let us first define deployment graphs in more rigorous terms. Recall that there are two distinct notions of such graphs, one where arcs correspond to (strictly) profitable unilateral deviations and one where are arcs correspond to unilateral deviations that are not harmful, which we will refer to as {\em strict deployment graphs} and {\em ordinal deployment graphs} respectively.

\begin{definition}[Strict deployment graph]
Let $S$ be the profile space of a normal form game $\Gamma$. We define the {\em improvement graph} $G(S, A)$ of $\Gamma$ as the directed graph that represents the players' (strictly) profitable unilateral deviations, that is, 
\begin{align*}
(s, s') \in A \Leftrightarrow \exists i : (s'_i, s_{-i}) = s' \mbox{ and } u_i(s') - u_i(s) > 0.
\end{align*}
\end{definition}

\begin{definition}[Ordinal deployment graph]
Let $\Gamma$ be a normal form game, and let $S$ be its profile space. We define the {\em advancement graph} $G(S, A)$ of $\Gamma$ as the directed graph that represents the players' weakly profitable unilateral deviations, that is, 
\begin{align*}
(s, s') \in A \Leftrightarrow \exists i : (s'_i, s_{-i}) = s' \mbox{ and } u_i(s') - u_i(s) \geq 0.
\end{align*}
If the above inequality is strict, we call the arc {\em positive}, and otherwise we call the arc {\em neutral}. Note that neutral arcs always come in {\em pairs} (that is, if $(s, s')$ is neutral, then $(s',s)$ is a neutral arc). 
\end{definition}

We are now ready to define our solution concept in two distinct variants according to the notion of a deployment graph that is appropriate in the modeling environment that is of interest.

\begin{definition}[Weakly maximal states]
Let $\Gamma$ be a normal form game, let $S$ be its profile space, let $G$ be its strict deployment graph, and order the vertices of this graph as follows:
\begin{align*}
s \succeq s' \Leftrightarrow \mbox{There exists a path in $G$ from $s'$ to $s$}.
\end{align*}
We call $\succeq$, the strict preference relation on $S$. We call the maximal elements of $(S, \succeq)$ the {\em weakly maximal states} of $\Gamma$.
\end{definition}

\begin{definition}[Strongly maximal states]
Let $\Gamma$ be a normal form game, let $S$ be its profile space, let $G$ be its ordinal deployment graph, and order the vertices of this graph as follows:
\begin{align*}
s \succeq s' \Leftrightarrow \mbox{There exists a path in $G$ from $s'$ to $s$}.
\end{align*}
We call $\succeq$, the drifting preference relation on $S$. We call the maximal elements of $(S, \succeq)$ the {\em strongly maximal states} of $\Gamma$.
\end{definition}

Both the strict and drifting preference relations are transitive, and since reflexivity is trivial to satisfy by a matter of definition, both types of preference relations are preorders. Proposition \ref{weoprtiueorituoirtu} then implies the following existence theorem on (weakly and strongly) maximal states.

\begin{theorem}
\label{existence}
Every strategic-form game is equipped with both weakly and strongly maximal states.
\end{theorem}

Therefore, as game-theoretic solution concepts in pure strategies, weakly and strongly maximal states satisfy existence, a property on par with the Nash equilibrium in mixed strategies.

We note that an alternative proof of existence of (weak and strong) maximal states can be couched in graph-theoretic terms using the well-known property that a ``directed acyclic graph'' (using the standard notion of acyclicity) is guaranteed to have a {\em sink} (i.e., a vertex without outbound arcs). To that end, we consider the quotient sets of the strict and drifting preference relations of a normal form game. Then, if we contract vertices belonging to the same member of the quotient set, we obtain a directed acyclic graph that, as mentioned above, is always equipped with a sink, which corresponds to a class of maximal elements (either in the strict or drifting preference relation).

\subsubsection{Maximal states and sink equilibria}

Maximal states are a game-theoretic solution concept defined in a ``static'' manner. In the sequel, we provide a dynamic interpretation of maximal states in a model of evolution using the well-studied concept of {\em Markov chains}. Let us first introduce some relevant background definitions.

A Markov chain is a sequence of discrete random variables taking values in some countable set, say $S$, called the {\em state space}. If the cardinality of $S$ is finite, we are referring to a {\em finite state} Markov chain. A Markov chain is called {\em homogeneous} is its behavior can be described in terms of {\em transition probabilities} between states: The probability of transitioning from one state, say $s \in S$, to another, say $s' \in S$, depends neither on history nor the time step but only on $s$ and $s'$. 

The ``solution concept'' in homogeneous finite state Markov chains is defined based on a notion of ``recurrency.'' In such a Markov chain, states are classified as being either {\em transient} or {\em recurrent}. A state is  recurrent if, starting from that state, the probability of returning to it is one, and it is transient otherwise. Finite state homogeneous Markov chains admit a graph-theoretic representation (cf. \cite{Gallager2}) in which states correspond to vertices and feasible transitions, i.e., transitions between states that may occur with positive probability, correspond to arcs. (Such a graph is typically referred to in the literature as the {\em transition graph}.) It is possible to show that recurrent states are independent on the precise values of transition probabilities (for as long as arcs correspond to transitions that may occur with strictly positive probabilities). That recurrent states are ensured to exist and that, starting from any state, a chain will eventually transition to a recurrent class of states (and remain there, thereafter) are well-known results in this theory.

With this background in mind, it seems natural to define a game-theoretic solution concept based on the notion of homogeneous finite state Markov chains known in the literature on algorithmic game theory as the {\em sink equilibrium} \citep{Goemans}. Sink equilibria are classes of recurrent states in a Markov chain where states correspond to the strategy profiles of a game in strategic form and where transition probabilities are defined based on better and best responses.  

It can be shown using standard characterizations of recurrent states (as can be found, for example, in \cite{Gallager2}) that the notions of maximal states and sink equilibria coincide (as long as the deployment and transition graphs coincide). This implies that maximality, as we earlier defined it using the static concept of incremental deployability, admits a rather appealing evolutionary interpretation in which players take turns and switch strategies according to the transition (deployment) graph. The outcome of this process can be either a pure Nash equilibrium, a collection (class) of equilibrium states, or more complicated cyclical behavioral patterns, as discussed in the sequel. In closing, we note that the solution concept of the sink equilibrium is used exclusively, to the extent of our knowledge, in the literature as being synonymous with the notion of weak maximality. In the sequel, we attempt a more thorough comparison between the notions of weak and strong maximality that we draw on later in the analysis of coordination mechanisms.

\subsection{Weak versus strong maximality}
\label{on_maximality}

Weakly and strongly maximal states are rather different solution concepts. It is easy to see that weakly maximal equilibria coincide with pure Nash equilibria, however, although strongly maximal equilibria are necessarily pure Nash equilibria, there exist pure Nash equilibria that are not strongly maximal states. (We will come across this rather fine point again in the analysis of the election mechanism.) In this section, we argue that the strongly maximal equilibrium is an as plausible equilibrium solution concept in strategic games as its well-established weakly maximal counterpart (the pure Nash equilibrium). Our argument is couched in the setting of {\em ordinal potential games}.

\subsubsection{Potential games}

Potential games were introduced by \cite{Potential} as strategic form games whose incentive structure is captured by a scalar potential function $P: S \rightarrow \mathbb{R}$ where $S$ is the space of strategy profiles of a game in strategic form. Although there are many variants, two important classes of potential games are the {\em ordinal potential games} and the {\em generalized ordinal potential games}. Let us, therefore, introduce these notions more formally: A function $P: S \rightarrow \mathbb{R}$ is an {\em ordinal potential} for a (finite) strategic form game $\Gamma$, if
\begin{align*}
u_i(s, \sigma_{-i}) - u_i(s', \sigma_{-i}) > 0 \Leftrightarrow P(s, \sigma_{-i}) - P(s', \sigma_{-i}) > 0, \forall i \in I, \forall s, s' \in S_i, \forall \sigma_{-i} \in S_{-i}.
\end{align*}
$\Gamma$ is called an {\em ordinal potential game} if it admits an ordinal potential. A function $P: S \rightarrow \mathbb{R}$ is a {\em generalized ordinal potential} for a strategic form game $\Gamma$, if
\begin{align*}
u_i(s, \sigma_{-i}) - u_i(s', \sigma_{-i}) > 0 \Rightarrow P(s, \sigma_{-i}) - P(s', \sigma_{-i}) > 0, \forall i \in I, \forall s, s' \in S_i, \forall \sigma_{-i} \in S_{-i}.
\end{align*}
$\Gamma$ is called an {\em generalized ordinal potential game} if it admits a generalized ordinal potential. These classes of potential games are distinct from each other: Ordinal potential games are necessarily generalized ordinal potential games, as is apparent from the definition, but the converse is not necessarily true. Monderer and Shapley show that the class of generalized ordinal potential games coincides with the class of (finite) games in strategic form having the {\em finite improvement property,} which is typically abbreviated as the FIP. A game in strategic form has the FIP if ``improvement paths,'' that is, sequences of unilateral better responses, are necessarily finite. They further show that games having the FIP are necessarily equipped with a pure Nash equilibrium.

\subsubsection{Strongly maximal equilibria}

In the sequel, we introduce an equilibrium solution in strategic form games in pure strategies, which we refer to as a {\em strongly maximal equilibrium,} that, in general, differs from the pure Nash equilibrium: Strongly maximal equilibria are necessarily pure Nash equilibria, but there exist pure Nash equilibria that are not strongly maximal. To introduce the strongly maximal equilibrium more formally, let us consider the ordinal deployment graph (as defined earlier). In such a graph, we say that two states (strategy profiles), say $s$ and $s'$, {\em communicate} if deployment paths exist in both directions, that is, from $s$ to $s'$ as well as from $s'$ to $s$. We say that a set of states is a {\em strongly maximal equilibrium class} if all states in the set are pure Nash equilibria that pairwise communicate. We call the states in a strongly maximal equilibrium class {\em strongly maximal equilibria}. In the sequel, we prove that maximal states in ordinal potential games are necessarily strongly maximal equilibria. Since maximal states are guaranteed to exist in any game in strategic form, strongly maximal equilibria are guaranteed to exist in ordinal potential games. We also give an example of a generalized ordinal potential game that is not equipped with a strongly maximal equilibrium.

\begin{definition}
Let $\Gamma$ be a game in strategic form, let $S$ be its space of strategy profiles (states), and let $G$ be its ordinal deployment graph. Let us further consider the quotient set $S/_\sim$. We say that $G$ is {\em ordinally acyclic} if all arcs within each member of $S/_\sim$ are neutral.
\end{definition}

In the following theorem, we call a cycle {\em positive} if it contains at least one positive arc.

\begin{theorem}
\label{opg}
A game $\Gamma$ in strategic form is an ordinal potential game if and only if $\Gamma$'s ordinal deployment graph $G$ is ordinally acyclic.
\end{theorem}

\begin{proof}
Let us assume first that $\Gamma$ is an ordinal potential game, let $S$ be $\Gamma$'s space of strategy profiles (states), and let us further assume, for the sake of contradiction, that $\Gamma$'s ordinal deployment graph $G$ has a positive cycle, say consisting of vertices (states) $s_1, \ldots, s_k$, which, by the definition of the quotient set $S/_\sim$, belong to the same member of $S/_\sim$. Then, by the definition of ordinal potential games, it must hold that $P(s_k) \geq \cdots \geq P(s_1) \geq P(s_k)$, where $P(\cdot)$ is $\Gamma$'s ordinal potential function, and at least one of the previous inequalities is strict. But this is an impossibility since potential functions are single-valued. Therefore, all arcs within a member of the quotient set $S/_\sim$ are neutral, and, thus, ordinal potential games are necessarily ordinally acyclic.

Conversely, suppose $\Gamma$ is ordinally acyclic (that is, all arcs in each member of the quotient set $S/_\sim$ are neutral). We may then construct an ordinal potential function $P : S \rightarrow \mathbb{R}$ as follows: Consider the graph that results once we contract (using the standard graph-theoretic notion of a contraction in directed graphs) all neutral arcs in each member of $S/_\sim$. Call this the strict ordinal deployment graph, which is easy to prove that it is acyclic. Consider further a topological ordering of the strict ordinal deployment graph, say $(v_1, \ldots, v_K)$, and assign potential values to these vertices such that $P(v_K) > \cdots > P(v_1)$. Noting that each vertex $v_i, i =1, \ldots, K$ corresponds to a set of vertices of the ordinal deployment graph $G$, and assigning to each such vertex the potential value that has been assigned to the corresponding vertex of the strict ordinal deployment graph, it can be easily verified the potential function we have defined is an ordinal potential function.
\end{proof}

The previous theorem implies that all maximal states of an ordinal potential game are strongly maximal equilibria. Theorems \ref{existence} and \ref{opg} imply the following corollary.

\begin{corollary}
\label{existence_strong}
Every ordinal potential game is equipped with a strongly maximal equilibrium.
\end{corollary}

Let us finally give an example of a generalized ordinal potential game that is not equipped with a strongly maximal equilibrium. Consider, to that end, the following game:
\begin{align*}
\left(\begin{array}{c c}
(1, 0) & (2, 0)\\
(2, 0) & (0, 1)\\
\end{array}\right).
\end{align*}
\cite{Potential} show that the previous game is not an ordinal potential game but it is a generalized ordinal potential game. Observe that the strategy profile $(2, 0)$ is the game's unique pure Nash equilibrium, however, it is easy to see that it is not strongly maximal.

\subsubsection{Plausibility of the strongly maximal equilibrium}

The Nash equilibrium in mixed strategies admits a characterization rendering it quite appealing in strategic environments where mixed strategies are a plausible modeling artifact. \cite{Norde} show that the mixed Nash equilibrium is characterized by three properties referred to as {\em existence}, {\em one-person rationality,} and {\em consistency}. It will become evident from the definition of one-person rationality and consistency that these are intuitive properties that {\em any} equilibrium concept should satisfy. (To some extent, the property of existence of an equilibrium, although desirable and pleasing, does not, in my opinion, affect the predictive value of a solution concept as argued by the abundance of non-equilibrium behavior in nature, biology, and society.) However, although the characterization of Norde et al. implies that any Nash equilibrium refinement must violate one of these properties within the entire class of normal form games, within narrower classes of games the set of mixed Nash equilibria is not necessarily {\em minimal} with respect to these properties. In particular, within the class of ordinal potential games, the mixed Nash equilibrium is known not to be minimal \citep{Peleg}. In the sequel, we show that in this class of games strongly maximal equilibria satisfy one-person rationality and consistency (existence is implied by Corollary \ref{existence_strong}).\\

Let $\mathcal{G}$ be a class of normal form games. A {\em solution concept} (or {\em solution rule}) in $\mathcal{G}$, $\phi : \mathcal{G} \rightarrow 2^{S}$, assigns to each normal form game $\Gamma \in \mathcal{G}$ a possibly empty subset of the profiles $S$ of $\Gamma$. A solution rule $\phi$ is said to have the {\em existence} property within the class $\mathcal{G}$ if, for all games $\Gamma \in \mathcal{G}$, there exists at least one profile selected by the rule, that is, $\phi(\Gamma) \neq \emptyset$. Let $\mathcal{G}_1$ be the class of one-person games within $\mathcal{G}$. A solution rule $\phi$ is said to have the {\em one-person rationality} property within the (broader) class $\mathcal{G}$ if, for all one-person games in $\mathcal{G}_1$, $\phi$ selects profiles whose payoff is maximum.

Let $\Gamma = (I, (S_i)_{i \in I}, (u_i)_{i \in I})$ be a normal form game within a class $\mathcal{G}$ of normal form games, let $s \in S$ where $S$ is the profile space of $\Gamma$, and let $J$ be a {\em proper subcoalition} of the player set $I$ of $\Gamma$, in that $J$ is a nonempty proper subset of the player set. The {\em reduced game} $\Gamma^{J, s} = (J, (S_j)_{j \in J}, (w_j)_{j \in J})$ of $\Gamma$ with respect to $J$ and $s$ is a normal form game whose player set is $J$. Each player $j \in J$ has the same strategy space $S_j$ as in $\Gamma$. Let $\mathcal{S} = \times_{j \in J} S_j$ be the profile space of the reduced game. The reduced game's payoffs are defined such that, for each profile $\sigma \in \mathcal{S}$, $w_j(\sigma) = u_j(\sigma, s_{-J})$. A solution rule is said to be {\em consistent} within a class $\mathcal{G}$ of normal form games, if, for all games $\Gamma \in \mathcal{G}$, all proper subcoalitions $J \subset I$ and all solutions $s^* \in \phi(\Gamma)$, we have that 
\begin{description}

\item[(i)] If players outside the subcoalition $J$ select strategies corresponding to $s^*$, then the resulting reduced game $\Gamma^{J, s^*}$ falls within the class $\mathcal{G}$.

\item[(ii)] The profile $s^*_J$ of the reduced game where players of the subcoalition $J$ also select strategies corresponding to $s^*$ is a solution of the reduced game  $\Gamma^{J, s^*}$.

\end{description}

In the proof of the following theorem, we need a lemma in whose statement and proof some definitions are in order. Given a (simple) path in the ordinal deployment graph, let us call the first vertex in the path that path's {\em tail} and the last vertex its {\em head}. Let us call a path in the ordinal deployment graph an {\em advancement path} if at least one of its respective arcs is a positive arc.

\begin{lemma}
\label{help_lemma}
Let $\Gamma$ be an ordinal potential game, let $S$ be its profile space, and let $\mathcal{W}$ be the set of all paths in $\Gamma$'s ordinal deployment graph. Then $s^* \in S$ is maximal if and only if there is no advancement path in $\mathcal{W}$ whose tail is $s^*$.
\end{lemma}

\begin{proof}
Suppose that there exists such a path whose tail is $s^*$ and let $s$ be its head. Then clearly $s \succeq s^*$. We claim then that, in fact, $s \succ s^*$ (where $\succ$ is the strict part of $\succeq$), for if any path exists whose head is $s^*$ and tail is $s$, then concatenating the advancement path from $s^*$ to $s$ and the path from $s$ to $s^*$ gives a positive cycle, contradicting the assumption $\Gamma$ is an ordinal potential game.

Suppose now that $s^*$ is not maximal. Then there exists $s$ such that $s \succ s^*$. We claim then that, in fact, the path from $s^*$ to $s$ is an advancement path. For if it is not, then it must be a neutral path (that is, a path all of whose arcs are neutral), which would contradict that $s \succ s^*$.
\end{proof}

\begin{theorem}
\label{qwerporeiroierurux}
Let $O(\mathcal{G})$ be the class of ordinal potential games. Then the solution rule $\phi : O(\mathcal{G}) \rightarrow 2^{S}$ that selects for each $\Gamma \in O(\mathcal{G})$, the set of strongly maximal equilibria of $\Gamma$ satisfies existence, one person rationality, and consistency.
\end{theorem}

\begin{proof}
Corollary \ref{existence_strong} implies existence. One-person rationality follows by the fact that the drifting preference relation (used in the definition of strongly maximal states) is {\em rational} in one-player games (that is, reflexive, transitive, and complete). To show consistency we follow a graph-theoretic argument the main idea of which is that since ordinal potential games are ordinally acyclic, they remain so if players outside of any subcoalition are restricted to a strongly maximal equilibrium.

More precisely, let $\Gamma \in O(\mathcal{G})$, let $J$ be a proper subcoalition of the player set of $\Gamma$, and let $s^* \in \phi(\Gamma)$. By Theorem~\ref{opg}, $\Gamma$'s ordinal deployment graph contains no positive cycles. Furthermore, the ordinal deployment graph of $\Gamma^{J, s^*}$ is a subgraph of $\Gamma$'s ordinal deployment graph. Therefore, $\Gamma^{J, s^*}$'s ordinal deployment graph contains no positive cycles either, which implies, also by Theorem~\ref{opg}, that $\Gamma^{J, s^*}$ is an ordinal potential game (that is, $\Gamma^{J, s^*} \in O(\mathcal{G})$). 

It remains to show that if $s^*$ is a strongly maximal equilibrium of $\Gamma$, then, for any proper subcoalition $J$ of players, $s^*_J$ is a strongly maximal equilibrium of $\Gamma^{J, s^*}$. Since $s^*$ is maximal in $\Gamma$ there is no advancement path in $\Gamma$'s ordinal deployment graph whose tail is $s^*$. But since the advancement graph of $\Gamma^{J, s^*}$ is a subgraph (using the standard definition as a subset of vertices and their incident arcs) of $\Gamma$'s advancement graph, the set of paths in $\Gamma^{J, s^*}$ is a subset of the set of paths in $\Gamma$. Therefore, there is no advancement path in $\Gamma^{J, s^*}$'s ordinal deployment graph whose tail is $s^*$ either, and, thus, Lemma \ref{help_lemma} implies $s^*_J$ is a strongly maximal equilibrium of $\Gamma^{J, s^*}$.
\end{proof}

\subsubsection{Weak acyclicity and weak ordinal acyclicity}

The definition of ``weak acyclicity'' (formulated in a study on the evolution of social conventions by \cite{Peyton2}) precedes the definition of the sink equilibrium, however, it can also be couched in sink-equilibrium (and, therefore, also evolutionary maximality) terms: A game in strategic form is weakly acyclic if and only if the only recurrent states of the Markov chain whose transition graph has arcs corresponding to strictly profitable unilateral deviations are pure Nash equilibria \citep{Fabrikant2}. Rephrasing this definition, a game is weakly acyclic if and only if weakly maximal states are pure Nash equilibria. Weak acyclicity bears keen relevance to a notion of {\em evolutionary dynamics}: In a weakly acyclic game, {\em better response dynamics,} that is, dynamics wherein players respond to the current state by unilaterally taking, in randomly selected turns, better responses, are ensured to converge to a weakly maximal equilibrium state (that is, a pure Nash equilibrium). The importance of weak acyclicity is further argued by related, in this vein, results on dynamics in strategic-form games, namely, that in weakly acyclic games {\em no-regret dynamics} are also guaranteed to converge to a pure Nash equilibrium \citep{Peyton2, MYAS}. In a manner analogous to the definition of weak acyclicity, we obtain the following definition that proves to be important in the subsequent analysis of coordination mechanisms:

\begin{definition}
A game in strategic form is called {\em weakly ordinally acyclic} if its strongly maximal states are necessarily strongly maximal equilibria.
\end{definition}

\section{Computational foundations of incremental deployability}
\label{computational_foundations}

We showed in the previous section that the variational (Nash) and sink equilibrium can be derived from the concept of incremental deployability and, thus, the latter is at least as general. The Nash equilibrium has received severe criticism as a foundation of strategic economic behavior owing to results in theoretical computer science suggesting that the problem of computing a Nash equilibrium is hard \citep{Daskalakis, CDT}. In this section, we ask: How can incremental deployability inform the debate concerning the complexity of equilibrium computation?

\subsection{Complexity of equilibrium computation in symmetric bimatrix games}

We focus in this section on continuous evolution spaces wherein equilibrium computation is {\bf PPAD}-hard. (Discrete evolution spaces face a higher barrier to cope with as a variety of computational tasks related to sink equilibrium computation become {\bf PSPACE}-hard \citep{Mirrokni}. We discuss discrete evolution spaces in more detail in the next section alongside with Internet routing.) \cite{CDT} show that finding a Nash equilibrium is a {\bf PPAD}-complete problem even in $2$-person (bimatrix) games and that an equilibrium fully polynomial time approximation scheme (FPTAS) in these games (under either of ``additive" or ``multiplicative" notions of payoff approximation for a Nash equilibrium) implies {\bf P = PPAD}. Drawing on this result, in Appendix \ref{symm_equilibrium_appendix}, we use a symmetrization method of \cite{Jurg} to show that an equilibrium FPTAS for an {\em symmetric equilibrium} in a symmetric bimatrix game continues to imply {\bf P=PPAD}.\footnote{An anonymous reviewer of a related paper that was submitted at the International Journal of Game Theory (IJGT), in a review that was communicated to me by co-editor Professor Bernhard von Stengel, in rejecting my submitted paper argues that a reduction from an $n \times n$ bimatrix game to a $2n \times 2n$ symmetric bimatrix game is a folklore result that implies the {\bf PPAD}-hardness of equilibrium approximation in symmetric bimatrix games, but says he or she was not able to trace this result in the literature (other than citing a related text authored by Constantinos Daskalakis, which doesn't prove the reviewer's assertion on the existence of such a reduction). Following \citep{PopulationGames}, consider a bimatrix game $(A, B)$ and write the pair of equilibrium conditions as
\begin{align*}
\left[ \begin{array}{c}
 P^* - P\\
 Q^* - Q\\
\end{array} \right] \left[ \begin{array}{cc}
 0  & A\\
 B^T  & 0\\
\end{array} \right]
\left[ \begin{array}{c}
 P^*\\
 Q^*\\
\end{array} \right] \geq
\left[ \begin{array}{c}
 0\\
 0\\
\end{array} \right].
\end{align*}
This is a linear variational inequality whose vector field is generated by a $2n \times 2n$ matrix, but the domain of this inequality is a Cartesian product of simplices and a symmetric bimatrix game cannot be obtained in this fashion.} This setting sets the stage for an important positive result on equilibrium computation we derive below.

\subsection{Computing a globally incrementally deployable strategy}

Toward the end of informing the debate on the complexity of computing a Nash equilibrium, the most natural question that naturally emerges in our framing is whether a {\em globally incrementally deployable Nash equilibrium} can be computed efficiently. To that end, note that we already showed that the Nash equilibrium is an {\em undominated} element of an order relation wherein a pair of states are related (in this order relation) if either state is incrementally deployable against the other. The notion of incremental deployability we used in this characterization is that of {\em evolutionary dominance}. That is, we showed that a state is a Nash equilibrium if and only if there is no other state that evolutionarily dominates it. We were, thus, able to characterize Nash equilibria through a (testable) property that they {\em lack}. In the sequel, we consider states that are {\em dominant} in that they evolutionarily dominate every other state. Such states are also equilibrium states, however, elementary intuition suggests that they may be easier to compute as they are defined by a testable property that is {\em conspicuously present} throughout the evolution space. To that end, we give the positive result that a {\em multiplicative weights dynamic} whose {\em learning rate parameter} follows a standard in convex optimization diminishing schedule converges to an evolutionarily dominant state (strategy) of a population game defined over the standard (probability) simplex.

\subsection{Multiplicative weights preliminaries}

Our equilibrium computation setting is, more generally, that of a population game over the probability simplex. Our multiplicative weights algorithm is Hedge \citep{FreundSchapire1, FreundSchapire2} that induces the following map in this setting:
\begin{align}
T_i(X) = X(i) \cdot \frac{\exp\left\{ \alpha E_i \cdot CX \right\}}{ \sum_{j=1}^n X(j) \exp \left\{ \alpha E_j \cdot CX \right\} }, \quad i = 1, \ldots, n.\label{main_exp}
\end{align}
$\alpha$ is called the {\em learning rate}. Let us summarize a few properties here proven in Appendix \ref{preliminary_properties}: First, $T$ never escapes the space of mixed strategies. Actually, if $T$ is started in the relative interior $\mathbb{\mathring{X}}(C)$, it cannot escape $\mathbb{\mathring{X}}(C)$ in a finite number of steps. Second, we characterize the fixed points of $T$ (that is, strategies $X$ such that $T(X) = X$, denoted by $FX(C)$) as strategies that are equalizers of their respective carrier games. This implies that all symmetric equilibrium strategies are fixed points. We also show that fixed points survive inversion of incentives (that is, the fixed points of $C$ coincide with those of $-C$) in contrast to equilibrium strategies that do not survive such inversion. Finally, we show Hedge is a {\em better response dynamic:} Unless $X$ is a fixed point, $(T(X) - X) \cdot CX > 0$. If $C$ is the gradient field of a potential, the better response property implies Hedge is for a small enough learning rate a {\em ascent algorithm} for the potential function (cf. \citep[Chapter 1.2]{Bertsekas}).

\subsection{The convexity lemma of multiplicative weights}

We first prove an important technical lemma used to obtain a wealth of convergence and divergence results in evolution under multiplicative weights. We refer to this result as the ``multiplicative weights convexity lemma'' and note its proof holds under an arbitrary nonlinear operator $C$. A different approach to analyze the dynamics of Hedge is to use {\em the Kantorovich inequality} (together with relative entropy as a potential), but our convexity lemma yields a tighter analysis.

\subsubsection{Relative entropy (or Kullback-Leibler divergence)}

Our analysis of Hedges relies on the relative entropy function between probability distributions (also called {\em Kullback-Leibler divergence}). The relative entropy between the $n \times 1$ probability vectors $P > 0$ (that is, for all $i = 1, \ldots, n$, $P(i) > 0$) and $Q > 0$ is given by 
\begin{align*}
RE(P, Q) \doteq \sum_{i=1}^n P(i) \ln \frac{P(i)}{Q(i)}.
\end{align*}
However, this definition can be relaxed: The relative entropy between $n \times 1$ probability vectors $P$ and $Q$ such that, given $P$, for all $Q \in \{ \mathcal{Q} \in \mathbb{X} | \mathcal{C}(P) \subset \mathcal{C}(\mathcal{Q}) \}$, is
\begin{align*}
RE(P, Q) \doteq \sum_{i \in \mathcal{C}(P)} P(i) \ln \frac{P(i)}{Q(i)}.
\end{align*}
We note the well-known properties of the relative entropy \cite[p.96]{Weibull} that {\em (i)} $RE(P, Q) \geq 0$, {\em (ii)} $RE(P, Q) \geq \| P - Q \|^2$, where $\| \cdot \|$ is the Euclidean distance, {\em (iii)} $RE(P, P) = 0$, and {\em (iv)} $RE(P, Q) = 0$ iff $P = Q$. Note {\em (i)} follows from {\em (ii)} and {\em (iv)} follows from {\em (ii)} and {\em (iii)}.

\subsubsection{The convexity lemma}

The following lemma generalizes \cite[Lemma 2]{FreundSchapire2}.

\begin{lemma}
\label{convexity_lemma}
Let $T$ be as in \eqref{main_exp}. Then
\begin{align*}
\forall X \in \mathbb{\mathring{X}}(C) \mbox{ } \forall Y \in \mathbb{X}(C) : RE(Y, T(X)) \mbox{ is a convex function of }\alpha.
\end{align*}
Furthermore, unless $X$ is a fixed point, $RE(Y, T(X))$ is a strictly convex function of $\alpha$.
\end{lemma}

\begin{proof}
The proof uses calculus---see Appendix \ref{convexity_lemma_proofs}.
\end{proof}

We first use the convexity lemma to prove stability and instability lemmas (bearing relevance to {\em Lyapunov theory}) that although are rather tangential to the main thread of this paper, proving them was instrumental to eventually obtain our results. The crucial to our thread result in this section is Lemma \ref{convexity_lemma_2} as it directly facilitates analyzing Hedge using techniques from convex analysis.

\subsubsection{The stability lemmas}

The next lemmas render Hedge amenable to analysis using Lyapunov theory. These results looked promising at the early stages of our work, but note that our efforts to directly apply Lyapunov's stability theorem or the LaSalle invariance principle to show asymptotic stability of an ESS or GESS in evolution under multiplicative weights (in a fashion analogous to how corresponding results are obtained for the aforementioned replicator dynamic) proved resistant to obtain convergence.

\begin{lemma}
\label{diamonds}
Let $Y, X \in \mathbb{X}(C)$ such that $\mathcal{C}(Y) \subset \mathcal{C}(X)$ and such that $(Y - X) \cdot CX > 0$. Assume $X \not\in FX(C)$. Then there exists $\bar{\alpha} > 0$ (which may depend on $Y$ and $X$) such that
\begin{align*}
\forall \alpha \in (0, \bar{\alpha}) : RE(Y, T(X)) < RE(Y, X).
\end{align*}
\end{lemma}

Under an additional assumption we can show an even stronger property:

\begin{lemma}
\label{cool_hedge_1}
Let $X \in \mathbb{\mathring{X}}(C)$, assume $X \not\in FX(C)$ and let $Y$ be a best response to $X$. Then
\begin{align*}
\forall \alpha > 0 : RE(Y, T(X)) < RE(Y, X).
\end{align*}
\end{lemma}

The proofs of Lemmas \ref{diamonds} and \ref{cool_hedge_1} are in Appendix \ref{convexity_lemma_proofs}.

\subsubsection{An instability lemma}

The following lemma is crucial in deriving divergence results on multiplicative weights in general. We can prove it in two ways, one invoking the aforementioned convexity lemma and the other by simply invoking Jensen's inequality. In the appendix, we show both proofs.

\begin{lemma}
\label{cool_hedge_2}
Let $Y, X \in \mathbb{X}(C)$ such that $\mathcal{C}(Y) \subseteq \mathcal{C}(X)$ and such that $Y \neq X$. If $X$ is not a fixed point, then
\begin{align*}
X \cdot CX - Y \cdot CX \geq 0 \Rightarrow \forall \alpha > 0 : RE(Y, T(X)) - RE(Y, X) > 0.
\end{align*}
\end{lemma}

\begin{proof}
See Appendix \ref{convexity_lemma_proofs}.
\end{proof}

Using the previous lemma, we can show using Lyapunov (in)stability theory that the uniform strategy in rock-paper-scissors is {\em repelling} under Hedge. However, using the techniques we develop in this paper, we can show that the empirical average of the strategies generated by Hedge converges to the equilibrium of rock-paper-scissors (and every symmetric zero-sum game).

\subsection{A version of the convexity lemma}

We use the following ``secant inequality'' for a convex function $F(\cdot)$ and its derivative $F'(\cdot)$:
\begin{align*}
\forall \mbox{ } b > a : F'(a) \leq \frac{F(b) - F(a)}{b - a} \leq F'(b).
\end{align*}

The following lemma is an analogue of \cite[Lemma 8.2.1, p. 471]{ConvexAnalysis}.

\begin{lemma}
\label{convexity_lemma_2}
Let $C \in \mathbb{\hat{C}}$. Then, for all $Y \in \mathbb{X}(C)$ and for all $X \in \mathbb{\mathring{X}}(C)$, we have that
\begin{align*}
\forall \alpha > 0 : RE(Y, T(X)) \leq RE(Y, X) - \alpha (Y-X) \cdot CX + \alpha (\exp\{\alpha\} - 1) \bar{C},
\end{align*}
where $\bar{C} > 0$ can be chosen independent of $X$ and $Y$.
\end{lemma}

\begin{proof}
Since, by Lemma \ref{convexity_lemma}, $RE(Y, T(X)) - RE(Y, X)$ is a convex function of $\alpha$, we have by the aforementioned secant inequality that, for $\alpha > 0$,
\begin{align}
RE(Y, T(X)) - RE(Y, X) \leq \alpha \left( RE(Y, T(X)) - RE(Y, X) \right)' = \alpha \cdot \frac{d}{d \alpha} RE(Y, T(X)).\label{ooone}
\end{align}
Straight calculus (cf. Lemma \ref{convexity_lemma}) implies that
\begin{align*}
\frac{d}{d \alpha} RE(Y, T(X)) = \frac{\sum_{j = 1}^n X(j) (CX)_j \exp\{ \alpha (CX)_j \}}{\sum_{j = 1}^n X(j) \exp\{ \alpha (CX)_j \}} - Y \cdot CX.
\end{align*}
Using Jensen's inequality in the previous expression, we obtain 
\begin{align}
\frac{d}{d \alpha} RE(Y, T(X)) \leq \frac{\sum_{j = 1}^n X(j) (CX)_j \exp\{ \alpha (CX)_j \}}{\exp\{ \alpha X \cdot CX \}} - Y \cdot CX.\label{vbvbvb}
\end{align}
Note now that
\begin{align}
\exp\{ \alpha x \} \leq 1 + (\exp\{ \alpha \} - 1) x, x \in [0, 1],\label{freund_schapire_in}
\end{align}
an inequality used in \cite[Lemma 2]{FreundSchapire2}. Using $C \in \mathbb{\hat{C}}$, \eqref{vbvbvb} and \eqref{freund_schapire_in} imply that
\begin{align*}
\frac{d}{d \alpha} RE(Y, T(X)) \leq \frac{X \cdot CX}{\exp\{ \alpha X \cdot CX \}} - Y \cdot CX + \left( \exp\{ \alpha \} - 1 \right) \frac{\sum_{j=1}^n X(j) (CX)_j^2}{\exp\{ \alpha X \cdot CX \}}
\end{align*}
and since $\exp\{\alpha X \cdot CX\} \geq 1$ (again by the assumption that $C \in \mathbb{\hat{C}}$), we have
\begin{align*}
\frac{d}{d \alpha} RE(Y, T(X)) \leq X \cdot CX - Y \cdot CX + \left( \exp\{ \alpha \} - 1 \right) \sum_{j=1}^n X(j) (CX)_j^2.
\end{align*}
Choosing $\bar{C} = \max \left\{\sum X(j) (CX)_j^2 \right\}$ and combining with \eqref{ooone} yields the lemma.
\end{proof}

\subsection{Asymptotic convergence to an evolutionarily dominant equilibrium}

In the next theorem we show that under the technical conditions on the learning rate given below, Hedge asymptotically converges to the EDset (in a single epoch), as shown in the following theorem. We note $C$ is a general (continuous) nonlinear operator.

\begin{theorem}
\label{yellow_gess}
Let $C \in \mathbb{\hat{C}}$ and assume $C$ is equipped with an EDset, say, $\mathbb{X}^*$. Furthermore, let $X^k \equiv T^k(X^0)$, where $X^0 \in \mathbb{\mathring{X}}(C)$ and assume the learning rate $\alpha_k > 0$ is chosen such that
\begin{align}
\lim_{k \rightarrow \infty} \alpha_k = 0 \mbox{ and } \sum_{k = 0}^{\infty} \alpha_k = +\infty.\label{lrate_assumption}
\end{align}
If every strategy in $\mathbb{X}^*$ is an EDS (for example, if the EDset consists of a singleton GESS or if the operator $C$ is monotone), then the sequence of iterates $\{ X^k \}$ converges to an EDS.
\end{theorem}

In the proof of Theorem \ref{yellow_gess} we first need a lemma. 

\begin{lemma}
\label{yellow_infimum}
Let $C \in \mathbb{\hat{C}}$ and $X^k \equiv T^k(X^0)$, where $X^0 \in \mathbb{\mathring{X}}(C)$. Assume the learning rate $\alpha_k > 0$ is chosen from round to round such that
\begin{align*}
\lim_{k \rightarrow \infty} \alpha_k = 0 \mbox{ and } \sum_{k = 0}^{\infty} \alpha_k = +\infty. 
\end{align*}
Then
\begin{align}
\forall Y \in \mathbb{X}(C) : \liminf\limits_{k \rightarrow \infty} (Y - X^k) \cdot CX^k \leq 0.\label{pepper2}
\end{align}
\end{lemma}

\begin{proof}
Let $Y \in \mathbb{X}(C)$ be arbitrary and assume that \eqref{pepper2} does not hold, that is,
\begin{align}
\liminf\limits_{k \rightarrow \infty} (Y - X^k) \cdot CX^k > 0.\label{not_true_assumption}
\end{align}
Then
\begin{align*}
\exists \theta > 0 \mbox{ } \exists K \geq 0 \mbox{ } \forall k \geq K : (Y - X^k) \cdot CX^k \geq 3 \theta.
\end{align*}
Continuity further implies that
\begin{align}
\exists \hat{Y} \in \mathbb{X}(C) \mbox{ } \forall k \geq K : (\hat{Y} - X^k) \cdot CX^k \geq 2\theta.\label{pool2}
\end{align}
Using Lemma \ref{convexity_lemma_2}, we obtain that
\begin{align}
RE(\hat{Y}, X^{k+1}) \leq RE(\hat{Y}, X^k) - \alpha_k (\hat{Y} - X^k) \cdot CX^k + \alpha_k (\exp\{\alpha_k\} - 1) \bar{C}.\label{party2}
\end{align}
Combining \eqref{pool2} and \eqref{party2} yields
\begin{align*}
\forall k \geq K : RE(\hat{Y}, X^{k+1}) \leq RE(\hat{Y}, X^k) - \alpha_k (2\theta - (\exp\{\alpha_k\} - 1) \bar{C}).
\end{align*}
Since $\alpha_k \rightarrow 0$, 
\begin{align*}
\exists K' \geq K \mbox{ } \forall k \geq K' : \theta > (\exp\{\alpha_k\}-1) \bar{C}.
\end{align*}
This implies that
\begin{align*}
\forall k \geq K' : RE(\hat{Y}, X^{k+1}) \leq RE(\hat{Y}, X^k) - \alpha_k \theta.
\end{align*}
Summing over $k = K', \ldots, K''$, we obtain that
\begin{align*}
RE(\hat{Y}, X^{K''}) \leq RE(\hat{Y}, X^{K'}) - \theta \sum_{k = K'}^{K''} \alpha_k.
\end{align*}
Letting $K'' \rightarrow \infty$ and using the assumption $\sum \alpha_k = + \infty$, we obtain that the right-hand-side tends to $-\infty$ whereas the left-hand-side is positive. Therefore, \eqref{not_true_assumption} cannot be true.
\end{proof}

\begin{proof}[Proof of Theorem \ref{yellow_gess}]
Let $X^*$ be an EDS. Using Lemma \ref{convexity_lemma_2}, we obtain that
\begin{align*}
RE(X^*, X^{k+1}) &\leq RE(X^*, X^k) - \alpha_k (X^* - X^k) \cdot CX^k + \alpha_k (\exp\{\alpha_k\} - 1) \bar{C}\\
  &\leq RE(X^*, X^k) + \alpha_k (\exp\{\alpha_k\} - 1) \bar{C}
\end{align*}
where in the last inequality we used that $X^*$ is dominant. Summing the previous inequalities over $k = \hat{k}, \ldots, K$ for some arbitrary $\hat{k}$ and $K$ with $\hat{k} < K$, we obtain
\begin{align}
RE(X^*, X^K) \leq RE(X^*, X^{\hat{k}}) + \bar{C} \sum_{k=\hat{k}}^K \alpha_k (\exp\{\alpha_k\} - 1).\label{ytrty}
\end{align}
Letting $Y = X^*$ in inequality \eqref{pepper2} in Lemma \ref{yellow_infimum}, we obtain that
\begin{align*}
\liminf\limits_{k \rightarrow \infty} (X^* - X^k) \cdot CX^k = 0.
\end{align*}
Therefore, there is an accumulation point of $\{X^k\}$ that is an EDS. Let $\{X^{k_i}\}$ be a subsequence such that $X^{k_i} \rightarrow \hat{X}^*$, where $\hat{X}^*$ is such an EDS. By setting $X^* = \hat{X}^*$ and $\hat{k} = k_i$ in \eqref{ytrty} and by letting $i \rightarrow \infty$, we obtain
\begin{align*}
\limsup\limits_{K \rightarrow \infty} RE(\hat{X}^*, X^K) \leq \lim_{i \rightarrow \infty} RE(\hat{X}^*, X^{\hat{k}_i}) + \lim_{i \rightarrow \infty}  \bar{C} \sum_{k=\hat{k}_i}^{\infty} \alpha_k (\exp\{\alpha_k\} - 1) = 0,
\end{align*}
since 
\begin{align*}
\lim_{i \rightarrow \infty} RE(X^*, X^{\hat{k}_i}) = 0
\end{align*} 
and 
\begin{align*}
\lim_{i \rightarrow \infty} \sum_{k=\hat{k}_i}^{\infty} \alpha_k (\exp\{\alpha_k\} - 1) = 0,
\end{align*}
by the assumption $\lim_{k \rightarrow \infty} \alpha_k = 0$. Thus, the whose sequence converges to $\hat{X}^*$.
\end{proof}

Our technique can be combined with {\em simplicial decomposition methods} \citep[Chapter 4.2]{ConvexOptimizationAlg} to yield algorithms for computing an EDS in convex optimization problems when the constraint set is a convex set more general than a simplex (by expanding a simplex inscribed in the constraint set). Generalizing this approach to non-convex problems is a promising direction for future work.

\subsection{Evidence that the class {\bf PPAD} may not be hard}

Can we use Hedge to compute (approximate) a symmetric equilibrium in a symmetric bimatrix game? Our numerical experiments in randomly generated symmetric bimatrix games suggest that this is possible: Our simulations always converge (to an equilibrium). Let us first show that if Hedge converges, then the limit point is necessarily an equilibrium. To that end we show that non-equilibrium fixed point are {\em repelling} in that Hedge diverges away from such fixed points.

\subsubsection{Repelling fixed points}

Recall that $X^*$ is a fixed point of $T$ if $T(X^*) = X^*$. We may assume $T$ is as in \eqref{main_exp}. We are interested in the behavior of $T$ starting from an interior strategy of the probability simplex.

\begin{definition}
Let $T : \mathbb{X} \subset \mathbb{R}^n \rightarrow \mathbb{X}$ such that $X^*$ is a fixed point of $T$. $X^*$ is {\em repelling} under $T$ if, for all $X_0$ in the relative interior of $\mathbb{X}$,
\begin{align*}
\lim_{k \rightarrow \infty} T^k(X_0) \neq X^*.
\end{align*}
\end{definition}

\begin{lemma}
\label{repelling_lemma}
Suppose $X^*$ is a fixed point of $T$. If $\dot{V}$ and $V$ are positive definite with respect to $T$ and $X^*$, $X^*$ is repelling under $T$.
\end{lemma}

\begin{proof}
Suppose there exists $X_0$ such that $T^k(X_0)$ converges to $X^*$. Since $V$ is continuous, $V(T^k(X_0))$ also converges, in fact, to $0$. Consider any $k$ large enough that such $T^k(X_0)$ is close to $X^*$. Since, for all $X$ in a neighborhood of $X^*$, $\dot{V}(X) > 0$, for all $\hat{k} > k$, $V(T^{\hat{k}}(X_0)) > V(T^k(X_0))$, which is a positive constant bounded away from $0$, contradicting convergence of $V(T^k(X_0))$ to $0$.
\end{proof}

\subsubsection{Non-equilibrium fixed points are repelling}

Starting off with preliminaries, the following propositions are standard \citep{Econ_analysis}.

\begin{proposition}
Let $f : \mathbb{O} \rightarrow \mathbb{O}'$ be a continuous function between the topological spaces $\mathbb{O}$ and $\mathbb{O}'$. Then if $\mathbb{O}$ is connected, $f(\mathbb{O}) \subset \mathbb{O}'$ is connected.
\end{proposition}

The following proposition is known as the {\em intermediate value theorem}.

\begin{proposition}
If $f : |a, b| \subset \mathbb{R} \rightarrow \mathbb{R}$, where $|a, b|$ is an interval, is continuous, $y, y'' \in f(|a, b|)$ and $y < y' < y''$, then there exists $x \in |a, b|$ such that $f(x) = y'$.
\end{proposition}

A slightly more general version can be stated as follows.

\begin{proposition}
\label{gv_int_v_theorem}
Suppose $\mathbb{O}$ is connected and $f : \mathbb{O} \rightarrow \mathbb{R}$ is continuous. If $a, b \in \mathbb{O}$ and $f(a) < y < f(b)$, there exists $x \in \mathbb{O}$ such that $f(x) = y$. 
\end{proposition}

The following lemma is a straightforward implication of the previous propositions.

\begin{lemma}
\label{fp_instability_lemma}
Let $f : O \subset \mathbb{R}^n \rightarrow \mathbb{R}$ be a continuous map such $O$ is a neighborhood of $X^* \in \mathbb{R}^n$. If $f(X^*) > 0$, there exists a neighborhood $O' \subset O$ of $X^*$ such that, for all $X \in O'$, $f(X) > 0$.
\end{lemma}

\begin{proof}
If, for all $X \in O$, $f(X) \neq 0$, the lemma is trivially true by letting $O' = O$. Suppose, therefore, there exists $X \in O$ such that $f(X) = 0$ and let
\begin{align*}
\mathbb{F} = \{X \in O | f(X) > 0 \}.
\end{align*}
$\mathbb{F}$ is an open set containing $X^*$ and, therefore, there exists a neighborhood $O' \subset \mathbb{F}$ of $X^*$ such that, for all $X \in O'$, $f(X) > 0$. Since $f$ is continuous and $O'$ is connected, $f(O')$ is connected, and the intermediate value theorem implies that, for all $X \in O'$, $f(X) > 0$. For if $0 \in f(O')$, then, since $f(O')$ is open (since $O'$ is open and $f$ is continuous) and connected, there exists $X \in O'$ such that $f(X) < 0$, and Proposition \ref{gv_int_v_theorem} implies $\exists Y \in O' : f(Y) = 0$, contradicting the existence of $O'$.
\end{proof}

\begin{lemma}
\label{yannis_lemma}
Non-equilibrium fixed points are repelling under \eqref{main_exp} for all values of the learning rate.
\end{lemma}

\begin{proof}
Let $X^*$ be a non-equilibrium fixed point of \eqref{main_exp} and
\begin{align*}
i \in \arg\max \{ E_j \cdot CX^* | j \in \mathcal{K} \}.
\end{align*}
By the assumption $X^* \not\in NE^+(C, C^T)$, $E_i \cdot CX^* > X^* \cdot CX^*$. Now let
\begin{align*}
V(X) = X(i) - X^*(i) = X(i)
\end{align*}
since $X^*(i) = 0$. Note $V(X^*) = 0$. Letting $\hat{X} = T(X)$, where $T$ is as in \eqref{main_exp},
\begin{align*}
\dot{V}(X) &= V(\hat{X}) - V(X)\\
  &= \hat{X}(i) - X(i)\\
  &= X(i) \cdot \frac{\exp\left\{ \alpha E_i \cdot CX \right\}}{ \sum_{j=1}^n X(j) \exp \left\{ \alpha E_j \cdot CX \right\} } - X(i).
\end{align*} 
Let 
\begin{align*}
f(\alpha) \equiv \frac{\exp\left\{ \alpha E_{i} \cdot CX \right\}}{ \sum_{j=1}^n X(j) \exp \left\{ \alpha E_j \cdot CX \right\} } \equiv \frac{g(\alpha)}{h(\alpha)}.
\end{align*}
Letting $(CX)_j \equiv E_j \cdot CX$, we have
\begin{align}
\frac{d f}{d \alpha} = \frac{ (CX)_i g(\alpha) h(\alpha) - g(\alpha) \left( \sum_{j=1}^n X(j) (CX)_j \exp \left\{ \alpha (CX)_j \right\} \right) }{ h^2(\alpha)},\label{cutie_eq}
\end{align}
which implies
\begin{align*}
\left. \frac{df}{d \alpha} \right|_{\alpha = 0} = E_{i} \cdot CX - X \cdot CX.
\end{align*}
Therefore, since $E_i \cdot CX^* > X^* \cdot CX^*$, Lemma \ref{fp_instability_lemma} implies
\begin{align*}
\exists O \mbox{ } \forall X \in O/\{X^*\} : E_{i} \cdot CX > X \cdot CX.
\end{align*}
where $O$ is a neighborhood of $X^*$. Therefore, for all $X \in O$,
\begin{align*}
\left. \frac{d f}{d \alpha} \right|_{\alpha = 0} > 0.
\end{align*} 
To complete the proof we need to show the latter property holds for all $\alpha > 0$. But this is a simple implication of choice of $i$ as a best response: Since
\begin{align*}
\sum_{j=1}^n X(j) (CX)_j \exp \left\{ \alpha (CX)_j \right\} < (CX)_i \sum_{j=1}^n X(j) \exp \left\{ \alpha (CX)_j \right\} = (CX)_i h(\alpha),
\end{align*}
\eqref{cutie_eq} gives
\begin{align*}
\forall \alpha > 0 : \frac{df}{d \alpha} > 0.
\end{align*}
Therefore, 
\begin{align*}
\forall X \in O \mbox{ and } \forall \alpha > 0 : \dot{V}(X) > 0
\end{align*}
and Lemma \ref{repelling_lemma} completes the proof. 
\end{proof}

We believe that the previous lemma can also be proven using Lemma \ref{yellow_infimum} instead.

\subsubsection{The equilibrium strategy of rock-paper-scissors is also repelling}

Our analysis begs the question whether there are equilibria that are also repelling. The answer is affirmative and the interior equilibrium of rock-paper-scissors (the uniform strategy) is an example. Rock-paper-scissors is a zero-sum symmetric bimatrix game with payoff matrix
\begin{align*}
C = \left(\begin{array}{ r  r  r }
0 & -1 & 1\\
1 & 0 & -1\\
-1 & 1 & 0\\
\end{array}\right).
\end{align*}
$C$ is {\em skew-symmetric,} that is, $C^T = -C$, therefore, for all $X$, $X \cdot CX = 0$. Furthermore, $X^* = (1/3, 1/3, 1/3)^T$ is the unique equilibrium strategy, implying after straight algebra that, for all $X \in \mathbb{X}(C)$, $(X^* - X) \cdot CX = 0$. Using this latter property, Lemma \ref{cool_hedge_2} and our previous stronger version of the standard Lyapunov instability result imply $X^*$ cannot be a limit point in evolution under \eqref{main_exp}. Thus, since $X^*$ is the unique equilibrium strategy, there exist symmetric bimatrix games such that \eqref{main_exp} does not converge to an equilibrium strategy. However, the empirical average of strategies iterative applications of Hedge generate in zero-sum games (such as the rock-paper-scissors) converge to the minimax strategy. It is then natural to ask what the computational limits of this approach are in general symmetric bimatrix games. We leave this question as future work.

\section{Limitations of BGP (Border Gateway Protocol) routing}
\label{BGP_section}

One of the main motivations in this paper is to contribute to the understanding of the Internet's architecture and to the design of mechanisms that improve its properties. In this section, we focus on the limitations of IP routing that, unfortunately, is not based on a sound computational foundation. The Internet systems research community has been working on a variety of systems designs whose deployment would improve on the properties of the global routing system. In this paper, we do not aim to address IP routing per se but, in the following sections, we discuss mechanisms that can be used to replace the incumbent IP routing architecture with improved emergent systems designs.

\subsection{BGP routing is vulnerable to routing instabilities}

The software system that determines the network paths {\em data packets} traverse in IP (Internet Protocol) networks (from source to destination) is called {\em routing}. The routing infrastructure is divided into {\em administrative domains} called {\em autonomous systems}. Each autonomous system belongs to a single organization (e.g., an ISP (Internet Service Provider)). Each autonomous system has an independent routing infrastructure, which consists of hardware and software. The hardware infrastructure consists of {\em routers} and {\em links} whereas the software infrastructure implements {\em routing protocols} that autonomous systems use to coordinate in order to provide global connectivity. 

Routing protocols are divided into two classes, namely, {\em intradomain routing protocols} that determine network paths within an autonomous system network and {\em interdomain routing protocols} that determine network paths across autonomous systems to provide global connectivity. There is currently one interdomain routing protocol, the Border Gateway Protocol (BGP). BGP factors autonomous system {\em routing policies} in computing interdomain paths. There are no restrictions on the policies an autonomous system can implement using BGP (but there are standard practices).

Drawing on \cite{Griffin}, \cite{Fabrikant2} show that the computations performed by the BGP protocol (toward computing interdomain paths) can be modeled as an algorithm converging to a sink equilibrium in a game where autonomous systems are players, strategies are the autonomous systems one hop away from the corresponding player, and payoffs are determined according to each player's preferences across routing system configurations (as they are expressed in routing policies). \cite{Fabrikant2} further show that the problem of determining if the space of possible routing system configurations (as that is determined by the autonomous system policies) does not have singleton sink equilibria (pure Nash equilibria) is {\bf PSPACE}-complete. Note that a non-singleton sink equilibrium (consisting of more than one states) implies that the interdomain routing system will oscillate and such oscillations have been observed in practice. That the interdomain system is mostly stable can be attributed to {\em best common practices} in routing policies \citep{Rexford}. But there is no rule that prevents more general policies from being implemented. What's more this is not the end of the story:

\subsection{BGP routing is vulnerable to malicious attack}

BGP is also vulnerable to attacks where adversarial networks, for example, announce routes for IP address blocks (prefixes) they do not own. Such routing announcements are called {\em prefix hijacking attacks}. In fact, prefix hijacking is so easy it happens by accident. The consequences of arbitrary (and malicious) routing behavior, such as prefix hijacking and other forms of bogus routes announcements, are serious because the packets destined to the victim prefix are instead delivered to the adversary, who may drop the traffic, impersonate the destination, modify the payload, or snoop on the communication. The best way to defend against prefix hijacking is the subject of much debate. The role of secure routing protocols, in particular, has received considerable attention. One of the earliest proposals to secure interdomain routing is the Secure-BGP protocol \citep{Kent}, which cryptographically protects routing announcements and protects not only against prefix hijacking but also against more sophisticated mechanisms that try to manipulate the global computation of interdomain routes by BGP. Unfortunately, Secure-BGP is not incrementally deployable and despite significant efforts from the industrial sector and (especially) the US government, Secure-BGP has not received deployment traction. We come back to the question of securing BGP after discussing a complementary aspect of IP networks, namely, {\em congestion control}.

\section{Equilibrium computation in TCP/IP congestion control}
\label{TCP_section}

Our inquiry into mechanisms to address the limitations of Internet routing (discussed in the sequel) benefits from the study of TCP/IP congestion control as an equilibrium computation algorithm: \cite{Internet-Game-Theory} in his influential paper prompting theoretical computer scientists to inquire on the computational foundations of the Internet poses as an open problem the understanding of congestion control using game theory as the tool. In particular, he asks: {\em ``Of which game is TCP/IP congestion control the Nash equilibrium?''} In this section, we give an answer to this question.

\subsection{TCP/IP congestion control basics}

The Internet infrastructure consists of communication links and computing devices (such as switches, routers, and hosts) as well as code (such as the TCP/IP architecture) to render this infrastructure serviceable to hosts (hosts are, for example, mobile computing devises or enduser desktops running application software). To transfer data, pairs of hosts establish {\em transport sessions}. Each transport session has a source, a destination, and a path in the network (a sequence of routers and links). Each link has a capacity (a {\em bitrate}), which must be shared among the competing transport sessions crossing that link. One of the primary functions of TCP (the {\em transport layer} on top of IP) is to control congestion so that communication links are not overwhelmed by communication sessions between endusers. To that end, TCP stipulates that endusers should slow down their transmission rates on detecting congestion. Endusers follow this constraint although they are not obliged to do so and in spite of the fact that they would, perhaps temporarily, benefit by not conforming.

TCP/IP congestion control can be modeled as an algorithm (operating in a distribution fashion across the Internet) for allocating {\em bitrates} to {\em transport sessions}. From a theoretical standpoint, the rates being allocated can be understood as a solution to an optimization problem, namely, the {\em rate allocation problem:} Given a network, the objective is to allocate bitrates (for all transport sessions) that are feasible (in that the capacity constraints are satisfied) so that the allocation is {\em efficient} (in that network capacity is not wasted) and {\em fair} (in that transport sessions are treated ``equally''). 

The allocation algorithms receiving most attention in the literature are the {\em max-min fair} (for example, see \citep{Bertsekas-Gallager}) and {\em proportionally fair} ones (for example, see \citep{Kelly1}). In a max-min fair allocation, priority is given to the most poorly treated sessions (thus maximizing the minimal bitrate allocation) subject to the rule that all sessions are entitled to the same share of their corresponding bottleneck link and that excess capacity is shared fairly among other sessions. Proportionally fair allocations approximate max-min fair ones (through violating the constraint that poorly treated flows should have absolute priority in the allocation process).

Proportionally fair rate allocation algorithms can be understood as solving a convex optimization problem whose objective function is {\em aggregate utility} (that is, the sum of the utilities of transport sessions) subject to capacity constraints. TCP/IP congestion control is a distributed (``primal-dual'') algorithm that induces a rate allocation, which is approximately proportionally fair \citep{Low}. In this perspective, TCP/IP congestion control is an algorithm for optimizing {\em social welfare} in that it solves a global optimization problem with aggregate utility in the objective function. It is, therefore, natural to ask: What makes such an algorithm stable in a game-theoretic sense?

\subsection{TCP/IP congestion control as an institution}

The previous view of congestion control as a technology in which the Internet has the structure of a distributed system provides, however, only a partial, if not misleading, view of the Internet's nature. The Internet is unlike any other technology we have produced (if it is a technology at all) but rather an ``ecosystem'' that emerges as the collective outcome of choices made by {\em population systems} of equipment manufacturers, service providers, content providers, and endusers toward the goal of global connectivity. Population members' choices are inexorably shaped not only by technological feasibility but also by the {\em incentive structure} that emerges as these population members interact. 

In a perspective that deviates from the dominant way of looking at Internet systems engineering, we posit that TCP/IP congestion control is an {\em institution}. \cite{North} defines institutions as ``the rules of the game in a society or, more formally ... [as] the humanly devised constraints that shape human interaction. In consequence they structure incentives in human exchange whether political, social, or economic.'' That TCP/IP sessions slow down their transmission rates upon detecting congestion is a humanly devised constraint implemented as software in operating systems and forms an essential part of every enduser device or content provider server. Although manually changing this software falls outside the expertise of a typical enduser, such changes fall squarely within the purview of operating systems manufacturers (such as Microsoft and Apple) and content providers (such as Google): That the lines of code implementing TCP/IP congestion control remain intact is a technical decision influenced by the incentives that shape the Internet infrastructure.\footnote{For other examples of institutions that complement and are complemented by the Internet's technical design see \url{http://www.ted.com/talks/jonathan_zittrain_the_web_is_a_random_act_of_kindness}.} 

TCP/IP congestion control was designed and deployed in the Internet in response to a congestion collapse in 1986 \citep{Jacobson} and has remained in stable operation since then. The stability of this mechanism has puzzled researchers and systems engineers since the mechanism's inception. Worried that the possible emergence of aggressive behavior in TCP sessions would jeopardize the stability of the Internet at large, various researchers in the networking systems community proposed mechanisms to counterbalance aggression using software (protocol-based) mechanisms (for example, see \citep{Floyd}). However, these mechanisms have not been used in practice. 

The scepticism on the means by which TCP/IP congestion control is attained is not confined to the Internet architecture. Economics researchers have been similarly puzzled by institutions that emerge concerning the management of {\em common pool resources} \citep{Ostrom}. To a large degree, TCP/IP congestion control is a global-scale institution used to manage Internet bandwidth as a common pool resource. But how can we understand such incentive structure mathematically?

\subsection{TCP/IP congestion control as an equilibrium}

In the rest of this section, we identify a mathematical game realistically modeling TCP congestion control in a simple setting and an equilibrium in this game such that if players use the corresponding equilibrium strategies, the outcome is consistent with the empirical outcome of transport sessions using TCP universally in the Internet. This analysis is the first attempt to model rate allocation as a problem of technological competition between TCP and a more aggressive variant, and, although we will focus on one particular (but realistic) scenario, we will be able to conclude that a general theory of TCP's ``architectural stability'' cannot but be based on {\em conditionally cooperative strategies}. The main idea comes out in the elementary setting of two TCP sessions sharing one link.

\subsubsection{Rate allocation as a $2$-player prisoner's dilemma}

Consider as link that has bandwidth of 100 units and serves flows $a$ and $b$ (as shown in Fig. \ref{saldkjfnalxdkhjf}). Flows correspond to the players. Either player may choose to cooperate ($C$) or defect ($D$). If both cooperate, then the bandwidth is split evenly among them, and each gets a share of 50 units. If either player uses $C$ and the other uses $D$, the former receives 90 units of bandwidth whereas the latter only 10, and if both use $D$, because of congestion collapse, each receives 15 units. This interaction can be represented as the two-person prisoner's dilemma game in Fig. \ref{saldkjfnalxdkhjf}. The unique equilibrium in this game is for both players to choose $D$, that is, to obviate congestion control. To sustain the cooperative equilibrium in this model we need to consider repeated interactions in an {\em infinite} time horizon giving rise to an incentive structure that is called a {\em supergame}. (It is well-known that cooperation cannot emerge under an a priori determined finite time horizon.)

\begin{figure}[tb]
\centering
\includegraphics[width=7.5cm]{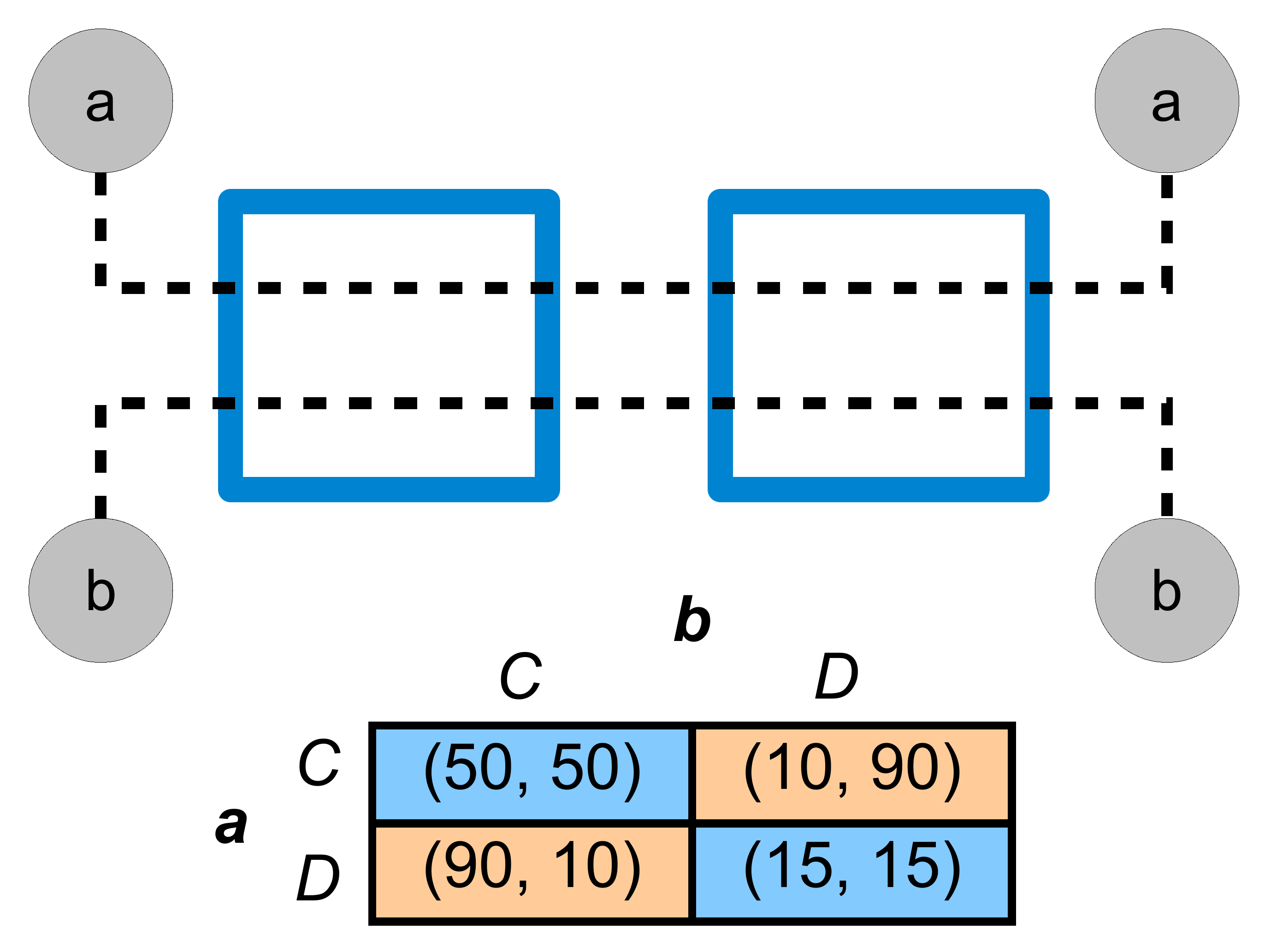}
\caption{\label{saldkjfnalxdkhjf}
Two flows crossing a link and one stage of the corresponding supergame.}
\end{figure}

\subsubsection{Conditional cooperation as an equilibrium in repeated interactions}

Supergames extend normal-form games from a single to multiple (repeated) interactions. Players interact in stages that correspond to identical (with respect to the incentive structure) repetitions of a basis game in normal form, called a {\em stage game}, under an {\em infinite time horizon} (that is, the stages extend to infinity). The strategies of a stage game are called {\em actions} in the supergame. Then a strategy (of the supergame) becomes an infinite sequence of actions (one for each stage of the game). Each action can depend on the sequence of play observed in previous time periods. In a supergame, the payoffs that combinations of strategies yield to each corresponding player are obtained as {\em discounted infinite sums} of the payoffs that actions yield (as determined by the stage game). The discount factor weighs in the significance of future payoffs to the present action.

Drawing on the analysis of prisoner's dilemma supergames by \citep{Axelrod}, we note that the incentive structure of the $2$-player prisoner's dilemma corresponds to a game-theoretic version of the {\em rate allocation problem} where a pair of similar (in the content they transfer) transport sessions of unknown duration simultaneously compete for the bandwidth of a single (bottleneck) link. These sessions correspond to the players in our supergame. At each stage game, each player has a choice of either using TCP/IP congestion control (action $C$) or a custom transport protocol that obviates congestion control (action $D$). A strategy of our supergame is a sequence of actions where each action is a function mapping the history of play onto the action space $\{ C, D \}$ of the stage game. We denote the strategy of always cooperating by $C^{\infty}$ and that of always defecting by $D^{\infty}$.

The empirical outcome of play in the Internet can be approximated as {\em universal cooperation} in that all players cooperate in all stage games. However, it is easy to show that universal {\em unconditional} cooperation corresponding to the strategy profile $(C^{\infty}, \ldots, C^{\infty})$ is {\em not} an equilibrium. To explain the empirical outcome we must, therefore, consider {\em conditionally cooperative strategies}, that is, strategies where cooperation is contingent upon previous outcomes. The most prominent of these strategies is {\em tit-for-tat,} which is to cooperate in the first stage game and to choose the opponent's previous strategy in succeeding games. \cite{Axelrod} shows that tit-for-tat is an equilibrium in the $2$-player prisoner's dilemma supergame provided the discount factor is large enough. Our toy theoretical construct is, thus, able to explain TCP/IP congestion control being a game-theoretic equilibrium: Each player is content being cooperative provided the other player cooperates.

\subsubsection{Rate allocation as a multi-player prisoner's dilemma}

The conclusions of the previous analysis carry over to the setting of {\em symmetric prisoner's dilemma supergames} \citep{Taylor} which we can use as a model of the interaction of multiple transport sessions sharing a single (bottleneck) link. (The assumption that sessions compete for the bandwidth of one bottleneck link is standard in the analysis of TCP/IP congestion control protocols.) 

The supergame is symmetric in that at each stage game a player's payoff depends only on the player's own choice and the number of other players choosing $C$. To assume that payoffs are symmetric in the profile that all players cooperate is justified by the assumption of proportional fairness (which, as mentioned earlier, is a property TCP/IP congestion control satisfies). To assume that payoffs are symmetric among defecting players is justifiable owing to all such players using the same transport protocol (and there is no reason to believe that congestion collapse would favor one defecting player over another). Let $f_k$ be a player's payoff if the player chooses $C$ and $k$ others choose $C$ and let $g_k$ be the corresponding payoff if the player chooses $D$ and $k$ others choose $C$. It is assumed that: (1) $\forall k \geq 0, g_k > f_k$, (2) $f_{N-1} > g_0$, and (3) $\forall k > 0, g_k > g_0$. The first assumption captures that the choice to defect (i.e., to obviate congestion control) dominates cooperation. Indeed the first player to defect seizes a disproportionate fraction of the bandwidth, which he has to share with other defecting players at a much lower efficiency level than if everyone cooperated; that the efficiency level of a system that obviates congestion control is lower is attributed to the ensuing congestion collapse and is captured in the second assumption. The third assumption captures the fact that the payoff of using the aggressive transport protocol decreases as more players defect.

A natural generalization of tit-for-tat in corresponding $N$-person games is to cooperate in the first game, and then cooperate if and only if at least $n$ players cooperated in the preceding game. Call this strategy $B_n$. If every player uses $B_n$, then the supergame's outcome is indeed universal cooperation, however, is $(B_n, \ldots, B_n)$ an equilibrium? \cite{Taylor} answers this question in the affirmative provided that $n = N-1$ and the discount factor exceeds a certain lower bound. This analysis implies that {\em conditional} use of TCP/IP congestion control by all players is an equilibrium. Cooperation is conditional in that players are eager to obviate congestion control to penalize defectors, implying defection does not pay off in the long run. Since {\em unconditional cooperation} can be exploited by defectors and is, therefore, not an equilibrium, we may conclude that TCP?IP congestion control's architectural stability is imputable to players' having a {\em theory of mind}.

What's more, tit-for-tat is, to the extent of our knowledge, the {\em simplest} of all strategies that can sustain cooperation as an equilibrium, which implies that the Occam's razor singles out this strategy as the most likely explanation of the behavioral phenomenon we are trying to understand. We, may, thus, postulate that Internet players, and, in particular, content providers, do not defect from TCP/IP congestion control {\em thinking} that defection will be counteracted with retaliation from other content providers. Papadimitrou was, therefore, in part, correct in predicting that ``If we see Internet congestion control as a game, we can be sure that its equilibrium is not achieved by rational contemplation, but by interaction and adaptation \ldots'' as tit-for-tat is indeed an adaptive strategy, however, we saw that it is rational contemplation averting players from defecting.

\subsubsection{Conditional cooperation as a focal equilibrium}

The conditionally cooperative equilibrium that emerges in the prisoner's dilemma supergame is not unique; for example, {\em unconditional defection} corresponding to the strategy profile $(D^{\infty}, \ldots, D^{\infty})$ is also an equilibrium. It is then natural to ask why the Internet population selects the cooperative outcome. We believe TCP/IP congestion control is a {\em focal equilibrium} \citep{Myerson, Schelling}, that is, an equilibrium conspicuously distinctive from other equilibria and, therefore, one that everyone expects to be the outcome, and, hence fulfill. There are many reasons for this: 
\begin{itemize}

\item TCP has been in stable operation for over 30 years, and it is known to work. 

\item TCP is efficient, maximizing aggregate social utility, and {\em equitable}, providing a level playing field: Content providers would have abandoned it had it not facilitated fair competition. 

\item The equilibrium of universal defection leads to congestion collapse, making it conspicuously unattractive; the Internet population would prefer to steer away from such an equilibrium.

\end{itemize}

\section{Evolutionary approaches for revolutionary change}
\label{equilibrium_transitions}

Roughly fifteen years ago, Jennifer Rexford sketched an influential to my thinking blueprint for effecting revolutionary change in the Internet architecture using an evolutionary approach.\footnote{http://www.cs.princeton.edu/~jrex/position/disrupt.txt} She stipulated that to that end we should think about the {\em target end state} and devise a {\em backward compatible transition plan} consisting of {\em incentive compatible steps}. Since then clean slate design has been a vibrant research topic, incremental deployability has maintained a strong foothold in the networking literature as a highly sought after property in systems design, and there are ongoing research and investment efforts to render the Internet infrastructure more easily programmable. In this section, we discuss architectural evolution drawing on this blueprint's insightfulness.

%The bottleneck in evolution through incremental deployability

To that end, we are concerned with social systems free to change their state using incremental deployability as the evolutionary force. Our characterization of equilibrium as a population state such that no other state is incrementally deployable against it implies that (possibly undesirable) equilibria may emerge as {\em bottlenecks} in an evolutionary trajectory (toward a target). To drive then a population into a target equilibrium, mechanisms must, therefore, be in place to drive that population out of undesirable equilibrium states. This situation emerges, for example, in the process of changing Internet protocols from parochial incumbent versions toward emerging versions that are technically superior but their benefits do not manifest unless several adopters upgrade (but similar phenomena can manifest in situations other than the Internet). In this section, we capture these phenomena in an (archetypical in game theory) abstract mathematical model and design mechanisms that ensure equilibrium transitions from socially inferior to superior equilibria.

\subsection{On the division of responsibility between code and institutions}

Engineering is concerned with the analysis and design of {\em technologies} whereas economics is concerned with the analysis and design of {\em institutions} (such as {\em markets}). We believe there is common ground to explore between these disciplines. To that end, our characterization of (Nash) equilibrium based on the maximality of preference relations defined in terms of incremental deployability is theoretical evidence that engineering science can inform our understanding of economic theory (where equilibrium lies at its elementary foundation). Furthermore, our formulation of TCP/IP congestion control as an equilibrium computation algorithm is evidence in the other direction, namely, that economic theory can inform our understanding of the Internet architecture. 

But there is also more direct practical evidence that engineering and economics can cross-pollinate. Consider, for example, cryptocurrencies such as Bitcoin. These currencies draw their ability to carry economic (exchange) value from their technical design (that engineers scarcity using cryptography as the elementary primitive). Furthermore, the economic (financial) properties of cryptocurrencies are bound to align to engineering decisions concerning the technical design. 

In the rest of this paper, we exploit this {\em duality} between technologies and institutions to design mechanisms that benefit both engineering and economic ventures. First we discuss at length the complementary character that technologies and institutions play in the design of social systems.

\subsubsection{The capacity limits of software computation and evolution}

Theoretical computer science in its branch called {\em computability theory} studies {\em the limits of computation} using the {\em Turing machine} as its computational model. Computability theory asks which tasks admit implementations in a Turing machine, and it is well-known that there are simple tasks Turing machines cannot perform. But computation in software has other limits (beyond the current scope of computability theory). We discuss the limits that ensue as software technologies are integrated in social systems (for example, in the social system that shapes the architecture of the Internet) first, followed by a discussion of the limits of software computation as these relate to the applications of {\em computational complexity theory} in designing and building distributed systems.

% The ossification of the Internet

One of the main motivations for the research in this paper is a phenomenon that manifests at the core architecture of the Internet (that is, the TCP/IP architecture), namely, that despite significant attempts from the networking community (in the academic, industrial, and governmental sectors) to effect change, the TCP/IP architecture is defiant of these efforts. Quoting \cite{RSM}: ``In the early days of the commercial Internet (mid 1990's) there was great faith in Internet evolution. \ldots The remarkable success of the Internet surpassed our wildest imagination, but our optimism about Internet evolution proved to be unfounded. \ldots Thus the ISPs which where once thought the agents of architectural change are now seen as the cause of the Internet impasse \ldots'' The core architecture of the Internet has facilitated significant innovation at the link layer (e.g., the rapid transition to 3G, 4G, and 5G systems in wireless access and the rapid transition to optical technologies for fixed access) as well as the application layer (e.g., Google, Facebook, Twitter), however, it is increasingly being held responsible for stifling the emergence of new applications.\footnote{\url{http://www.guardian.co.uk/technology/2012/apr/29/internet-innovation-failure-patent-control}}

The Internet started out as a research experiment whose stunning success led to an expansion of unprecedented scale. The Internet grew precipitously without giving room to the community for introspection on the design choices that had been made; brilliant though as these design choices were in many respects but confined to research experimentation in others, the designers could not do more than observe a research experiment escape the lab to weave into our societal fabric.

Users and operators are calling for significant change toward improving performance, quality of service, availability, and security and the impossibility of effecting change generates frustration even in the networking research community. Quoting Jennifer Rexford\footnote{http://www.cs.princeton.edu/~jrex/position/disrupt.txt}: ``In the past several years, the networking research community has grown increasingly frustrated with the difficulty of effecting substantive change in the Internet architecture.  This frustration is responsible, in part, for a shift in focus to research in ``green-field" environments  such as overlays, peer-to-peer, sensor networks, and wireless networks.  Although these topics are compelling and important in their own right, I would argue that the underlying Internet infrastructure, IP protocols, and operational practices are plagued with serious problems that warrant  significant research attention from our community.''

Jennifer Rexford then goes on to argue in the same position statement that: ``However, having revolutionary impact on the Internet architecture is
extremely challenging because of the commercialization of the network
and the large installed base of legacy equipment running the existing
protocols.  As a result, many studies of the Internet's ``underlay"
focus on collecting and analyzing measurement data \ldots  These
kinds of studies are certainly interesting and important \ldots Yet, we
need to find ways to effect more fundamental change in the years ahead.''

The networking community's efforts toward effecting fundamental change in the Internet have been channeled toward two distinct directions. Computer scientists have become accustomed to attributing the innovation slump, on one hand, to the lack of appropriate technical characteristics in the existing technologies and attempt to counterbalance such limitations through devising technologies that are easier to implement, have better performance, and have better reliability and security, and that are incrementally deployable. However, these efforts have been to no avail. On the other hand, the community is creating large scale experimentation environments capable of even attracting a real user base such as the Global Environment for Network Innovations (GENI), an effort driven by the belief that adoption of proposals to effect change is hindered by the inevitable lack of deployment experience for untested in the field ideas.\footnote{An argument for building such large scale testbeds was articulated by \cite{Impasse}.} Deployment experience is certainly a necessary condition for widespread adoption to take place, but it is by no means sufficient.

Much like deployment experience is a necessary condition for adoption of an emerging technology to take place, {\em agreement} on the benefits of an emerging technology seems to also be a necessary condition for adoption, but it is by no means sufficient. As \cite{Olson} lucidly argues in a generally acceptable (and widely celebrated) viewpoint (that has its critics however), groups aiming to take collective action do not form through mere agreement of their members on the possible benefits of furthering their objectives or their members' likemindedness; incentives play a significant role in the formation of such groups. Olson's perspective forms the basis of this paper's thesis.

Our thesis in this paper is that to find plausible explanations for the architectural stalemate (and to effect change) one must look beyond software specifications into the incentive structure that shapes the adoption of technologies in the Internet. The networking community is not oblivious of this fact: In an influential paper, \cite{Impasse} note that ``In addition to requiring changes in routers and host software, the Internet's multiprovider nature also requires that ISPs jointly agree on any architectural change. [changing paragraph] The need for consensus is doubly damning: Not only is reaching agreement among the many providers difficult to achieve, attempting to do so also removes any competitive advantage from architectural innovation.''

The (rather standard) picture of the Internet (painted in academic texts) portrays an engineering artifact being {\em run by code} and humans merely as {\em operators} with the ability to influence the behavior of hosts and routers through configuration commands. In reality, however, there is nothing fundamental preventing the same humans from making wholesale changes to the architecture other than the fact that one individual's or organization's choices depend on the choices other individuals or organizations (from equipment manufacturers and service providers to content providers and endusers) make. The Internet is the collective outcome of choices made by these population systems, and thinking of the Internet as being run by code is, therefore, in part, misleading: The Internet is run by the {\em incentive structure} that emerges from (social) population interactions.

An explanation for the stalemate in the Internet architecture can be obtained from an assumption that incremental deployability is the primitive that drives change in the Internet environment: The Internet as a global communications platform is characterized by positive network externalities effects in that the utility of adopting a technology to any given adopter increases with the number of adopters. However, it is typically the case that if the number of adopters is small, utility is negative in that the adoption effort exceeds the benefit. Therefore, Internet technologies require a {\em critical mass} of adoption to thrive, and this critical mass creates an {\em incremental deployability barrier}. But note that such barriers are not unique to the architecture of the Internet but are rather characteristic of a number of Internet technologies. For example, the adoption of cryptocurrencies (such as Bitcoin) also faces a similar incentive structure. We stipulate these barriers can be circumvented by intervening in the adoption environment through institutional mechanisms. 

But such institutional intervention is further warranted by a related inherent uncertainty that distributed systems such as Bitcoin face in relying on {\em cryptography} to draw their value. Internet cryptography is based on computational assumptions that do not have a mathematically sound basis. The dominant algorithm for digital signatures used in the Internet is RSA (after the names of the inventors of this cryptographic algorithm) that draws its security from the computational difficulty of {\em factoring}. But that factoring is a computationally hard algorithmic process is only a conjecture. In a similar fashion, cryptocurrencies use digital signature algorithms that are not inherently safe. In my opinion, the ensuing ambiguity present in the economies based on cryptocurrencies can be managed through institutional mechanisms such as, for example, {\em insuring} monetary value, a value that could be eliminated through, a possibly catastrophic for cryptosystems (but fortuitous in other ventures), algorithmic progress toward solving {\bf NP}-hard problems.\footnote{For example, advances in neural networks, which have been successful at solving a variety of hard computational tasks, could in practice break asymmetric cryptographic systems without directly rendering a proof that {\bf P = NP}.}

\subsubsection{The capacity limits of using markets to incentivize production}

We have insofar explored the capacity limits of technological evolution advocating that software technologies be complemented by institutional mechanisms, for example, to drive technological change in social systems (such as that of the Internet) or to protect the value of financial cryptography. But there are limitations to the institutional mechanisms we have at our disposal in order to solve the problems that motivate our efforts in this paper: One of the most prominent institutional mechanisms at work in our societies is the {\em market}. Markets incentivize production in the economy through the financial gains being earned in market-based exchanges of goods and services, however, economic theory (in its branch called {\em microeconomics}) is primarily concerned with mechanisms that allocate {\em scarce resources,} namely, the output of industrial production activity. In a sense, microeconomic theory is concerned with the consumption side of market-based production whereas the production process per se is poorly understood (for example, see \citep{BMN} for an excellent theoretical treatment of production economies). Effecting architectural innovation in the Internet is in a sense an {\em information production problem} (for example, see \citep{Benkler}) in that the innovation needed corresponds to an upgrade of software and hardware algorithms. Therefore, the microeconomic perspective of understanding markets as allocating scarce resources can shed little light (in the present standing of microeconomic theory) on architectural innovation.

\subsection{The stag hunt as a model of coordination}

To a large extent, the most elementary of social activities humans engage in whenever they form groups  are {\em cooperation} (which emerges as humans put synergistic efforts) and {\em conflict} (which emerges as humans put antagonistic efforts). Of all social activities, the formation of an {\em economy} is essential in most, if not all, societies. Viewing the economy from the aforementioned perspective, breaking down economic activity as a spectrum of situations consisting of cooperation opportunities and conflict possibilities (as well as their combination) regarding the production and consumption of goods and services, would perhaps closely approximate most significant economic enterprises. 

Such a breakdown of economic activity is justifiable on the terms that cooperation {\em creates value} and conflict {\em allocates that value}. In a sense, groups form to create and share value for their members. Our goal in this paper has been to investigate methods of organization (that groups can employ to create, through production mechanisms, and share, through allocation mechanisms, wealth in a manner that advances their well-being and prosperity as they form complex societies). In the sequel, we focus on the cooperation and coordination side of organization: In starting out this paper, we posited the thesis that the Internet is an equilibrium. To a first approximation, the Internet is a {\em stag hunt} (and the problem of evolving the Internet architecture can be captured in this model). \cite{Skyrms} defines the stag hunt as follows: ``Let us suppose that the hunters [in a group] each have just the choice of hunting hare or hunting deer. The chances of getting a hare are independent of what others do. There is no chance of bagging a deer by oneself, but the chances of a successful deer hunt go up sharply with the number of hunters. A deer is much more valuable than a hare. Then we have the kind of interaction that is now generally known as the stag hunt.''

In the setting of architectural evolution, deer is the emerging technology, hare the incumbent, and the hunters are the potential adopters. That the chances of getting a hare are independent of what others do reflects the often realistic assumption that the incumbent neither benefits nor suffers from adoption of the emerging technology. That there is no chance of bagging a deer by oneself reflects that the emerging technology is not incrementally deployable, that the chances of a successful deer hunt go up sharply with the number of hunters reflects positive externalities, and that a deer is much more valuable than a hare reflects the superiority of the emerging technology. 

To provide another illustration of this game-theoretic model, let us give an example on the problem of effecting architectural innovation in the Internet infrastructure (organized as mentioned earlier in autonomous systems, administratively independent organizations offering Internet connectivity to endusers). To provide this function autonomous systems interconnect in a rather general topology (which satisfies certain statistical properties such as the power law distribution). We may represent this topology as a graph, which at the moment has tens of thousand of vertices.

In a typical adoption environment (in particular, for {\em network-layer} innovation), players are the administrative authorities of these autonomous systems---we use $i$ as the running index of a player. Each player has a choice of either adopting (strategy $A$) or defecting (strategy $D$). The payoff of an adopter depends on the connected component of adopters in the graph to which the adopter's autonomous system belongs, and this payoff goes up sharply with the size of this component. The payoff of a player who defects is zero. If the size of an adopter's component is small enough, the adopter's payoff may be negative (meaning that the deployment effort costs more than the benefit).

Denoting a combination of the players' choices (strategy profile) by $s$, writing this from the perspective of player $i$ as $(s_i, s_{-i})$ where $s_{-i}$ is a combination of strategies for all players except $i$, and denoting the payoff function of player $i$ as $u_i(\cdot)$, the previous discussion can be captured in the following simple formula:
\[
  u_i(s_i, s_{-i}) = \left\{ 
  \begin{array}{l l}
    \beta_i(s) - \gamma_i, & \text{$s_i=A$}\\
    0, & \text{$s_i=D$,}\\
  \end{array} \right.
\]
where $\beta_i(s) - \gamma_i$ is the adoption benefit (which may be negative) and $\gamma_i > 0$ is the deployment cost.

\subsubsection{A formal definition of the stag hunt}

We provide below a formal definition of the stag hunt that we will use in the analysis in following sections; to the extent of our knowledge this is the first formal definition of this game (although it closely reflects the aforementioned description of \cite{Skyrms}).

\begin{definition}[Stag hunt]
We say that $\Gamma = (I = \{1, \ldots, n\}, (S_i)_{i \in I}, (u_i)_{i \in I})$ is a {\em stag hunt} if, for all $i \in I$, $S_i = \{A, D\}$ (that is, all players have a common pair of strategies) and the payoff of choosing $A$ (hunting deer) to each player that hunts deer is an increasing function of the players that hunt deer, whereas the payoff of choosing $D$ (defecting from hunting deer and hunting hare) is a constant, say, $c$. Furthermore, letting $A^n$ denote the profile in which every player hunts deer, we have that, for all $i \in I$, $u_i(A^n) > c$. Finally, all unilateral deviations from $D^n$ (that is, from the profile where everyone defects and hunts hare) are harmful.
\end{definition}

The last requirement that unilateral deviations from $D^n$ are harmful can be understood as follows: To hunt deer a player expends resources (such as his or her energy) and unless the hunt is successful these resources are wasted. In the setting of network evolution, the resources being wasted would typically correspond to an investment that does not pay off. An example of a two-player stag hunt is the following game:
\begin{align*}
\left(\begin{array}{c c}
(10, 10) & (-1, 0)\\
(0, -1) & (0, 0)\\
\end{array}\right).
\end{align*}
If both players cooperate in choosing the superior strategy each earns a payoff of $10$ whereas the uncooperative strategy yields a payoff of $0$ irrespective of what the other player does. If a player chooses to cooperate but the other player doesn't, the former earns a negative payoff.

\subsubsection{Stag hunts are weakly acyclic games}

The stag hunt has two Nash equilibria in pure strategies, namely, universally adopting the superior cooperative strategy and universally defecting from using that strategy. The manifestation of the latter equilibrium is known as a {\em coordination failure}. The stag hung also has an inferior mixed Nash equilibrium providing that is deemed unstable and it is generally ignored in the literature.

Predicting the outcome of a stag hunt is a notoriously hard problem that can be formalized as an {\em equilibrium selection problem} \citep{HS} provided we accept the Nash equilibrium as the starting point of the analysis (a rather common assumption in game theory). A fundamental assumption in most attempts to provide a rigorous foundation for game theory since its inception is that players attempt to {\em maximize utility}. In the stag hunt there is a unique outcome where utility is maximized uniformly across all players that corresponds to the selection of the superior equilibrium. However, experience suggests that the inferior equilibrium may also manifest in practice, an observation that is supported by empirical evidence in experimental settings. 

The manifestation of the inferior equilibrium agrees with Nash equilibrium theory, however, solving the equilibrium selection problem amounts to more. Interestingly, as discussed in a related Wikipedia entry,\footnote{\url{http://en.wikipedia.org/wiki/Risk_dominance}} in their attempt to provide a solution to this latter problem, \cite{HS} identified selection of the superior equilibrium as corresponding to a ``rational decision'' but later Harsanyi retracted this position and identified selection of the inferior equilibrium as the rational choice. We note, however, that what constitutes rational choice in such environment is an unresolved problem that warrants further investigation. The mechanisms we present in this paper are designed to eschew ambiguities that arise owing to multiplicity of equilibria.

The stag hunt has three Nash equilibria, the aforementioned pure equilibria corresponding to universal defection and universal adoption as well as one Nash equilibrium in mixed strategies, which (as mentioned earlier) is generally ignored in the literature. To a large extent this adds support to an argument that mixed strategies are not well-suited to serve as an analytical framework for the class of coordination games that are of interest in this paper, which motivated our focus on alternative analytical techniques based on pure-strategy dynamics that formalize the empirical notion of incremental deployability capturing the essence of how the term is used by the networking community. It is natural then to ask if the stag hunt has maximal states other than the Nash equilibria of universal adoption and universal defection. In this section, we show that it doesn't. Figure \ref{saldkjfnalxdkhjff} demonstrates by means of an example of a stag hunt wherein a transient cycle appears, that the stag hunt does not have the FIP, implying it is not a generalized ordinal potential game. However, in the following theorem, we show that cycles in stag hunts are necessarily transient.

\begin{figure}[tb]
\centering
\includegraphics[width=5.5cm]{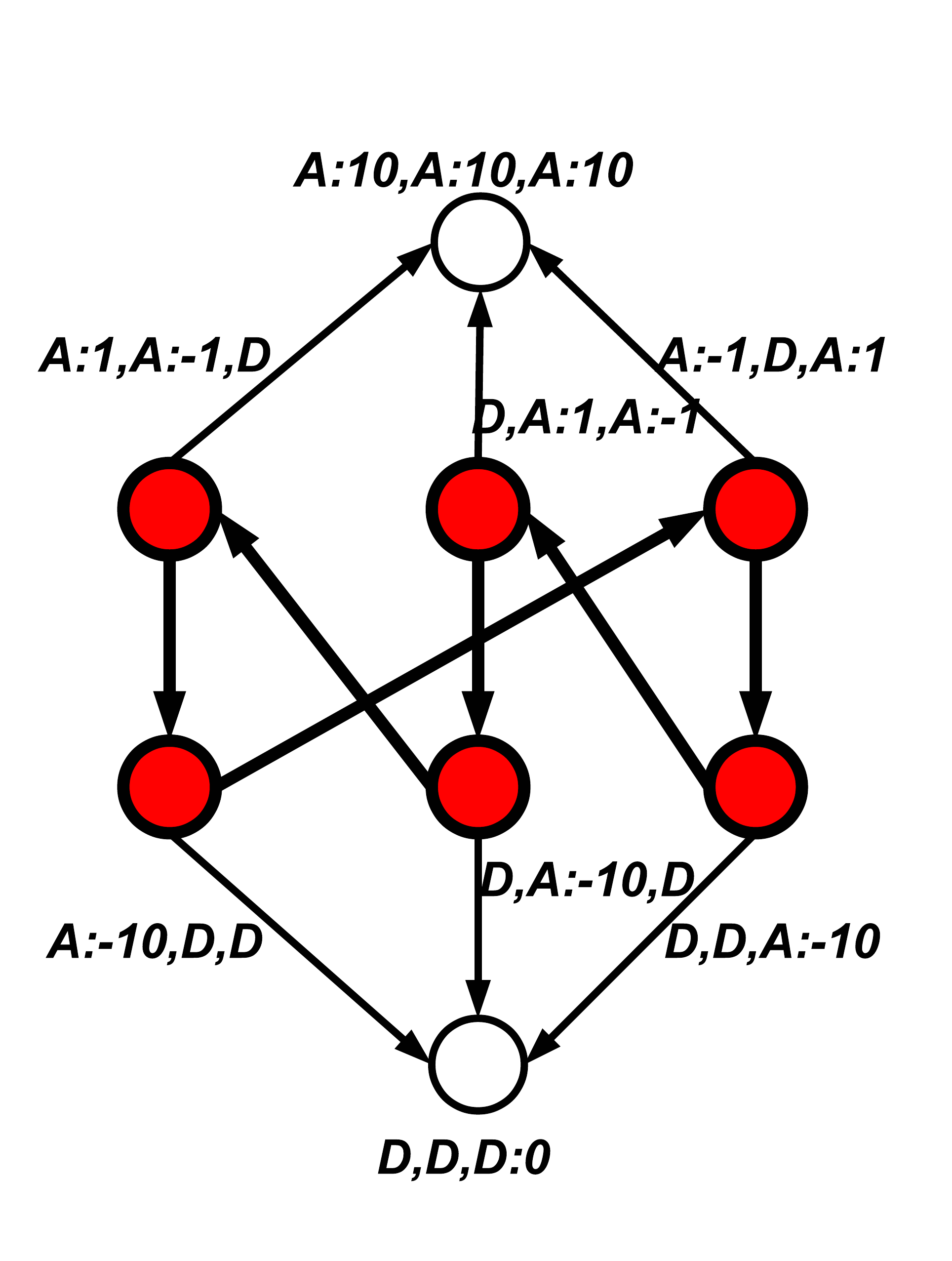}
\caption{\label{saldkjfnalxdkhjff}
Example of a stag hunt with a cycle. The figure shows the deployment graph of the game as well as the payoffs corresponding to each profile. There are three players. The payoff to a player that uses $D$ is $0$ irrespective of what other players do. The payoff of playing $A$ is shown in the figure for each profile in the game.}
\end{figure}

\begin{theorem}
\label{asldkfklfkjjasdp}
Stag hunts are weakly acyclic games.
\end{theorem}

\begin{proof}
We show that starting from any state we are ensured, by following profitable unilateral devisions, to reach either $A^n$ or $D^n$, where $n$ is the number of players. So let us start from an arbitrary state other than $A^n$ and $D^n$ and let us iteratively switch to $D$ all $A$-players that prefer $D$ to $A$. This iterative process may lead to $D^n$. If it does not, then a state is reached in which all $A$-players prefer $A$. Since the only pure Nash equilibria are $A^n$ and $D^n$, a $D$-player must exist that prefers $A$. Let us switch such a player to $A$. Then, by the definition of the stag hunt, all $A$-players in this new state prefer $A$ to $D$. Unless this new state is $A^n$, there must again exist an $D$-player that prefers $A$. Continuing in this fashion, the state $A^n$ is ensured to be reached.
\end{proof}

We note that, using our aforementioned definition of the stag hunt, the distinction between weakly and strongly maximal states vanishes as all unilateral deviations are, by definition, either (strictly) profitable or harmful, and, therefore, the stag hunt is also weakly ordinally acyclic. We believe it is also worth considering variants of the stag hunt where the aforementioned assumption is lifted (in the sense that if a player switches from the noncooperative to the cooperative strategy not all players having adopted the cooperative strategy (strictly) benefit), but we leave the corresponding analysis of such (plausible in some settings) variants as future work.

\subsection{Incrementally deployable coordination using commitments}

In the rest of this section, we are concerned with the evolution of cooperation as an incrementally deployable social outcome. Equilibrium is a universal phenomenon in nature at large and society in particular. To a large extent, systems design, whether of social or engineering systems, is concerned with {\em equilibrium transitions} that bring a system from an {\em incumbent} to an {\em emerging} equilibrium. Consider, for example, the introduction to the economy of a new digital currency. At the incipient stages corresponding to limited use in market exchanges, few to no parties find valuable to adopt such currency for their transactions, but as the number of adopters increases, and the currency's value also increases, the new currency gradually becomes an attractive medium of exchange.

To facilitate such equilibrium transitions we introduce one of the main ideas in this paper, namely, that of a ``coordination mechanism,'' first in abstract terms, which we then instantiate using particular examples. In the setting of Internet architecture, these mechanisms could be employed to overcome the deployability barriers that innovative architectures whose success depends on positive externalities face. In this setting, coordination mechanisms correspond to a radical departure from ongoing innovation efforts in Internet systems engineering. However, we should note that the application span of such mechanisms extends beyond the Internet to other endeavors facing a similar incentive structure. In fact, as argued in the sequel, the notion of a coordination mechanism is, in a sense, already being used in practical applications such as {\em crowdfunding} with great success (although viewing crowdfunding as a coordination mechanism in the sense the term is used in this paper is a novel, to the extent of our knowledge, perspective in the corresponding literature).

\subsubsection{Definition of a coordination mechanism}

We define a {\em coordination mechanism} to be a {\em social institution} that induces the superior cooperative outcome in a stag hunt (called the {\em basis game}) by means of merely adding strategies in the basis game, without affecting the (incumbent) payoff structure of that game. That is, to qualify as a coordination mechanism, a social institution must not only induce cooperation but it must also be based on the voluntary participation of the players (who may choose to opt out if they so desire). Of course players opting in have to abide by the rules of the coordination mechanism. Note that whether an institution qualifies to be a coordination mechanism depends on the analytical method that is used to predict the social outcome in the corresponding adoption environment. Our main analytical tool is the incremental deployability theory we developed but we also present the analysis of other approaches to predicting the outcome of a game (such as {\em dominance solvability}).

\subsubsection{Examples of plausible coordination mechanisms}

In the rest of this section, we provide two examples of mechanisms that seem to be plausible candidate institutions meant to induce the cooperative outcome in a stag hunt. In subsequent sections, we prove that these mechanisms qualify as coordination mechanisms.\\

\noindent
The mechanisms' efficacy rests on the following, common to both mechanisms, assumptions: 
\begin{itemize}

\item[A1.] It is possible to enforce {\em commitments}. 

\item[A2.] Using the cooperative strategy requires an {\em investment,} and a third party can measure the cost of individual investment. 

\item[A3.] A third party can verify ex post whether a particular player used the cooperative strategy.

\end{itemize}

In Internet architecture, an investment could correspond to the purchase of equipment and the implementation of technological standards, for example. Enforcing commitments is feasible across national borders through international commerce and contract law, while whether the second and third assumptions are satisfied depends on the nature of the emerging technology. Other mechanisms can also be employed toward the end of enforcing commitments however. For example, instead of relying on international contract law, a particularly compelling alternative to that end is to rely on nothing more than the institutional and organizational structure and the technological capabilities of the global Internet including the service providers, their business relationships, and the capabilities of their networks. For example, detecting infringement of commitments is possible through {\em peer review} and to punish infringement a coalition of service providers could {\em rate limit} those who infringe. The benefit of such design would be to decouple the Internet's evolution from international politics, and meeting the ensuing challenges is a fruitful direction for future work.

\subsubsection{An insurance mechanism}

The first mechanism is based on observing that adoption of the cooperative strategy in a stag hunt is a {\em risky undertaking} (as the investment involved in adopting such a strategy will not reap benefits unless others cooperate) and consists of instituting a carrier that {\em insures} these investments. There exist various possibilities on the carrier's governance structure as well as the specific insurance policies that can be offered to potential adopters. Perhaps the simplest of all insurance mechanisms is to offer policies that cover amounts slightly exceeding the corresponding potential adopters' deployment costs, and it is precisely this case that we analyze in the sequel. Observe that by offering such insurance policies, we do not modify the payoff structure of the basis game, and therefore the insurance mechanism seems quite a plausible candidate to qualify as a coordination mechanism. In our analysis in later sections, we show that this mechanism is indeed a coordination mechanism and that universal cooperation can, in principle, manifest without any of the potential adopters purchasing insurance, a rather surprising finding. We note that the idea of using insurance to induce the cooperative outcome in a stag hunt has been previously proposed by \cite{Autto}, but our analysis is more thorough and, in part, relies on a different analytical framework.

\subsubsection{An election mechanism}

Our second mechanism is based on the intuitive observation that {\em communication helps coordination,} a fact supported by ample experimental evidence (e.g.,~\cite{coord-comm}). However, in a group as large as the group of autonomous systems in the Internet, pairwise communication is untenable. But this does not preclude the possibility that indirect (more economical) communication can be as effective. For example, an interesting idea is to give players the opportunity to vote for or against adoption of the emerging technology as a means of either stimulating confidence that adoption is viable or renouncing adoption ahead of the investments. Such an idea is amenable to empirical testing, but, lacking precise models of the influence of signals (such as the outcome of the vote) on human behavior, predicting this mechanism's effect on adoption is a hard if not intractable problem. Our solution draws, however, on this idea and on the aforementioned feasibility assumptions.

We only consider in this proposal a simple election mechanism, the simplest we can think of that can be analytically shown to avert a coordination failure. According to this mechanism players have the option to vote for adoption, and once voting is complete, the mechanism outputs `1' if all players voted and `0' otherwise; the mechanism does not provision for negative votes and it does not disclose the identity of the players who voted. Such a mechanism being in place, players can condition their adoption decision on their own vote and on the outcome of the election. However, to have a predictable effect on adoption, payoffs must depend on the election outcome, and, to that effect, we introduce the following rule: Voting becomes a {\em commitment} to adopt if all players vote in which case defectors must suffer a penalty not smaller than their investment cost. The aforementioned feasibility assumptions ensure that such a mechanism can be implemented.  

We should note that, from a practical perspective, in adoption environments consisting of thousands of players the assumption of requiring that all players vote in favor of adoption to enforce a commitment is perhaps not realistic; for example, even if we assume that all players are rational, human errors, software errors (for example, in the software system that implements the election mechanism), or combinations thereof are not unlikely to manifest. However, commitments can be reasonably enforced as long as a large enough subset of players vote favorably without affecting the mechanism's efficacy: If a large enough subset of players adopt, then evolutionary forces alone inherent in the adoption environment are likely to induce themselves manifestation of the superior equilibrium of universal adoption (throughout the entire player population).

Our election mechanism is not a mere theoretical curiosity; it is akin to a {\em funding mechanism} (called {\em crowdfunding}) various platforms such as Kickstarter use to fund creative projects with great success. In this mechanism, possible contributors interested in seeing the idea of a project proposal come to reality can {\em pledge} an amount to a particular project proposal and if the proposal meets a target amount, the amounts pledged are contributed to the project; such a mechanism is evidently akin to voting in favor of an emerging technology and committing to adopt if the election outcome is positive; as noted in Kickstarter's homepage: ``All-or-nothing funding might seem scary, but it's amazingly effective in creating momentum and rallying people around an idea.'' Although it is not immediately clear whether the incentive structure of the environment crowdfunding platforms face can be modeled after a stag hunt, it is certainly a problem of {\em collective action} (a class of problems that, in a certain broad definition, encompasses stag hunts). Coordination mechanisms in these related environments warrant further independent scrutiny and we believe the analytical ideas introduced in this paper can make an important contribution to that end.

\subsubsection{Prelude to the analysis of the election mechanism}

Let us briefly discuss the approach we take in analyzing the election mechanism. Observe to that end that, as mentioned earlier, without the election mechanism universal defection is a strict equilibrium. However, once this mechanism is instituted it is not harmful for any player to vote in favor of adoption, assuming, for example, players keep in mind the reasonable strategy of defecting if the outcome of the vote is negative. Observe further that if all players think in this way and decide to vote, then, by the design of the mechanism, all players commit to adopt the cooperative strategy (implying emergence of the cooperative outcome). This rather reasonable intuitive analysis does not admit a straightforward game-theoretic justification; in fact, even if the election mechanism is instituted, universal defection remains a pure Nash equilibrium, albeit weak (and, therefore, of less predictive value than its strict counterpart). In the sequel, one of our main goals is to place squarely the aforementioned informal intuitive observations as of the election mechanism's efficacy in the game-theoretic apparatus, and many of our results are obtained in this vein.

\subsection{Analysis of the insurance and election mechanisms}
\label{Analysis}

In this section, we analyze the insurance and election mechanisms from the perspective of Nash equilibrium, dominance solvability, and maximality. In our analysis, we find that the insurance mechanism provides somewhat stronger analytical guarantees, however, in practice it may be somewhat less attractive than the election mechanism (largely due to the requirement of an insurance carrier committing, if not spending, a potentially significant budget to insure deployment efforts).

\subsubsection{Dominance solvability analysis}

Let us first introduce some elementary background on dominance solvability: A basic rationality postulate in game theory is that a player will not play a (strictly) ``dominated'' strategy. There are two notions of dominance, namely, {\em strict} (also called {\em strong}) dominance and {\em weak} dominance:

\begin{definition}[Strict dominance]
Let $\Gamma = (I, (S_i)_{i \in I}, (u_i)_{i \in I})$ be a game in strategic form, and consider player $i \in I$. We say that pure strategy $s_i$ {\em strictly dominates} the pure strategy $s'_i$ if if for all $\sigma_{-i} \in S_{-i}$, we have that
\begin{align*}
u(s_i, \sigma_{-i}) > u(s'_i, \sigma_{-i}).
\end{align*}
\end{definition} 

\begin{definition}[Weak dominance]
Let $\Gamma = (I, (S_i)_{i \in I}, (u_i)_{i \in I})$ be a game in strategic form, and consider player $i \in I$. We say that pure strategy $s_i$ {\em weakly dominates} the pure strategy $s'_i$ if if for all $\sigma_{-i} \in S_{-i}$, we have that
\begin{align*}
u(s_i, \sigma_{-i}) \geq u(s'_i, \sigma_{-i})
\end{align*}
with the inequality being strict for at least one $\sigma_{-i}$.
\end{definition} 

Strict dominance is, in a sense, compatible with the Nash equilibrium as it is possible to show that a strictly dominated strategy will never appear in a Nash equilibrium. Therefore, eliminating strictly dominated strategies from a game is a means of simplifying a game without running the risk of eliminating Nash equilibria unlike eliminating weakly dominated strategies.

The idea of eliminating dominated strategies from a game gives rise to the solution concept of {\em iterated dominance,} which appears in the literature in two variations, one {\em strict} and one {\em weak} according to the elimination criterion being used. Iterated dominance proceeds in rounds where in each round dominated strategies are eliminated until no strategies that can be eliminated remain. If the outcome of this process is a unique strategy for each player, the respective game is called {\em dominance solvable}. If a game is dominance solvable by strict iterated dominance, the respective strategy profile can be shown to be the unique (pure) Nash equilibrium of the game. The outcome of iterated weak dominance depends, unlike its strict counterpart, on the order by which weakly dominated strategies are eliminated. Despite this limitation weak iterated dominance is generally acknowledged as a credible means of drawing conclusions on a game's possible outcome.

\subsubsection{Analysis of the insurance mechanism}

The following analysis of the insurance mechanism is based on the assumption that corresponding contracts (between the carrier and potential adopters) are designed such that, in the event of a coordination failure, the insurance carrier offers a reimbursement amount to an insured player that slightly exceeds that player's deployment cost. Observe that this mechanism induces a game where each player has three strategies (instead of two), namely, to adopt without purchasing insurance (strategy $A$), to defect without purchasing insurance (strategy $D$), and to purchase insurance (strategy $X$). Recall that once a player purchases insurance she is bound by the respective contract to bear the deployment cost. Under the previous assumptions, we have the following theorem.

\begin{theorem}
\label{sid_result}
The insurance mechanism induces a game solvable by strict iterated dominance whose solution is universal adoption without purchasing insurance.
\end{theorem}

\begin{proof}
Observe that strategy $X$ strictly dominates strategy $D$ as it yields a higher payoff irrespective of the strategies of other players (by the assumption that reimbursement under coordination failure slightly exceeds the deployment cost). Once $D$ is eliminated, strategy $A$ strictly dominates strategy $X$, as strategy $X$ yields a payoff lower than that corresponding to universal adoption due to the insurance premium from the players to the insurance carrier. The theorem, therefore, follows.
\end{proof}

Observe in the previous proof that universal adoption is not a dominant strategy but rather two rounds of elimination are required to obtain universal adoption as the solution. Let us now analyze the insurance mechanism from the perspective of maximality. We have the following theorem.

\begin{theorem}
The insurance mechanism induces a game wherein the only (strongly and weakly) maximal state is universal adoption without purchasing insurance.
\end{theorem}

\begin{proof}
As mentioned earlier, in a game that is solvable by strict iterated dominance, the (necessarily) unique outcome of iterated elimination is the unique Nash equilibrium of the game. Therefore, Theorem \ref{sid_result} implies, $A^n$, where $n$ is the number of players, is the unique Nash equilibrium of the induced game. Therefore, universal adoption is a weakly maximal state. Furthermore, universal adoption is also strongly maximal as any unilateral deviation from $A^n$ is harmful. (Unilaterally deviating from $A$ to $X$ is harmful due to the insurance premium whereas unilaterally deviating from $A$ to $D$ is harmful by the definition of the stag hunt.) It remains to show that the induced game is weakly acyclic. To that end, we show that starting from any state, say $s$, $A^n$ is incrementally deployable against $s$. But this is a straightforward implication of the fact that switching from $D$ to $X$ is (strictly) profitable for the corresponding player, and, once all players have switched to $X$, it is also (strictly) profitable for each player to switch from $X$ to $A$. This completes the proof.
\end{proof}

\subsubsection{Analysis of the election mechanism}

Let us now analytically demonstrate the election mechanism's efficacy. This mechanism induces a game with four strategies, namely, to adopt without voting (strategy $A$), to defect without voting (strategy $D$), to vote and adopt if the outcome of the vote is positive defecting otherwise (strategy $X$), and to vote and adopt irrespective of the outcome of the vote (strategy $Y$). In fact, there are two additional possible strategies that involve the possibility of breaking the commitment that, by design of the mechanism, a positive vote implies; assuming a large enough penalty in the event of breaking the commitment, we may safely ignore these strategies. 

From the perspective of dominance solvability, we have the following theorem:

\begin{theorem}
The election mechanism induces a game that is dominance solvable by weak iterated dominance whose solution corresponds to the strategy profile wherein all players adopt a combination of strategies $X$, that is, to vote and adopt if the outcome of the vote is positive defecting otherwise, and $Y$, that is, to vote and adopt irrespective of the outcome of the vote. 
\end{theorem}

\begin{proof}
Note first that strategy $X$ weakly dominates strategy $D$: Since $X$ entails voting and defecting if the outcome of the vote is negative, switching from $D$ to $X$ implies that the respective player cannot lose, but also (strictly) benefits if all players choose $X$. We may, therefore, eliminate strategy $D$. The remaining strategies are $A$, $X$, and $Y$. Note now that strategy $Y$ weakly dominates strategy $A$: Selecting strategy $Y$ provides the same payoff to the respective player as $A$ except in the event that some other player has chosen strategy $X$, in which case strategy $Y$ provides a (strictly) higher payoff that $A$ since the outcome of the vote becomes negative implying the players having chosen $X$ defect, therefore, also implying a lower payoff to the adopters by the definition of the stag hunt. We may, therefore, also eliminate $A$. The strategies that remain are $X$ and $Y$.
\end{proof}

Note that if all players adopt a combination of $X$ and $Y$, the outcome implies universal adoption. Note further that by the narrow definition of dominance solvability (that a unique strategy remains for each player) the election mechanism does not qualify to be called as such, however, the previous theorem predicts manifestation of the desirable outcome in the basis stag hunt game.

Let us start the analysis of the election mechanism with respect to maximality with an example in a stag hunt with two players. The incentive structure of the induced game as well as the strongly maximal solutions are illustrated in Figure \ref{exem}. As shown in the figure, all strongly maximal states entail adoption of the cooperative strategy, however, universal defection remains a (weak) pure Nash equilibrium. This point is discussed further at the end of this section. In general stag hunt games, the efficacy of the election mechanism is argued by the following theorem.

\begin{figure}[tb]
\centering
\includegraphics[width=9cm]{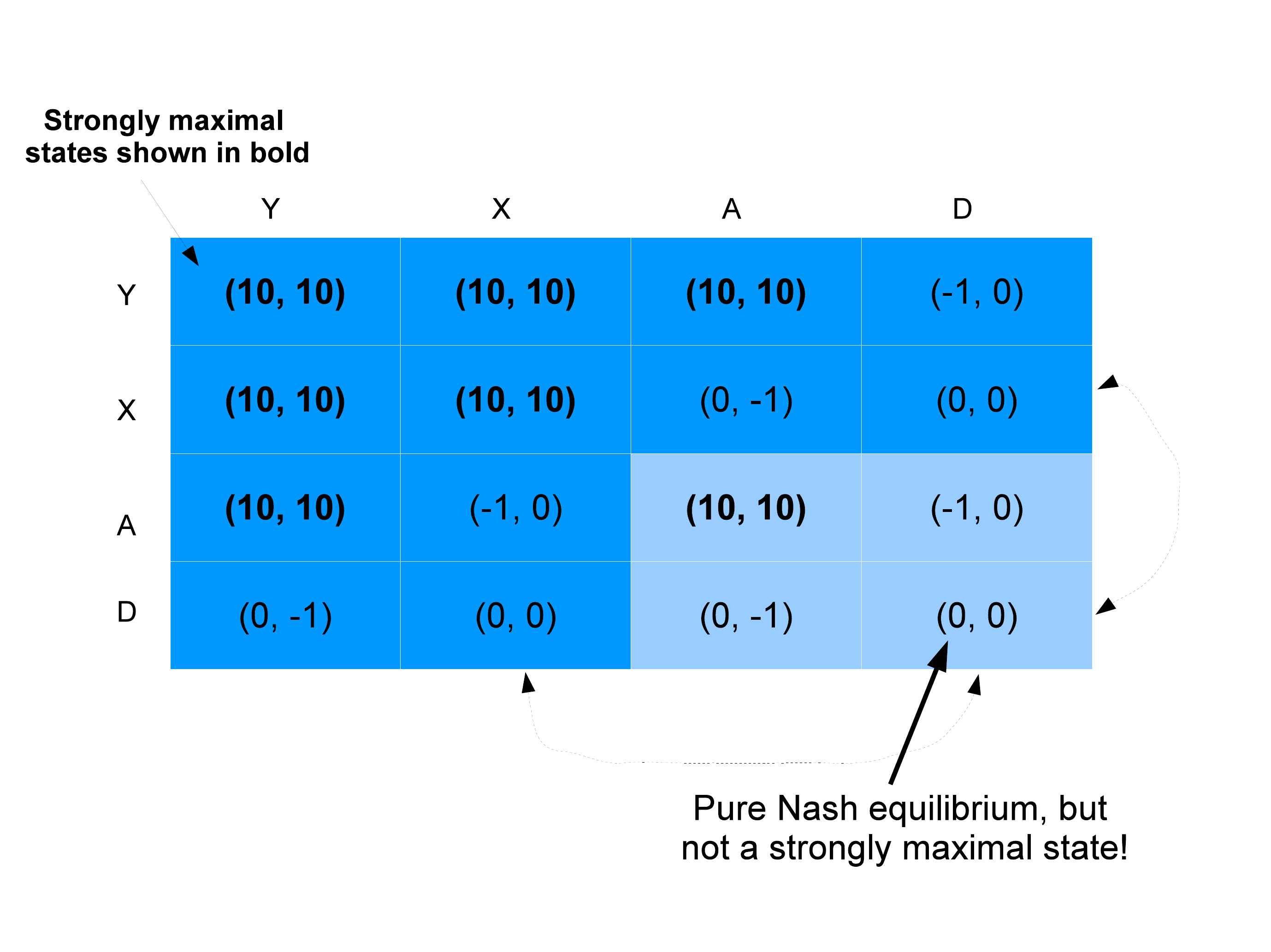}
\caption{\label{exem}
Example of incentive structure the election mechanism induces in a (two player) stag hunt.}
\end{figure}

\begin{theorem}
The election mechanism induces a game that is weakly ordinally acyclic all of whose strongly maximal equilibria imply manifestation of the superior outcome of the basis stag hunt game.
\end{theorem}

\begin{proof}
Consider the ordinal deployment graph and a state, say $s$, such that one or more players select strategy $D$. Observe that $D$-players never lose by switching from $D$ to $X$, and, therefore, switching all $D$-players in this way implies an outcome, say $s'$, that is incrementally deployable against $s$. Let us now focus attention to the players in $s'$ that select $A$. Since such players do not vote, players that select $X$ defect. However, switching $A$-players to $X$ implies that they never lose and in fact benefit once all players have been switched from $A$ to $X$. Call this latter state $s''$, and note that $s''$ is incrementally deployable against $s'$ whereas $s'$ is not so against $s''$. Since the ordinal preference relation is transitive, this implies that $s''$ is incrementally deployable against $s$ and, it is easy to observe, that $s$ is not so against $s''$. Therefore, strategy $D$ never manifests in any state that is (strongly) maximal. Note finally that $A^n$ is also a strongly maximal state: Switching players from $X^n$ to $Y^n$ implies no loss in payoff, and, therefore, that $Y^n$ is incrementally deployable against $X^n$. Furthermore, switching players from $Y^n$ to $A^n$ also implies no loss in payoff, and therefore that $A^n$ is incrementally deployable against $Y^n$ and by transitive also against $X^n$. Thus, $A^n$, combinations of $A$ and $Y$, and combinations of $X$ and $Y$ are strongly maximal equilibria belonging to the same equilibrium class. Finally observe that as noted above no state that includes $D$ is maximal.
\end{proof}

We note that universal defection remains a pure Nash equilibrium in the game that the election mechanism induces and, therefore, universal defection is also a weakly maximal state. We also note, however, that universal defection is a {\em weak} Nash equilibrium after the election mechanism is introduced in contrast to being a {\em strict} Nash equilibrium in the basis game (the stag hunt on which the election mechanism is applied), and naturally the predictive value of weak Nash equilibria is (in general) weaker than those of their strict counterparts. In the particular setting we consider, the ``evolutionary path'' from universal defection to universal adoption is one where no player loses during the transition while all players benefit after the transition has taken place. 

We finally note that although weak iterated dominance predicts that all players will vote in favor of adoption (implying that the mechanism will commit them to do so), the analysis based on strong maximality is inconclusive on whether players will vote; the outcome of universal adoption without participating in the voting procedure is as plausible an outcome as being committed to adoption by a universal vote. However, both solution concepts (namely, weak iterated dominance and strong maximality) agree on the final outcome that universal adoption will manifest (whether by means of the commitment mechanism or not). Empirical observation of the outcome in experimental settings is, therefore, of natural interest, and an interesting direction for future work.

\section{Other related work}
\label{other_related_work}

Let us, finally, discuss the related literature on overcoming incremental deployability barriers.

\subsection{On software methods for overcoming deployability barriers}

How to effect architectural innovation has been a long-standing question in the networking research community (for example, see \cite{Impasse, Feldmann} and references therein). In contrast to our approach of effecting innovation through institutional mechanisms, computer scientists have instead attempted to facilitate adoption of innovative technologies through {\em (i)} increasing the adoption benefits under ubiquitous deployment or {\em (ii)} lowering (to the extent possible) the barrier required to render the respective innovative technologies incrementally deployable. 

Both types of efforts improve deployability in different ways: Under the first strategy, equilibrium selection theories (e.g., \citep{HS}) predict an increased likelihood of adoption whereas under the second strategy the early stage adoption rate may become sustainable to drive growth. However, under both approaches, we cannot make accurate predictions on the outcome of the adoption process even in theory as the incentive structure remains a stag hunt and, as mentioned earlier, there is a lack of precise game-theoretic tools to predict which equilibrium is selected.

For example, \cite{Gill} show that the deployment of Secure-BGP can be facilitated by a pair of mechanisms, namely, {\em (i)} routing policies that prefer Internet paths safeguarded by Secure-BGP in partial deployment and {\em (ii)} offloading cryptography from {\em autonomous systems} without a customer (known as {\em stub} autonomous systems) to their providers. These techniques improve the incremental deployability of Secure-BGP, however, the incentive structure these techniques induce remains akin to the stag hunt, and, in fact, Gill {\em et al.} propose concentrating peer pressure and regulatory efforts to a small fraction of core autonomous systems for their techniques to be effective in driving the growth of Secure-BGP. Our point, therefore, that, although such line of effort is certainly helpful to the adoption of emerging technologies it is not in general conclusive, remains.

A question that has sparked interest (for example, see \citep{FII, Ghodsi, XIA}) is to design a technical architecture that supports innovation as a first order desideratum and to concentrate efforts on deploying that architecture first. The question of transitioning the Internet to such an evolvable architecture remains unresolved. Furthermore, in my opinion, a purely software-technology-based supporting foundation (as the aforementioned approaches advocate) would limit both the functionality and efficiency of architectural innovation. 

Consider, for example, TCP/IP congestion control that has managed to remain an important and stable element of the Internet without any auxiliary mechanisms, such as traffic policing, to sustain its stability, as early research had thought necessary, if not mandatory, to avoid collapse \citep{Floyd}. It would not be an overstatement to say that TCP/IP congestion control is the largest scale self-policing institution our societies have experienced to this date, and it could well serve as an example of what is achievable once we factor institutions in architectural design.

\subsection{On the diffusion of innovations in social systems}

The problem of diffusing innovation in social systems has been the subject of extensive scrutiny primarily by social scientists. The study of diffusion of technology falls squarely in the ballpark of social science: Quoting \cite{Technology}, ``The people who have thought hardest about the general questions of technology have mostly been social scientists and philosophers, and understandably they have tended to view technology from the outside as stand-alone objects. \ldots \hspace{0.2 mm} Seeing technology this way, from the outside, works well enough if we want to know how technologies enter the economy and spread within it.'' But computer scientists have also been involved with the diffusion of innovation from an algorithmic perspective. The term {\em innovation} assumes various interpretations in the literature from agricultural practices, social norms (such as which hand to extent in a handshake), medical drugs, commercial products (such as fax machines and cellphones), to networking technologies such as secure versions of BGP and quality-of-service capabilities in IP networks.

In a variety of diffusion models, potential adopters are assumed to interact according to network structures broadly referred to as social networks (whether networks based on kinship or friendship or online social networks). \cite{Diffusion} tracks the history of research in this area.

Various researchers, especially computer scientists, have recently been concerned with an algorithmic problem related to the diffusion of innovation typically referred to as {\em influence maximization} \citep{Domingos,Richardson}. Influence maximization refers to selecting an initial set of adopters in an adoption environment characterized by network structure such that the corresponding innovation will spread by forces of evolution to an as large fraction of the population as possible (possibly the entire population). (The diffusion model typically assumed in this line of work is based on a notion of ``value'' the innovative technology or behavior yields that, for each potential adopter, is dependent on the state of other potential adopters.) Kempe {\em et al.} \cite{Spread} as well as Goldberg and Liu \cite{Goldberg2} pose this problem as a combinatorial optimization problem, shown to be {\bf NP}-hard, and devise approximation algorithms in obtaining efficient solutions (that have more recently been improved in followup studies). This line of research leaves open, however, the issue of how initial adopters are enticed to adopt.

In contrast to these earlier works, our model of evolution based on the stag hunt, although motivated by the diffusion of Internet-based technologies, does not make specific assumptions on the structure of the adoption environment other than adopters benefitting from the decisions of others. In spite of its generality, we show that coordination mechanisms incite widespread diffusion while at the same time escaping the fundamental question of kickstarting the diffusion process, which we achieve by the mere act of instituting coordination mechanisms in these environments.

\subsection{On inciting collective action}

Related to the problem of diffusing innovation is that of inciting {\em collective action} (in the sense that the decision to act in a particular fashion diffuses in a social system) a problem whose study was initiated in the seminal work of \cite{Olson}. The study of collective action was later taken up by several authors such as, for example, in the form of {\em critical mass theories} (bearing relevance to nuclear fission explosions in physics) \citep{OMT,Markus}. \cite{Medina} motivates the study of collective action in settings bearing a political nature related to citizen oppression in rogue states, however, models of collective action range from {\em crowdsourcing} in the Internet \citep{Shirky} to enabling countermeasures against climate change (such as controlling carbon emissions). Although various authors discuss the problem of inciting collective action through incentives mechanisms, the idea of creating institutions to that effect (as we do in this paper) and proving their efficacy has not, to the extent of our knowledge, been considered before by these earlier works. (An exception is \citep{Autto} where the analysis is at an elementary stage.)

\subsection{On the evolution of cooperation}

The archetypical model to study the evolution of cooperation in society is the prisoner's dilemma. The stag hunt is the archetypical model for the study of coordination. But as \cite{Skyrms} notes the models are related: Iterated prisoner's dilemma with an infinite number of stages can assume the form of a stag hunt. The literature on the game-theoretic study of cooperation and coordination (starting with \citep{Axelrod}) has been concerned with the emergence of cooperative phenomena without a centralized entity to facilitate their manifestation. For example, \cite{SocialContract} studies how {\em signaling systems} can facilitate coordination. In this paper, we study the emergence of cooperation with the assistance of an exterior to the players entity that is able to enforce commitments but our analysis suggests that players may coordination while altogether eschewing commitments.

\section{Concluding remarks}
\label{conclusion}

The thesis of this paper has been that economic ventures (such as the design and organization of markets) and engineering ventures (such as the design and organization of the Internet as a global communications service) would benefit by a common ground set of mathematical theory and design methodologies that should over time be developed by economists and engineers coming together in joint intellectual efforts informing each other's domains. To the extent of my knowledge, coordinated efforts to bring together researchers from these communities have been in vain. 

I tried to support this thesis with mathematical results, engineering mechanisms, and their rigorous analysis using techniques from game theory and the mathematical formalization of incremental deployability. Our work has been motivated by the problem of {\em effecting fundamental innovation in the Internet architecture} but its scope extends far beyond Internet architecture to predicaments that extend from the {\em financial crisis} to {\em climate change}. The question of taking corrective action against these seemingly disparate plights that challenge present-day science and political leadership, once treated in a sufficiently abstract unifying setting, reveals striking commonalities between strategies to attack them that would remain hidden had mathematical rigor were to be left aside. 

For example, the organization supervising the Greek austerity program demands from our government to implement changes (through regulation) in the public sector (and society at large) that are not incrementally deployable and as a result unreasonably strain the Greek economy. 

The goal of our work has been to reveal these commonalities in order to formulate strategies aimed at facilitating enterprises as diverse as effecting architectural innovation, financing growth, and reverting climate change (where global coordination is necessary) in a unifying manner.

Contrary to the inertia that manifests in, for example, remedying the aforementioned social plights, the advent of the 21st century has been hallmarked by the emergence of socially beneficial peer-production technologies whose success also depends on positive externalities such as Wikipedia, YouTube, Facebook, Twitter, etc. that the expansion of the Internet has made possible, a trend that is expected to continue especially as we learn to harness their potential. 

The adoption and proliferation of technologies falling into the pattern of externalities-driven production of goods and services can create, I believe, a significant enough global momentum to drive economic growth in the 21st century. But to reap the benefits we must devise the means by which to make the adoption and proliferation of these technologies happen. My approach to that end is through designing and deploying appropriate {\em institutions} that facilitate such enterprises.

One of the main challenges to fostering adoption of emerging technologies whose success depends on positive externalities is {\em bootstrapping,} that is, how to induce early adoption, and one of our main goals has been to investigate means by which this bootstrapping problem can be solved. The approach has been to intervene in the adoption environment through appropriate institutions whose goal is to drive equilibrium transitions in the social systems forming the corresponding adoption environments. In a broad sense, institutions are not but behavioral constraints that facilitate the solution of coordination problems in the sense that they enable population systems to leverage {\em coordination opportunities (synergies)}. I dub these efforts {\em institutional engineering}.

Institutional engineering bears close kinship, but is not to be confused with {\em mechanism design}. An example of a {\em mechanism,} as the term is typically understood in economic theory, is an {\em auction}. Given an item and a set of players/bidders not all of which value the item identically, the goal is to allocate the item to the player who values it the most. An auction is a {\em set of rules} for performing this allocation assuming the players' valuations of the item are private, and assuming that players bid strategically. In this example, Vickrey's {\em second price auction,} according to which the item is allocated to the highest bidder who pays the second highest bid, achieves the desired objective in that it elicits bids equal to the players' true valuations (assuming the bids are strategic).

Institutional engineering is also concerned with the design of mechanisms, albeit of a different nature than the typical mechanisms studied in economic theory where the objective is to elicit the players' private preferences. Instead institutional engineering (of the flavor studied in this paper) assumes, to a first approximation, that the players' {\em preferences} are common knowledge, however, their {\em choices} are limited by the incentive structure of the environment where their interaction takes place. Its objective is to design mechanisms that will intervene in this environment so as to transition the player population from one set of choices, for example, an incumbent technology, to a more favorable one, for example, a superior emerging technology. Institutional engineering and mechanism design are intended to work in tandem to both leverage coordination opportunities as well as share/allocate the benefits that the pursuit of these coordination opportunities yield.

One of this paper's contributions is to theoretically explain (using game theory as the tool) the manifestation of a global scale institution ``beating at the heart'' of the Internet architecture, namely, TCP/IP congestion control. This congestion control algorithm can be understood as an (institutional) mechanism that manages Internet bandwidth as a common pool resource. We note we are not the first to identify this {\em duality} between technologies and institutions: In arguing that {\em code is law,} \cite{Lessig1, Lessig2} has essentially done the same. Building on his perspective \cite{Schewick} studies how technological structures and architectural principles in the Internet affect the economy that is organized on top of them. What is different in our perspective as compared to these earlier works is that the basic architecture of the Internet itself is an, in part, a {\em social phenomenon,} a thesis that, to the extent of our knowledge, we are the first to advocate for.

The (standard) picture of the Internet academic texts portray depicts an engineering artifact being {\em run by code} and humans merely as {\em operators} with the ability to influence the behavior of network equipment (hosts and routers) through configuration commands. In this vein, Bob Kahn identifies the Internet as ``a set of protocols and procedures for connecting lots of different components.''\footnote{ \url{http://www.princeton.edu/main/news/archive/S39/51/52I22/index.xml?section=topstories}} Such technology-centric perspective to understanding the Internet is appositely, we believe, summarized in a quote attributed to Dave Clarke: ``We reject kings, presidents and voting. We believe in rough consensus and running code.''\footnote{\url{http://www.ietf.org/tao.html}} This quote is a motto of the Internet Engineering Task Force (IETF), the organization that supervises the publication of ``official'' documents on the Internet architecture. Dave Clarke's quote places emphasis on the technical aspects of the architecture, and deemphasizes (but arguably hardly eliminates) its social character.

To the technology-centric perspective, we juxtaposed in this paper a {\em human-centric perspective} to understanding the Internet architecture, one captured in a quote attributed to Scott Shenker: ``The Internet is an equilibrium, we just have to identify the game.'' In this latter perspective, we viewed the Internet as a society of actors who choose among competing technologies and the Internet architecture as the collective outcome of choices these agents make. Like in any society where interactions among individuals are shaped not only by technological capabilities but also by institutions that evolve to complement and be complemented by the technical infrastructure, we stipulated that the Internet can be prolifically understood in this broader social context. 

Extrapolating on the duality between technologies and institutions (viewing them as manifestations of a more general, yet to be defined, concept) presents us with the possibility of institutionally intervening in the efforts to effect innovation, and, by extension, also presents us with the following fundamental dilemma: {\em How should we divide responsibility between code and institutions in the future Internet?} Narrowing down the scope of this question to architectural innovation, we asked: {\em Should we solely rely on technological advancements to drive innovation in the Internet or should we deem appropriate the institutional intervention in those efforts, and to what degree?}

Our answer to this question involved drawing on a fundamental notion in Internet systems engineering, namely, that to receive deployment traction in the Internet an emerging technology should be incrementally deployable against the corresponding incumbent environment. The term {\em incremental deployability} is used by the systems community in a colloquial fashion, but we demonstrated that this notion can be made mathematically precise and used this mathematical formalism to derive from a more elementary foundation (using order theory) the concept of variational equilibrium (including the Nash equilibrium used in the theory of games). This formalization helped us devise mechanisms to change the Internet architecture in an incrementally deployable fashion (using institutions as the bridge between different equilibria in the Internet architecture) and gave us the mathematical tools that were necessary to mathematically analyze these mechanisms.

From a mathematical perspective, we showed that a ``globally incrementally deployable equilibrium'' (for example, a globally evolutionarily stable strategy (GESS)) can be computed efficiently in a simplex evolution space using multiplicative weights as the algorithmic basis (and analytical techniques from convex optimization theory) noting that the presence of a global incrementally deployable equilibrium corresponds to an incentive structure that generalizes a convex incentive structure rather significantly. We note that the computational tractability of global incrementally deployable states can serve as a basis for the design of the future Internet architectures. We believe there are many applications, for example, in interdomain routing (currently plagued with oscillations and vulnerability to malicious attacks), but we leave these applications as future work.

We also showed that equililibrium approximation in a simplex environment captures the entire complexity of class {\bf PPAD}. Our work suggests that {\bf PPAD} may not be a hard class.

\section*{Acknowledgments}

My account on YouTube has kept me company, through recommendations (for example, on music to listen to or on technical material to study) and through positive ``clockwork orange'' experiments (that, for example, helped challenge my thinking or make me relax), throughout the research in this paper. I am thankful to the people who participate in the management of this account. I am also thankful to The Pierces and Cat Pierce, in particular, whose music made my journey more manageable (for example, emotionally) and provided a source of ideas for the work in this paper.

\bibliographystyle{abbrvnat}
\bibliography{real}

\appendix

\newpage

\section{The complexity of symmetric equilibrium approximation}
\label{symm_equilibrium_appendix}

\subsection{Best responses in bimatrix games}

Propositions \ref{nash_equil_characterization_b_responses}, \ref{b_response_condition_A_B}, and \ref{Jurg_proof_help} are standard in the literature. We give their proofs for completeness.

\begin{definition}
Let $(A, B)$ be a bimatrix game and $\mathbb{P} \times \mathbb{Q}$ be the space of strategy profiles of $(A, B)$. Let $Q \in \mathbb{Q}$ be an arbitrary column-player strategy. We say $P^*(Q) \in \mathbb{P}$ is a best response to Q if
\begin{align*}
P^*(Q) \in \arg\max \left\{ P \cdot AQ | P \in \mathbb{P} \right\}.
\end{align*}
Best responses of the column player against row-player strategies are defined similarly.
\end{definition}

We have the following characterization of Nash equilibria in bimatrix games.

\begin{proposition}
\label{nash_equil_characterization_b_responses}
$(X^*, Y^*) \in \mathbb{P} \times \mathbb{Q}$ is a Nash equilibrium of $(A, B)$ if and only if $X^* \in \mathbb{P}$ is a best response to $Y^* \in \mathbb{Q}$ and, similarly, $Y^*$ is a best response to $X^*$.
\end{proposition}

\begin{proof}
Straightforward from the definition of a Nash equilibrium.
\end{proof}

An important elementary property of best responses is given by the (aforementioned) {\em best response condition}. In statement and proof of this condition we follow \cite{BVS_AGT}:

\begin{proposition}
\label{b_response_condition_A_B}
Given $(A, B)$ let $A$ and $B$ be $m \times n$ matrices, $\mathbb{P} \times \mathbb{Q}$ be the space of mixed strategy profiles of $(A, B)$, $X \in \mathbb{P}$, and $Y \in \mathbb{Q}$. Then $X$ is a best response to $Y$ if and only if
\begin{align*}
\forall i = \{1, \ldots, m \} : \left[ X(i) > 0 \Rightarrow (AY)_i = (AY)_{\max} \equiv \max_{i=1}^m (AY)_i \right].
\end{align*}
\end{proposition}

\begin{proof}
Letting $u \doteq (AY)_{\max}$, we have
\begin{align*}
X \cdot AY = \sum_{i=1}^m X(i) (AY)_i = \sum_{i=1}^m X(i) (u - (u - (AY)_i)) = u - \sum_{i=1}^m X(i) (u - (AY)_i).
\end{align*}
Therefore, $X \cdot AY \leq u$ since $X(i) \geq 0$ and, by the definition of $u$, $u - (AY)_i \geq 0$, for all $i = 1, \ldots, m$. Furthermore, $X \cdot AY = u$ if and only if $X(i) > 0$ implies $u = (AY)_i$ as claimed.
\end{proof}

The best response condition implies the following proposition.

\begin{proposition}
\label{Jurg_proof_help}
$(X^*, Y^*) \in NE(A, B)$ if and only if we simultaneously have:\\

(i) for all pure strategies $E_i \in \mathbb{P}$ such that $i \in \mathcal{C}(X^*)$, $E_i$ is a best response to $Y^*$ and

(ii) for all pure strategies $E_j \in \mathbb{Q}$ such that $j \in \mathcal{C}(Y^*)$, $E_j$ is a best response to $X^*$,\\

\noindent
where $\mathbb{P} \times \mathbb{Q}$ is the space of mixed strategies of $(A, B)$.
\end{proposition}

\begin{proof}
Let us assume first that $(X^*, Y^*)$ is a Nash equilibrium of $(A, B)$. Then Proposition \ref{nash_equil_characterization_b_responses} implies that $X^*$ is a best response to $Y^*$ and that $Y^*$ is a best response to $X^*$. But then Proposition \ref{b_response_condition_A_B} implies that, for all $i \in \mathcal{C}(X^*)$, the pure strategy $E_i$ is a best response to $Y^*$ and, similarly, for all $j \in \mathcal{C}(Y^*)$, the pure strategy $E_j$ is a best to response to $X^*$ as desired.

Let us now show the reverse implication. To that end, let us assume that $(X^*, Y^*)$ is an arbitrary pair of mixed strategies of $(A, B)$ such that every pure strategy of the row player supporting $X^*$ is a best response to $Y^*$ and every pure strategy of the column player supporting $Y^*$ is a best response to $X^*$. Then, for all $i \in \mathcal{C}(X^*)$,
\begin{align*}
E_i \cdot A Y^* = (AY^*)_{\max},
\end{align*}
and, therefore,
\begin{align*}
X^* \cdot A Y^* = \sum_{i \in \mathcal{C}(X^*)} X^*(i) E_i \cdot AY^* = (AY^*)_{\max},
\end{align*}
which implies $X^*$ is a best response to $Y^*$. Similarly, given $X^*$, for all $j \in \mathcal{C}(Y^*)$,
\begin{align*}
X^* \cdot B E_j = (B^T X^*)_{\max},
\end{align*}
and, therefore,
\begin{align*}
X^* \cdot B Y^* = \sum_{j \in \mathcal{C}(Y^*)} Y^*(j) X^* \cdot BE_j = (B^T X^*)_{\max},
\end{align*}
which implies $Y^*$ is a best response to $X^*$. Invoking Proposition \ref{nash_equil_characterization_b_responses} completes the proof.
\end{proof}

Proposition \ref{Jurg_proof_help} is our result but it is implicitly used by \cite{Jurg} without proof. In the proof of Theorem \ref{PPAD_theorem}, we explain how to use Proposition \ref{Jurg_proof_help} to understand \citep[Theorem 1]{Jurg}. In proving our main result, namely, Theorem \ref{PPAD_theorem}, we extend Proposition \ref{Jurg_proof_help} as Lemma \ref{approximate_best_responses_lemma} below.

\subsection{Approximate Nash equilibria and best responses}

Let us now introduce some definitions and results related to approximate equilibria, which we already came across in the previous section.

\begin{definition}
\label{epsilon_approx_symm_equil_str}
$(Y^*, Y^*)$ is an $\epsilon$-approximate Nash equilibrium of $(C, C^T)$ if
\begin{align*}
\forall X \in \mathbb{X} : (Y^* - X) \cdot CY^* \geq - \epsilon
\end{align*}
where $\mathbb{X}$ is the probability simplex of $(C, C^T)$ and $\epsilon \geq 0$.
\end{definition}

The previous definition generalizes as:

\begin{definition}
$(P^*, Q^*)$ is an $\epsilon$-approximate Nash equilibrium of $(A, B)$ if
\begin{align*}
\forall P \in \mathbb{X} &: P^* \cdot AQ^* \geq P \cdot AQ^* - \epsilon\\ 
\forall Q \in \mathbb{Y} &: P^* \cdot BQ^* \geq P^* \cdot B Q - \epsilon
\end{align*}
where $\mathbb{P}$ ($\mathbb{Q}$) is the probability simplex of $A$ ($B$) and $\epsilon \geq 0$.\\ 

We denote the set of $\epsilon$-approximate Nash equilibria of $(A, B)$ by $NE(\epsilon, A, B)$.
\end{definition}

\begin{definition}
Given a bimatrix game $(A, B)$ and a strategy $Q \in \mathbb{Q}$ of the column player, we say that the strategy $P^* \in \mathbb{P}$ of the row player is an $\epsilon$-best response against $Q$ if
\begin{align*}
\forall P \in \mathbb{P} : P^* \cdot AQ \geq P \cdot AQ - \epsilon.
\end{align*}
$\epsilon$-best responses for the column player are defined similarly.
\end{definition}

We have the following lemma:

\begin{lemma}
\label{approximate_best_responses_lemma}
Given a bimatrix game $(A, B)$ and $(P^*, Q^*) \in \mathbb{P} \times \mathbb{Q}$, let
\begin{align*}
I(P^*) &= \{ i \in \mathcal{C}(P^*) | E_i \mbox{ is an } \epsilon\mbox{-best response against } Q^* \} \mbox{ and }\\
J(Q^*) &= \{ j \in \mathcal{C}(Q^*) | E_j \mbox{ is an } \epsilon\mbox{-best response against } P^* \}.
\end{align*}
Furthermore, assume 
\begin{align*}
P^*(i) > 0 &\mbox{ if and only if } i \in I(P^*) \mbox{ and }\\
Q^*(j) > 0 &\mbox{ if and only if } j \in J(Q^*).
\end{align*} 
Then $(P^*, Q^*) \in NE(\epsilon, A, B)$.
\end{lemma}

\begin{proof}
Under the assumptions of the lemma, 
\begin{align*}
\forall i \in I(P^*) \mbox{ } \forall P \in \mathbb{P} : E_i \cdot AQ^* \geq P \cdot AQ^* - \epsilon
\end{align*}
and
\begin{align*}
\forall j \in J(Q^*) \mbox{ } \forall Q \in \mathbb{Q} : P^* \cdot BE_j \geq P^* \cdot BQ - \epsilon.
\end{align*}
We then have that, for all $P \in \mathbb{P}$,
{\allowdisplaybreaks
\begin{align*}
P^* \cdot AQ^* &= \sum_{i \in \mathcal{C}(P^*)} P^*(i) E_i \cdot AQ^*\\
  &\geq \sum_{\mathcal{C}(P^*)} P^*(i) (P \cdot AQ^* - \epsilon)\\
  &= \left(\sum_{\mathcal{C}(P^*)} P^*(i) \right) \left( P \cdot AQ^* - \epsilon \right)\\
  &= P \cdot AQ^* - \epsilon
\end{align*}
}
and, for all $Q \in \mathbb{Q}$,
{\allowdisplaybreaks
\begin{align*}
P^* \cdot BQ^* &= \sum_{j \in \mathcal{C}(Q^*)} Q^*(j) P^* \cdot BE_j\\
  &\geq \sum_{j \in \mathcal{C}(Q^*)} Q^*(j) (P^* \cdot BQ - \epsilon)\\
  &= \left( \sum_{j \in \mathcal{C}(Q^*)} Q^*(j) \right) \left( P^* \cdot BQ - \epsilon \right)\\
  &= P^* \cdot BQ - \epsilon
\end{align*}
}
implying $(P^*, Q^*)$ is an $\epsilon$-approximate equilibrium of $(A, B)$ as desired. 
\end{proof}

Let us now give an important result on approximate equilibria. To that end, we need a definition:

\begin{definition}
\label{cool}
$(X^*, Y^*)$ is an $\epsilon$-well-supported Nash equilibrium of $(A, B)$ if
\begin{align*}
&E_i \cdot A Y^* > E_k \cdot A Y^* + \epsilon \Rightarrow X^*(k) = 0 \mbox{ and }\\
&X^* \cdot B E_j > X^* \cdot B E_k + \epsilon \Rightarrow Y^*(k) = 0.
\end{align*}
\end{definition}

Definition \ref{cool} is due to \cite{Daskalakis}. We note that an $\epsilon$-well-supported Nash equilibrium of $(A, B)$ is necessarily an $\epsilon$-approximate equilibrium of $(A, B)$ but the converse is not generally true. However, given an approximate equilibrium we can obtain a well-supported equilibrium:

\begin{proposition}
\label{cgood}
Let $(A, B)$ be such that $0 \leq A, B \leq 1$. Given an $\epsilon^2/8$-approximate Nash equilibrium of $(A, B)$, where $0 \leq \epsilon \leq 1$, we can find an $\epsilon$-well-supported Nash equilibrium in polynomial time.
\end{proposition}

The previous proposition is due to \citep{CDT} motivated by a related result in \citep{Daskalakis}. $\epsilon$-well-supported equilibria admit the following characterization:

\begin{lemma}
\label{my_5_cents}
$(X^*, Y^*)$ is an $\epsilon$-well-supported Nash equilibrium of $(A, B)$ if and only if
\begin{align*}
X^*(i) > 0 &\Rightarrow E_i \cdot AY^* \geq \max_{k = 1}^m E_k \cdot A Y^* - \epsilon \mbox{ and }\\
Y^*(j) > 0 &\Rightarrow X^* \cdot BE_j \geq \max_{\ell = 1}^n X^* \cdot BE_\ell - \epsilon.
\end{align*}
\end{lemma}

\begin{proof}
The statement of the lemma is just the contrapositive of Definition \ref{cool}.
\end{proof}

We also have the following lemma:
\begin{lemma}
\label{my_2_cents}
Let $(A, B)$ be such that $0 \leq A, B \leq 1$. Then $(X^*, Y^*)$ is an $\epsilon$-well-supported equilibrium of $(A, B)$ if and only if it is an $\epsilon/c$-well-supported equilibrium of 
\begin{align*}
\left(\frac{1}{c} (A + c_a \mathbf{1} \mathbf{1}^T), \frac{1}{c} (B + c_b \mathbf{1} \mathbf{1}^T) \right), \mbox{ } c > 0.
\end{align*}
\end{lemma}

\begin{proof}
By Lemma \ref{my_5_cents}, $(X^*, Y^*)$ is $\epsilon$-well-supported equilibrium of $(A, B)$ if and only if
\begin{align*}
X^*(i) > 0 &\Rightarrow E_i \cdot AY^* \geq \max_{k = 1}^m E_k \cdot A Y^* - \epsilon \mbox{ and }\\
Y^*(j) > 0 &\Rightarrow X^* \cdot BE_j \geq \max_{\ell = 1}^n X^* \cdot BE_\ell - \epsilon.
\end{align*}
Suppose $X^*(i) > 0$. Then
\begin{align*}
E_i \cdot \left(\frac{1}{c} (A + c_a \mathbf{1} \mathbf{1}^T)\right) Y^* &= \frac{1}{c} \left( E_i \cdot A Y^* + c_a \right)\\
  &\geq \frac{1}{c} \left( \max_{k = 1}^m E_k \cdot A Y^* - \epsilon + c_a \right)\\
  &\geq \frac{1}{c} \left( \max_{k = 1}^m E_k \cdot A Y^* + c_a \right) - \epsilon/c\\
  &\geq \max_{k = 1}^m E_k \cdot \left( \frac{1}{c} (A + c_a \mathbf{1} \mathbf{1}^T) \right) Y^* - \epsilon/c.
\end{align*}
Furthermore, if $Y^*(j) > 0$, we similarly get 
\begin{align*}
X^* \cdot \left( \frac{1}{c} (B + c_b \mathbf{1} \mathbf{1}^T) \right) E_j \geq \max_{\ell = 1}^n X^* \cdot \left( \frac{1}{c} (B + c_b \mathbf{1} \mathbf{1}^T) \right) E_\ell - \epsilon/c.
\end{align*}
This completes the proof.
\end{proof}

\cite{CDT} discuss the previous lemma without proof.

\subsection{Our main result on equilibrium approximation}

\begin{theorem}
\label{PPAD_theorem}
If the problem of computing an $\epsilon$-approximate symmetric equilibrium strategy of $(C_0, C_0^T)$, $C_0 > 0$ admits a fully polynomial time approximation scheme, then {\bf PPAD} $\subseteq$ {\bf P}.
\end{theorem}

\begin{proof}
\cite{Jurg} show that every bimatrix game $(A, B)$ can be ``symmetrized.'' That is, $(A, B)$ can be represented as a symmetric bimatrix game $(C, C^T)$ such that every equilibrium of $(C, C^T)$ yields a pair of equilibria of $(A, B)$ using a simple formula. The symmetrization is named after Gale, Kuhn, and Tucker \citep{GKT} by \cite{Jurg} as it is based on a symmetrization for zero-sum games that they extended to symmetric bimatrix games. The {\em GKT symmetrization method} of a bimatrix game $(A, B)$ gives a symmetric bimatrix game $(C, C^T)$ in which the number of pure strategies is $a + b + 1$, where $a$ is the number of row-player strategies of $(A, B)$ and $b$ is the number of column-player strategies, and where the payoff matrix $C$ is 
\begin{align}
C = \left[ \begin{array}{ccc}
0 & A & -\mathbf{1}_{a} \\
B^T & 0 & \mathbf{1}_{b} \\
\mathbf{1}^T_{a} & - \mathbf{1}^T_{b} & 0 \end{array} \right]\label{GKT_symmetrization}
\end{align}
where $\mathbf{1}_{a}$ is an $a \times 1$ vector of ones and $\mathbf{1}_{b}$ similarly. Their proof requires that $A > 0$ and $B < 0$. Note that Nash equilibria in bimatrix games $(A, B)$ as well as $\epsilon$-approximate Nash equilibria in this class of games are invariant under the addition of an arbitrary constant to every element of $A$ and respectively for $B$. Therefore, the requirement that $A > 0$ and $B < 0$ does not pose any loss of generality: Given any bimatrix game $(\hat{A}, \hat{B})$ we can obtain a bimatrix game $(A, B)$ such that $A > 0$ and $B < 0$ and such that $NE(A, B) = NE(\hat{A}, \hat{B})$ (in polynomial time). Given now $(X^*, Y^*) \in NE^+(C, C^T)$, we obtain a pair of equilibria $(P^*_1, Q^*_1)$ and $(P^*_2, Q^*_2)$ of $(A, B)$ by letting
\begin{align}
P^*_1 = \frac{1}{X^*_{a} \cdot \mathbf{1}_{a}} X^*_{a} \qquad Q^*_1 = \frac{1}{Y^*_{b} \cdot \mathbf{1}_{b}} Y^*_{b}\label{equilibrium_formulas}
\end{align}
where $X^*_{a} := (X^*(1), \ldots, X^*(a))^T$ and $Y^*_{b} := (Y^*(a+1), \ldots, Y^*(a+b))^T$ and
\begin{align*}
P^*_2 = \frac{1}{Y^*_{a} \cdot \mathbf{1}_{a}} Y^*_{a} \qquad Q^*_2 = \frac{1}{X^*_{b} \cdot \mathbf{1}_{b}} X^*_{b}
\end{align*}
where $Y^*_{a} := (Y^*(1), \ldots, Y^*(a))^T$ and $X^*_{b} := (X^*(a+1), \ldots, X^*(a+b))^T$. 

It follows immediately from the previous expressions that if $(X^*, Y^*)$ is a symmetric equilibrium ($Y^* = X^*$) of $(C, C^T)$, then $(P^*_1, Q^*_1) = (P^*_2, Q^*_2)$. Note that the number of pure strategies of $(C, C^T)$ obtained by the GKT symmetrization of $(A, B)$ is polynomial in $a+b$ and converting a symmetric equilibrium of $(C, C^T)$ to an equilibrium of $(A, B)$ requires evaluating \eqref{equilibrium_formulas}. 

\cite{CDT} show in their Theorem 6.2 that if the computation of an $\epsilon$-approximate Nash equilibrium in a bimatrix game $(A, B)$ admits an FPTAS, then {\bf PPAD} $\subseteq$ {\bf P}. One approach to finding an $\epsilon$-approximate equilibrium of $(A, B)$ is to first find an $\bar{\epsilon}$-approximate equilibrium of $(C, C^T)$ and to then use the GKT symmetrization method to obtain the desired $\epsilon$-approximate equilibrium of $(A, B)$. But does this method always work and how are $\bar{\epsilon}$ and $\epsilon$ related? Before answering these questions we sketch the the proof of \cite{Jurg} (which concerns exact equilibria).\\

Letting $X^*$ be a symmetric equilibrium strategy of $(C, C^T)$, as \cite{Jurg} argue in their Theorem 1, straight algebra using \eqref{GKT_symmetrization} gives
\begin{align*}
E_i \cdot C X^* = \begin{cases} E_i  \cdot A X^*_{b} - X^*(a + b + 1), &i=1, \ldots, a\\
X^*_{a} \cdot B E_{i - a} + X^*(a + b + 1), &i=a+1, \ldots, a+b\\
X^{*T}_{a} \mathbf{1}_{a} - X^{*T}_{b} \mathbf{1}_{b}, & i = a + b +1
\end{cases}
\end{align*}
and
\begin{align*}
X^* \cdot C^T E_j = \begin{cases} E_j  \cdot A X^*_{b} - X^*(a + b + 1), &j=1, \ldots, a\\
X^*_{a} \cdot B E_{j - a} + X^*(a + b + 1), &j=a+1, \ldots, a+b\\
X^{*T}_{a} \mathbf{1}_{a} - X^{*T}_{b} \mathbf{1}_{b}, & j = a + b +1.
\end{cases}
\end{align*}
Furthermore, by \citep[Lemma 1]{Jurg}, $X^*_{b} \neq 0$ and $X^*_{a} \neq 0$. To prove that \eqref{equilibrium_formulas} gives an equilibrium $(P^*, Q^*)$ of $(A, B)$, they proceed as follows: Assuming $X^*(i^*) > 0$ for $i^* \in \{1, \ldots, a\}$, strategy $E_{i^*}$ is a pure best response to $X^*$ by the assumption that $X^*$ is a symmetric equilibrium strategy of $(C, C^T)$ and, therefore, 
\begin{align*}
E_{i^*} \cdot CX^* = \max_{i=1}^{a+b+1} E_i \cdot CX^*,
\end{align*}
which implies
\begin{align*}
E_{i^*} \cdot CX^* = \max_{i=1}^{a} E_i \cdot CX^*.
\end{align*}
Thus, invoking the previous expressions and since $i^* \in \{1, \ldots, a \}$,
\begin{align*}
E_{i^*} \cdot AX^*_b + X^*(a + b + 1) = \max_{k=1}^{a} \left\{  E_k \cdot A X^*_{b} \right\} + X^*(a + b + 1)
\end{align*}
or, equivalently,
\begin{align}
E_{i^*} \cdot AX^*_b = \max_{k=1}^{a} \left\{  E_k \cdot A X^*_{b} \right\}.\label{funny_bit_1}
\end{align}
Similarly, $X^*(j^*) > 0$ for $j^* \in \{a+1, \ldots, a+b\}$, implies that
\begin{align}
X^*_{a} \cdot B E_{j^*-a} = \max_{\ell = 1}^{b} \left\{ X^*_{a} \cdot B E_{\ell} \right\}.\label{funny_bit_2}
\end{align}
\eqref{funny_bit_1} and \eqref{funny_bit_2} imply that
\begin{align*}
P^*(i) > 0 &\Rightarrow E_{i} \cdot AQ^* = \max_{k=1}^{a} \left\{  E_k \cdot A Q^* \right\} \mbox{ and }\\
Q^*(j) > 0 &\Rightarrow P^* \cdot B E_{j} = \max_{\ell = 1}^{b} \left\{ P^* \cdot B E_{\ell} \right\}
\end{align*}
where $P^*$ and $Q^*$ are as in \eqref{equilibrium_formulas}. Therefore, the pure strategies supporting $P^*$ are pure best responses against strategy $Q^*$ and that the pure strategies supporting $Q^*$ are pure best responses against strategy $P^*$. Proposition \ref{Jurg_proof_help} implies then that $(P^*, Q^*) \in NE(A, B)$.\\

Suppose now $(Y^*, Y^*)$ is a $\epsilon$-well-supported equilibrium of $(C, C^T)$ where $C$ is as in \eqref{equilibrium_formulas}. For technical reasons, whose necessity become apparent in the rest of the proof, we assume that:
\begin{align*}
0 < A \leq 1; \quad -1 \leq B < 0; \quad \epsilon < \min \left\{1/3, \mbox{ } \min_{ij} A_{ij}, \mbox{ } \min_{ij} (-B_{ij}) \right\}. 
\end{align*}
We will see that these assumptions can be easily satisfied without missing out on our objective.

Since $(Y^*, Y^*)$ is a $\epsilon$-well-supported equilibrium of $(C, C^T)$, by Lemma \ref{my_5_cents}:
\begin{align*}
Y^*_{a}(i) > 0 &\Rightarrow E_i \cdot A Y^*_{b} \geq \max_{k=1}^{a} \left\{  E_k \cdot A Y^*_{b} \right\} - \epsilon \mbox{ and }\\
Y^*_{b}(j) > 0 &\Rightarrow Y^*_{a} \cdot B E_j \geq \max_{\ell = 1}^{b} \left\{ Y^*_{a} \cdot B E_{\ell} \right\} - \epsilon.
\end{align*}
This implies that the pure strategies in $Y^*_a$ receiving positive probability mass are pure $\epsilon$- best responses with respect to $A$ against strategy $Y^*_b$ and that the pure strategies in $Y^*_b$ receiving positive probability mass are $\epsilon$-pure best responses with respect to $B$ against strategy $Y^*_a$. This proves, by Lemma \ref{approximate_best_responses_lemma}, that, assuming there are strategies in both $Y^*_a$ and $Y^*_b$ that receive positive probability mass, $(P^*, Q^*)$, where $P^*$ and $Q^*$ are as in \eqref{equilibrium_formulas}, is an $\epsilon$-approximate equilibrium of $(A, B)$.\\ 

Therefore, the problem of finding an $\epsilon$-approximate equilibrium of $(A, B)$ has been transformed by a polynomial time reduction to the problem of finding an $\epsilon$-well-supported symmetric equilibrium of $(C, C^T)$. Before completing this step of the proof it remains to show that if $(Y^*, Y^*)$ is an $\epsilon$-well-supported symmetric equilibrium of $(C, C^T)$, then $Y^*_{b} \neq 0$ and $Y^*_{a} \neq 0$. (This is the part of the proof where the previous assumptions are needed.) Following the proof of \citep[Lemma 1]{Jurg} (which is about exact Nash equilibria), we show that 
\begin{align*}
Y^*_a \neq 0 \Rightarrow Y^*_b \neq 0 \Rightarrow Y^*(a+b+1) \neq 0 \Rightarrow Y^*_a \neq 0.
\end{align*}  
Let us now proceed with this proof scheme:\\

\noindent
``$Y^*_a \neq 0 \Rightarrow Y^*_b \neq 0$'' : Assume $Y^*_a \neq 0$. Then, by the assumption that $Y^*$ is an $\epsilon$-well-supported symmetric equilibrium strategy of $(C, C^T)$, there is $i_0 \in \{1, \ldots, a\}$ such that
\begin{align*}
E_{i_0} \cdot C Y^* = E_{i_0} \cdot A Y^*_b - Y^*(a+b+1) \geq \begin{cases} \max_{i=1}^a E_i  \cdot A Y^*_{b} - Y^*(a + b + 1) - \epsilon\\
\max_{j = 1}^{b} Y^*_{a} \cdot B E_{j} + Y^*(a + b + 1) - \epsilon\\
Y^{*T}_{a} \mathbf{1}_{a} - Y^{*T}_{b} \mathbf{1}_{b} - \epsilon.
\end{cases}
\end{align*}
Suppose for the sake of contradiction that $Y^*_b = 0$. Then 
\begin{align*}
E_{i_0} \cdot C Y^* = - Y^*(a+b+1) \leq 0,
\end{align*} 
which implies that
\begin{align}
Y^{*T}_{a} \mathbf{1}_{a} - Y^{*T}_{b} \mathbf{1}_{b} \leq \epsilon \Rightarrow Y^{*T}_{a} \mathbf{1}_{a} \leq \epsilon\label{better1}
\end{align} 
and that  
\begin{align}
\max_{j = 1}^{b} Y^*_{a} \cdot B E_{j} + Y^*(a + b + 1) - \epsilon \leq 0\label{better2}
\end{align}
However, since $-1 \leq B \leq 0$, \eqref{better1} implies that
\begin{align}
\forall j \in \{1, \ldots, b\} : Y^*_{a} \cdot B E_{j} \geq - \epsilon \Rightarrow \max_{j = 1}^{b} Y^*_{a} \cdot B E_{j} \geq - \epsilon\label{better3}
\end{align}
and \eqref{better2} together with \eqref{better3} implies that
\begin{align}
Y^*(a+b+1) \leq 2\epsilon.\label{better4}
\end{align}
Consequently, the assumption that $Y^*_b = 0$ together with \eqref{better1} and \eqref{better4} imply that $\mathbf{1}^T Y^* \leq 3\epsilon$, which is a contradiction (as $\epsilon < 1/3$ by assumption). Hence $Y^*_b \neq 0$.\\

\noindent
``$Y^*_b \neq 0 \Rightarrow Y^*(a+b+1) \neq 0$'' : Assume $Y^*_b \neq 0$. Then there is a $j_0 \in \{a+1, \ldots, a+b\}$ such that
\begin{align*}
Y^* \cdot C^T E_{j_0} = Y^*_a B E_{j_0 - a} + Y^*(a+b+1) \geq \begin{cases} \max_{i=1}^{a} E_i  \cdot A Y^*_{b} - Y^*(a + b + 1) - \epsilon\\
\max_{j=1}^b Y^*_{a} \cdot B E_j + Y^*(a + b + 1) - \epsilon\\
Y^{*T}_{a} \mathbf{1}_{a} - Y^{*T}_{b} \mathbf{1}_{b} - \epsilon.
\end{cases}
\end{align*}
Suppose $Y^*(a+b+1) = Y^*_a = 0$. Then 
\begin{align}
\max_{i=1}^{a} E_i  \cdot A Y^*_{b} \leq \epsilon.\label{dirty}
\end{align}
Note that $Y^*_b$ is a probability vector. Therefore, 
\begin{align}
\min_{ij} A_{ij} \leq \max_{i=1}^{a} E_i  \cdot A Y^*_{b}.\label{too_dirty}
\end{align}
\eqref{dirty} and \eqref{too_dirty} along with the assumptions on $\epsilon$ lead to a contradiction.\\

Suppose now that $Y^*(a+b+1) = 0$ and $Y^*_a \neq 0$. Note that in this case
\begin{align*}
\left[ \begin{array}{c}
Y^*_a \\
Y^*_b \\
\end{array} \right]
\end{align*}
is a probability vector. Note further that
\begin{align*}
\max_{i=1}^{a} E_i  \cdot A Y^*_{b} - Y^*_a B E_{j_0 - a} \leq \epsilon,
\end{align*}
which we may equivalently write as
\begin{align}
\left[ \begin{array}{c}
E_{i^*} \\
E_{j_0 - a} \\
\end{array} \right] \cdot
\left[ \begin{array}{cc}
0 & A\\
-B^T & 0
 \end{array} \right]
\left[ \begin{array}{c}
Y^*_a \\
Y^*_b \\
\end{array} \right]
\leq \epsilon\label{d_bit}
\end{align}
where $i^* = \arg\max\{ E_i \cdot A Y_b^* | i \in \{1, \ldots, a \} \}$. However,
\begin{align}
\min \{ \min_{ij} A_{ij}, \min_{ij} (-B_{ij}) \} \leq \left[ \begin{array}{c}
E_{i^*} \\
E_{j_0 - a} \\
\end{array} \right] \cdot
\left[ \begin{array}{cc}
0 & A\\
-B^T & 0
 \end{array} \right]
\left[ \begin{array}{c}
Y^*_a \\
Y^*_b \\
\end{array} \right] \label{very_d_bit}.
\end{align}
\eqref{d_bit} and \eqref{very_d_bit} contradict our assumption on $\epsilon$. Consequently $Y^*(a+b+1) \neq 0$.\\

\noindent
``$Y^*(a+b+1) \neq 0 \Rightarrow Y^*_a \neq 0$'' : Assume $Y^*(a+b+1) \neq 0$. Then
\begin{align*}
E_{a+b+1} CY^* = Y^{*T}_{a} \mathbf{1}_{a} - Y^{*T}_{b} \mathbf{1}_{b} \geq \begin{cases} \max_{i=1}^{a} E_i  \cdot A Y^*_{b} - Y^*(a + b + 1) - \epsilon\\
\max_{j=1}^b Y^*_{a} \cdot B E_j + Y^*(a + b + 1) - \epsilon.
\end{cases}
\end{align*}  
Suppose for the sake of contradiction that $Y^*_a = 0$. Then 
\begin{align*}
0 \geq - Y^{*T}_{b} \mathbf{1}_{b} \geq \max_{j=1}^b Y^*_{a} \cdot B E_j + Y^*(a + b + 1) - \epsilon = Y^*(a + b + 1) - \epsilon.
\end{align*}
Hence 
\begin{align}
Y^*(a + b + 1) \leq \epsilon.\label{muchmuchbetter}
\end{align}
Consequently,
\begin{align*}
- Y^{*T}_{b} \mathbf{1}_{b} \geq \max_{i=1}^{a} E_i  \cdot A Y^*_{b} - Y^*(a + b + 1) - \epsilon \geq \max_{i=1}^{a} E_i  \cdot A Y^*_{b} - 2 \epsilon
\end{align*} 
or, equivalently,
\begin{align}
Y^{*T}_{b} \mathbf{1}_{b} + \max_{i=1}^{a} E_i  \cdot A Y^*_{b} \leq 2 \epsilon.\label{muchmuchbetter2}
\end{align}
Since $0 \leq A \leq 1$, \eqref{muchmuchbetter2} implies that 
\begin{align}
Y^{*T}_{b} \mathbf{1}_{b} \leq 2\epsilon.\label{muchmuchbetter3}
\end{align} 
Consequently, the assumption that $Y^*_a = 0$ together with \eqref{muchmuchbetter} and \eqref{muchmuchbetter3} imply that $\mathbf{1}^T Y^* \leq 3\epsilon$, which contradicts $\epsilon < 1/3$. Hence $Y^*_a \neq 0$.\\

\noindent
Therefore, overall, $Y^*_{a} \neq 0$ and $Y^*_{b} \neq 0$ as claimed.\\

Let us now give the reduction from the problem of computing an equilibrium of $(A, B)$ to the problem of computing a symmetric equilibrium of $(C_0, C_0^T)$ for some $C_0 > 0$. First we add, if necessary, a large enough constant to $A$ and a small enough (negative) constant to $B$ such that the resulting game, say $(\hat{A}, \hat{B})$, satisfies, $\hat{A} > 0$ and $\hat{B} < 0$. As argued earlier this does not affect the $\epsilon$-approximate equilibria one of which we are trying to compute. Then we multiply both $\hat{A}$ and $\hat{B}$ by
\begin{align*}
c' = \left( \max\left\{ \max_{ij} A_{ij}, \max_{ij} B_{ij} \right\} \right)^{-1}
\end{align*}
to obtain a bimatrix game $(\hat{A}', \hat{B}')$ such that $0 < \hat{A}' \leq 1$ and $-1 \leq \hat{B}' < 0$. Finally, we construct $(C, C^T)$ from $(\hat{A}', \hat{B}')$ such that $C$ is as in \eqref{GKT_symmetrization}. Our goal is to find an $c'\epsilon$-well-supported equilibrium of $(C, C^T)$ as we can then invoke \eqref{equilibrium_formulas} (by the previous arguments and Lemma \ref{my_2_cents}) to obtain the desired $\epsilon$-approximate equilibrium of $(A, B)$. 

To that end, we construct $(\hat{C}, \hat{C}^T)$ such that $\hat{C} > 0$, by adding a large enough constant to $C$, which does not affect the $\epsilon$-approximate equilibria of $(C, C^T)$, and then scale $(\hat{C}, \hat{C}^T)$ by 
\begin{align*}
c \doteq \frac{1}{\max_{ij} \hat{C}_{ij}}
\end{align*} 
to obtain $(C_0, C_0^T)$ such that $0 < C_0 \leq 1$. Furthermore, if $(X^*, Y^*)$ is an $cc'\epsilon$-well-supported equilibrium of $(C_0, C_0^T)$, Lemma \ref{my_2_cents} implies it is a $c'\epsilon$-well-supported equilibrium of $(\hat{C}, \hat{C}^T)$ and, therefore, of $(C, C^T)$.  Moreover, Proposition \ref{cgood} implies that from a $(cc'\epsilon)^2/8$-approximate equilibrium of $(C_0, C_0^T)$ we can obtain a $cc'\epsilon$-well-supported equilibrium of $(C_0, C_0^T)$ in {\bf P}-time. 

Overall from the previous argument we obtain that if the problem of computing an approximate symmetric equilibrium of a symmetric bimatrix game whose payoff matrix consists of positive elements is in {\bf P}, then the problem of computing a well-supported equilibrium of $(C, C^T)$ is in {\bf P} and the problem of computing an approximate equilibrium of $(A, B)$ is also in {\bf P}. Thus, by \cite[Theorem 6.2]{CDT}, an FPTAS for symmetric equilibria in symmetric bimatrix games wherein the payoff matrix has positive elements implies {\bf PPAD} is in {\bf P} as claimed.
\end{proof}

\section{Elementary properties of Hedge}
\label{preliminary_properties}

Note that, in general, we may choose $\alpha$ to depend on $X$---we make this explicit in some lemmas for clarity. The next lemma shows Hedge does not escape the simplex. 

\begin{lemma}
\label{MW_invariance_lemma}
We have
\begin{align*}
\forall \alpha(X) \geq 0 &: X \in \mathbb{X}(C) \Rightarrow T(X) \in \mathbb{X}(C)\\
\forall \alpha(X) \geq 0 &: X \in \mathbb{\mathring{X}}(C) \Rightarrow T(X) \in \mathbb{\mathring{X}}(C),
\end{align*}
\end{lemma}

\begin{proof}
We have that, given any $X \in \mathbb{X}(C)$,
\begin{align*}
\forall \alpha(X) \geq 0 : \sum_{i = 1}^n T_i(X) = \left(\sum_{j=1}^n X(j) \exp \left\{ \alpha(X) E_j \cdot CX \right\}\right)^{-1} \cdot \sum_{i=1}^n X(i) \exp \left\{ \alpha(X) E_i \cdot CX \right\} = 1,
\end{align*}
implying {\em positive invariance} of $\mathbb{X}(C)$. Furthermore,
\begin{align*}
\forall \alpha(X) \geq 0 : X(i) > 0 \Rightarrow T_i(X) > 0,
\end{align*}
further implying {\em positive invariance} of the relative interior $\mathbb{\mathring{X}}(C)$.
\end{proof}

A Nash equilibrium is a fixed point of the best response correspondence. In evolution under Hedge it is important to consider fixed points of such dynamic. To draw a distinction, we refer to Nash fixed points as ``equilibria'' and to multiplicative weights fixed points as ``fixed points.''

\begin{definition}
We say $X^* \in \mathbb{X}(C)$ is a {\em fixed point} of $C$ if it is a fixed point of \eqref{main_exp}. That is, if $X^* = T(X^*)$. We denote the set of fixed points of $C$ by $FX(C)$. Note that 
\begin{align*}
\forall i \in \mathcal{K}(C) : E_i \in FX(C).
\end{align*}
We denote the non-pure fixed points of $C$ by $FX^+(C)$.
\end{definition}

We have the following characterization of fixed points.

\begin{lemma}
\label{fixed_points_Hedge}
$X \in FX(C)$ if and only if, $X$ is a pure strategy or otherwise
\begin{align*}
\forall i, j \in \mathcal{C}(X) : (CX)_{i} = (CX)_{j}. 
\end{align*}
\end{lemma}

\begin{proof}
First we show sufficiency, that is, if for all $i, j \in \mathcal{C}(X)$, $(CX)_{i} = (CX)_{j}$,
then $T(X) = X$:
Some of the coordinates of $X$
are zero and some are positive. Clearly, 
the zero coordinates will not become positive after applying the map. 
Now, notice that, for all $i \in \mathcal{C}(X)$, $\exp\{\alpha(X) (CX)_i\} = \sum_{j = 1}^n X(j) \exp\{\alpha(X) (CX)_j\}$. Therefore, $T(X) = X$.

Now we show necessity, that is,
if $X \in FX(C)$, 
then for all $i$ and for all $i, j \in \mathcal{C}(X)$, $(CX)_i = (CX)_j$:
Let $\hat{X}(i) = T_i(x)$.
Because $X$ is a fixed point, $\hat{X}(i) = X(i)$. Therefore,
{\allowdisplaybreaks
\begin{align}
\hat{X}(i) &= X(i)\notag\\
\frac{X(i) \exp{\{\alpha(X) (CX)_i\}}}{\sum_{j} X(j) \exp{\{\alpha(X) (CX)_j\}}} &= X(i)\notag\\
\exp{\{\alpha(X) (CX)_i\}} &= \sum_{j} X(j) \exp{\{\alpha(X) (CX)_j\}},\notag
\end{align}}
which implies 
\[\exp{\{\alpha(X) ((CX)_i - (CX)_j)\}} = 1, X(i) > 0,\]
and, thus,
\[(CX)_i = (CX)_j, X(i) > 0.\]
This completes the proof.
\end{proof}

The previous result says that $X^* \in FX(C)$ if and only if $X^*$ is an equalizer of its carrier game. This is a necessary condition for equilibrium but it is not sufficient:

\begin{lemma}
\label{coffee}
If $X^* \in \mathbb{\mathring{X}}(C)$, $X^* \in FX(C)$ if and only if $X^* \in NE^+(C)$.
\end{lemma}

\begin{proof}
The lemma is a straightforward implication of the aforementioned {\em best response condition} and Theorem \ref{fixed_points_Hedge}: Suppose the interior $X^* \in NE^+(C)$ and recall the best response condition: $X$ is a best response to $Y$ if and only if $X(i) > 0 \Rightarrow (CY)_i = (CY)_{\max}$. But a symmetric Nash equilibrium strategy is a best response to itself, therefore, $X^*(i) > 0 \Rightarrow (CX^*)_i = (CX^*)_{\max}$. Since $X^*$ is interior, for all $i \in \mathcal{K}(C)$, $X^*(i) > 0$, therefore, 
\begin{align*}
\forall i, j \in \mathcal{C}(X^*) : (CX^*)_i = (CX^*)_j.
\end{align*}
The characterization of fixed points in Theorem \ref{fixed_points_Hedge} completes the proof.
\end{proof}

As a straightforward implication of the previous lemma we obtain that every non-equilibrium fixed point is a boundary strategy. Overall we obtain:

\begin{corollary}
\label{tea}
$NE^+(C) \subseteq FX(C)$.
\end{corollary}

\begin{proof}
Every equilibrium strategy is an equalizer of the its carrier by Lemma \ref{coffee}.
\end{proof}

The fundamental difference between fixed points and equilibria is that fixed points survive ``inversion'' of incentives. We call this principle ``the virtue of vice versa:''

\begin{lemma}
$FX(C) = FX(-C)$.
\end{lemma}

\begin{proof}
The statement clearly holds for all pure fixed points. Furthermore, by Proposition \ref{fixed_points_Hedge}, $X^* \in FX^+(C)$ if and only if
\begin{align*}
\forall X \in \mathbb{X}(C^*, C^{*T}) : (X^* - X) \cdot C^*X^* = 0
\end{align*}
where $C^*$ is $X^*$'s carrier game. This completes the proof.
\end{proof}

Let us finally prove the promised better response property:

\begin{definition}
We say that the map $T : \mathbb{X}(C) \rightarrow \mathbb{X}(C)$ is a {\em better response dynamic} if
\begin{align*}
\forall X \in \mathbb{X} : (T(X) - X) \cdot CX > 0,
\end{align*}
unless $X$ is a fixed point of $T$.
\end{definition}

\begin{lemma}
\label{br_property_lemma}
$X \in \mathbb{\mathring{X}}(C)$ implies $\hat{X} \cdot CX \equiv T(X) \cdot CX > X \cdot CX$ unless $X \in FX(C)$.
\end{lemma}

\begin{proof}
Let $\hat{X} \equiv T(X)$. If $X \in FX(C)$, then $\hat{X} = X$ (by the definition of a fixed point). Assume now $X$ is not a fixed point ($X \not\in FX(C)$), which implies $(CX)_{\max} > (CX)_{\min}$ (assuming $X$ is interior). The relative entropy between $\hat{X}$ and $X$ is positive:
\begin{align*}
RE(\hat{X}, X) = \sum_{i=1}^n \hat{X}(i) \ln \left( \frac{\hat{X}(i)}{X(i)} \right) = \sum_{i=1}^n \hat{X}(i) \ln \left( \frac{\exp\left\{ \alpha E_i \cdot CX \right\}}{\sum_{j=1}^n X(j) \exp \left\{ \alpha E_j \cdot CX \right\}} \right) > 0.
\end{align*}
Therefore,
\begin{align}
\ln \left( \sum_{j=1}^n X(j) \exp \left\{ \alpha E_j \cdot CX \right\} \right) < \sum_{i=1}^n \hat{X}(i) \ln\left( \exp\left\{ \alpha E_i \cdot CX \right\} \right) = \alpha \hat{X} \cdot CX.\label{some}
\end{align}
Furthermore, since $\ln(\cdot)$ is concave, we obtain by Jensen's inequality that
\begin{align}
\alpha X \cdot CX < \ln \left( \sum_{j=1}^n X(j) \exp \left\{ \alpha E_j \cdot CX \right\} \right).\label{awesome}
\end{align}
\eqref{some} and \eqref{awesome} imply $(\hat{X} - X) \cdot CX > 0$ as desired. 
\end{proof}

\section{Convexity lemma and implications}
\label{convexity_lemma_proofs}

\begin{proof}[Proof of Lemma \ref{convexity_lemma}]
We have
{\allowdisplaybreaks
\begin{align*}
\frac{d}{d \alpha}& RE(Y, \hat{X}) = \\
  &= \frac{d}{d \alpha} \left( \sum_{i \in \mathcal{C}(Y)} Y(i) \ln\left( \frac{Y(i)}{\hat{X}(i)} \right) \right)\\
  &= \frac{d}{d \alpha} \left( \sum_{i \in \mathcal{C}(Y)} Y(i) \ln\left( Y(i) 
  \cdot \frac{\sum_{j = 1}^n X(j) \exp\{ \alpha (CX)_j \}}{ X(i) \exp\{ \alpha (CX)_i \}} \right) \right)\\
  &= \frac{d}{d \alpha} \left( \sum_{i \in \mathcal{C}(Y)} Y(i) \ln\left( \frac{\sum_{j = 1}^n X(j) \exp\{ \alpha (CX)_j \}}{X(i)\exp\{ \alpha (CX)_i \}} \right) \right)\\
  &= \sum_{i \in \mathcal{C}(Y)} Y(i) \frac{d}{d \alpha} \left( \ln\left( \frac{\sum_{j = 1}^n X(j) \exp\{ \alpha (CX)_j \}}{X(i)\exp\{ \alpha (CX)_i \}} \right) \right).
\end{align*}}
Furthermore, using $(\cdot)'$ as alternative notation (abbreviation) for $d/d\alpha(\cdot)$,
\begin{align*}
\frac{d}{d \alpha} \left( \ln\left( \frac{\sum_{j = 1}^n X(j) \exp\{ \alpha (CX)_j \}}{X(i)\exp\{ \alpha (CX)_i \}} \right) \right) = \frac{X(i)\exp\{ \alpha (CX)_i \}}{\sum_{j = 1}^n X(j) \exp\{ \alpha (CX)_j \}} \left( \frac{\sum_{j = 1}^n X(j) \exp\{ \alpha (CX)_j \}}{X(i)\exp\{ \alpha (CX)_i \}} \right)'
\end{align*}
and
\begin{align*}
\left( \frac{\sum_{j = 1}^n X(j) \exp\{ \alpha (CX)_j \}}{X(i)\exp\{ \alpha (CX)_i \}} \right)'  &= \frac{\sum_{j = 1}^n X(j) (CX)_j \exp\{ \alpha (CX)_j \} X(i) \exp\{ \alpha (CX)_i \}}{\left( X(i) \exp\{ \alpha (CX)_i \} \right)^2} -\\
  &- \frac{X(i) (CX)_i \sum_{j = 1}^n X(j) \exp\{ \alpha (CX)_j \} \exp\{ \alpha (CX)_i \}}{\left( X(i) \exp\{ \alpha (CX)_i \} \right)^2} =\\
  &= \frac{\sum_{j = 1}^n X(j) (CX)_j \exp\{ \alpha (CX)_j \} \}}{ X(i) \exp\{ \alpha (CX)_i \} } -\\
  &- \frac{(CX)_i \sum_{j = 1}^n X(j) \exp\{ \alpha (CX)_j \} \}}{ X(i) \exp\{ \alpha (CX)_i \} }.
\end{align*}
Therefore,
\begin{align}
\frac{d}{d \alpha} RE(Y, \hat{X}) = \frac{\sum_{j = 1}^n X(j) (CX)_j \exp\{ \alpha (CX)_j \}}{\sum_{j = 1}^n X(j) \exp\{ \alpha (CX)_j \}} - Y \cdot CX.\label{valentine}
\end{align}
Furthermore,
\begin{align*}
\frac{d^2}{d \alpha^2} RE(Y, \hat{X}) &= \frac{\left( \sum_{j = 1}^n X(j) ((CX)_j)^2 \exp\{ \alpha (CX)_j \} \right) \left( \sum_{j = 1}^n X(j) \exp\{ \alpha (CX)_j \} \right)}{\left( \sum_{j = 1}^n X(j) \exp\{ \alpha (CX)_j \} \right)^2} -\\
  &- \frac{\left( \sum_{j = 1}^n X(j) ((CX)_j) \exp\{ \alpha (CX)_j \} \right)^2}{\left( \sum_{j = 1}^n X(j) \exp\{ \alpha (CX)_j \} \right)^2}.
\end{align*}
Jensen's inequality implies that
\begin{align*}
\frac{\sum_{j = 1}^n X(j) ((CX)_j)^2 \exp\{ \alpha (CX)_j \}}{\sum_{j = 1}^n X(j) \exp\{ \alpha (CX)_j \}} \geq \left( \frac{\sum_{j = 1}^n X(j) ((CX)_j) \exp\{ \alpha (CX)_j \}}{\sum_{j = 1}^n X(j) \exp\{ \alpha (CX)_j \}} \right)^2,
\end{align*}
which is equivalent to the numerator of the second derivative being nonnegative as $X$ is a probability vector. Note that the inequality is strict unless 
\begin{align*}
\forall i, j \in \mathcal{C}(X) : (CX)_i = (CX)_j.
\end{align*}
This completes the proof.
\end{proof}

\begin{proof}[Proof of Lemma \ref{diamonds}]
Let $\hat{X} = T(X)$. \eqref{valentine} implies that
\begin{align*}
\left. \frac{d}{d \alpha} RE(Y, \hat{X}) \right|_{\alpha = 0} = (X - Y) \cdot CX < 0.
\end{align*}
By the convexity $RE(Y, \hat{X})$ as a function of $\alpha$ either
\begin{align*}
\forall \alpha > 0 : RE(Y, \hat{X}) < RE(Y, X)
\end{align*}
or
\begin{align*}
\exists \bar{\alpha} > 0 : RE(Y, \hat{X}) = RE(Y, X)
\end{align*} 
implying
\begin{align*}
\forall \alpha \in (0, \bar{\alpha}) : RE(Y, \hat{X}) < RE(Y, X)
\end{align*}
as claimed.
\end{proof}

\begin{proof}[Proof of Lemma \ref{cool_hedge_1}]
Let $\hat{X} = T(X)$. If $X \in NE^+(C, C^T)$, $\forall \alpha \geq 0 : RE(Y, \hat{X}) = RE(Y, X)$ since $\forall \alpha \geq 0 : \hat{X} = X$. Otherwise:
{\allowdisplaybreaks
\begin{align}
RE(Y, \hat{X}) - RE(Y, X) &= \sum_{i \in \mathcal{C}(Y)} Y(i) \ln\left( \frac{\sum_{j = 1}^n X(j) \exp\{ \alpha (CX)_j \}}{\exp\{ \alpha (CX)_i \}} \right)\label{H1}\\
  &= \ln \left( \sum_{j = 1}^n X(j) \exp\{ \alpha (CX)_j \} \right) - \alpha Y \cdot CX\label{H2}\\
  &= \ln \left( \sum_{j = 1}^n X(j) \exp \left\{ \alpha \left( (CX)_j - Y \cdot CX \right) \right\} \right)\label{H3}\\
  &\leq \sum_{j = 1}^n X(j) \exp \left\{ \alpha \left( (CX)_j - Y \cdot CX \right) \right\} - 1.\label{H4}
\end{align}
}
\eqref{H1} follows by plugging in the definition of relative entropy and straight algebra. \eqref{H2} follows by straightforward properties of the natural logarithm. \eqref{H3} follows similarly. \eqref{H4} follows from the standard property of the natural logarithm that $\ln(y) \leq y - 1, y > 0$. If 
\begin{align*}
\forall j \in \{1, \ldots, n\} : (CX)_j = Y \cdot CX,
\end{align*}
then $X \in NE^+(C, C^T)$ and, as noted earlier, in this case
\begin{align*}
\forall \alpha \geq 0 : RE(Y, \hat{X}) = RE(Y, X). 
\end{align*}
Assume, therefore, that
\begin{align*}
Y \cdot CX = \max_{j = 1}^n (CX)_j > \min_{j=1}^n (CX)_j,
\end{align*}
where the inequality on the left follows by the {\em best response condition}. Then 
\begin{align*}
\forall \alpha \geq 0 : \exp \left\{ \alpha \left( (CX)_j - Y \cdot CX \right) \right\} \leq 1
\end{align*}
and there exists $j \in \{1, \ldots, n\}$ such that
\begin{align*}
\forall \alpha > 0 : \exp \left\{ \alpha \left( (CX)_j - Y \cdot CX \right) \right\} < 1.
\end{align*}
Therefore, 
\begin{align*}
\forall \alpha > 0 : \sum_{j = 1}^n X(j) \exp \left\{ \alpha \left( (CX)_j - Y \cdot CX \right) \right\} < \sum_{j=1}^n X(j) = 1.
\end{align*}
This completes the proof.
\end{proof}

\begin{proof}[First proof of Lemma \ref{cool_hedge_2}]
Let $\hat{X} = T(X)$.We have, by Jensen's inequality, that
{\allowdisplaybreaks
\begin{align}
RE(Y, \hat{X}) - RE(Y, X) &= \sum_{i \in \mathcal{C}(Y)} Y(i) \ln \frac{Y(i)}{\hat{X}(i)} - \sum_{i \in \mathcal{C}(Y)}  Y(i) \ln \frac{Y(i)}{X(i)}\notag\\
                       &= - \sum_{i \in \mathcal{C}(Y)} Y(i) \ln \hat{X}(i) + \sum_{i \in \mathcal{C}(Y)} Y(i) \ln X(i)\notag\\
                       &= \sum_{i \in \mathcal{C}(Y)} Y(i)  \ln \frac{X(i)}{\hat{X}(i)}\notag\\
                       &= \sum_{i \in \mathcal{C}(Y)} Y(i)  \ln \left( \frac{\sum_{j=1}^n X(j) \exp\{ \alpha E_j \cdot CX \}}{\exp\{\alpha E_i \cdot CX\}} \right)\notag\\
                       &= \ln \left( \sum_{j=1}^n X(j) \exp\{ \alpha E_j \cdot CX \} \right) - \sum_{i \in \mathcal{C}(Y)} Y(i)  \ln\left(\exp\{ \alpha E_i \cdot CX \} \right)\notag\\
                       &= \ln \left( \sum_{j=1}^n X(j) \exp\{ \alpha E_j \cdot CX \} \right) - \alpha \sum_{i \in \mathcal{C}(Y)} Y(i)  (E_i \cdot CX)\notag\\
                       &\geq \alpha  \sum_{j=1}^n X(j) (E_j \cdot CX) - \alpha \sum_{i \in \mathcal{C}(Y)} Y(i)  (E_i \cdot CX)\notag\\
                       &= \alpha (X \cdot CX - Y \cdot CX).\notag
\end{align}
}
Therefore, 
\begin{align*}
X \cdot CX - Y \cdot CX \geq 0 \Rightarrow \forall \alpha > 0: RE(Y, \hat{X}) - RE(Y, X) \geq 0.
\end{align*}
If $\max_{i \in \mathcal{C}(X)} \{ (CX)_i \} > \min_{i \in \mathcal{C}(X)} \{ (CX)_i \}$, since Jensen's inequality is strict,
\begin{align*}
X \cdot CX - Y \cdot CX \geq 0 \Rightarrow \forall \alpha > 0: RE(Y, \hat{X}) - RE(Y, X) > 0.
\end{align*}
This completes the proof.
\end{proof}

\begin{proof}[Second proof of Lemma \ref{cool_hedge_2}]
By the convexity lemma (Lemma \ref{convexity_lemma}), unless $X \in FX(C)$, $RE(Y, \hat{X})$ is strictly convex. \eqref{valentine} implies that
\begin{align*}
\left. \frac{d}{d \alpha} RE(Y, \hat{X}) \right|_{\alpha = 0} = (X - Y) \cdot CX \geq 0.
\end{align*}
Noting that $\left. RE(Y, \hat{X}) \right|_{\alpha = 0} = RE(Y, X)$ completes the proof.
\end{proof}

\end{document}